\documentclass[12pt]{article}  
\usepackage{natbib}
\usepackage{amsmath,amssymb}
\usepackage{bm}
\usepackage{graphicx}

\usepackage{amsthm}
\newtheorem{lemma}{Lemma}
\newtheorem{theorem}{Theorem}
\newtheorem{proposition}{Proposition}

\usepackage{lmodern}  
\usepackage[T1]{fontenc}

\usepackage{colortbl}
\definecolor{mygray}{gray}{.8}
\definecolor{mypink}{rgb}{0.99,.71,.95}
\definecolor{mycyan}{cmyk}{.3,0,0,0}
\definecolor{mypurple}{rgb}{0.79,.71,.95}
\definecolor{mint}{RGB}{170, 240, 209}
\definecolor{lavender}{RGB}{220, 208, 255}
\definecolor{sunshine}{RGB}{255, 230, 153}

\usepackage{multirow}
\usepackage{longtable,booktabs,array}
\usepackage{calc}

\usepackage{algorithmic}
\usepackage[ruled,linesnumbered]{algorithm2e}

\usepackage{hyperref}

\hypersetup{
  colorlinks=true,           
  linkcolor=blue,           
  citecolor=blue,           
  urlcolor=blue,            
  filecolor=blue,           
  pdfborder={0 0 0}         
}

\usepackage{parskip}
\usepackage{etoolbox}
\usepackage{float}
\floatstyle{ruled}
\newfloat{codelisting}{h}{lop}
\floatname{codelisting}{Listing}

\newcommand{\anon}{1}

\begin{document}

\def\spacingset#1{\renewcommand{\baselinestretch}%
{#1}\small\normalsize} \spacingset{1}


\if1\anon
{
  \title{\bf A Markov-switching  dynamic matrix factor model for the high-dimensional matrix  time series}
  \author{Chaofeng Yuan\thanks{
    The author gratefully acknowledges \textit{NSFC grants 12401389 and the Outstanding Youth Science Foundation of Heilongjiang University}}\hspace{.2cm}\\
    School of Mathematical Sciences,
       Heilongjiang University\\
   Sainan Xu \\
   School of Mathematical Sciences,
       Heilongjiang University\\
    Xingbing Kong \thanks{The author gratefully acknowledges \textit{NSFC grants 12431009 and 72342019}}\\
   School of Statistics \& Data Science, Southeast University \\
       Jianhua Guo \thanks{The author gratefully acknowledges \textit{NSFC grants 12431009}}\\
School of Mathematics and Statistics, Beijing Technology and Business University}
  \maketitle
} \fi

\if0\anon
{
  \bigskip
  \bigskip
  \bigskip
  \begin{center}
    {\LARGE\bf A Markov-switching  dynamic matrix factor model for the high-dimensional matrix  time series}
\end{center}
  \medskip
} \fi

\bigskip
\begin{abstract}
In this study, we propose a novel model called the Markov-switching dynamic matrix factor (Ms-DMF) model, which serves the dual purpose of structural interpretation and prediction for high-dimensional matrix time series.
When estimating the parameters of the Ms-DMF model,  an EM (expectation maximization) algorithm was used to get a quasi-maximum likelihood estimation, where all the parameters are estimated jointly.  A filtering and smoothing algorithm is used to compute the posterior expectations corresponding to the latent regimes and factors. The consistency, convergence rates, and limit distributions of the estimated parameters are established under mild conditions.  The effectiveness of this estimation method is also validated by rigorous numerical simulations.
Furthermore, we apply the Ms-DMF model to an international trade flow network. Compared to existing matrix factor models, our approach not only identifies the main import and export centers, but also recognizes the trade cycles between these centers. This provides profound insights and analytical capabilities to advance research in the field of international trade.
\end{abstract}

\noindent%
{\it Keywords:} Matrix time series, Markov-switching model, Filtering and smoothing, International trade flow network
\vfill

\newpage
\spacingset{1.8} 

\section{Introduction}\label{sec-intro}

Matrix time series refer to data sequences in which each observation is a matrix. They arise in many fields, including finance, engineering, and the social sciences. Recent research has been dedicated to the development of novel models for analyzing matrix time series. Most of the existing literature on matrix time series is based on factor modeling via Tucker decomposition; see \cite{2019Factor}, \cite{2021Statistical}, \cite{2021Projected} and \cite{2024Matrix}. This approach has now been extended to encompass high-dimensional tensor time series \citep{2021Statistically,2022Factor}, which also fall within the Tucker tensor decomposition framework in computer science \citep{2009Tensor}. \cite{2023Modelling} consider to model matrix time series based on a tensor canonical polyadic (CP)-decomposition. \cite{2023Two-way}  proposed a two-way dynamic matrix model for matrix time series.
Both of these methods can capture the complex dependencies and heterogeneity inherent within matrix time series.

In this paper, we focus on the matrix factor model via Tucker's decomposition, which was first introduced in \cite{2019Factor} as an extension of vector factor models to the case of matrices. Extensive research has been conducted on this model, both in terms of parameter estimation methods \citep{2021Projected, 2021Statistical,2024Matrix,2024Quasi,2024Adaptively}
and model extensions \citep{GARCH2024,Online2024,Quantile2024,Generalized2024,Modeling2024}.

Both of the above methods assume that factor loadings remain constant over time, which is an assumption not always realistic in practice. Recently, several methods have been proposed to detect structural shifts of the matrix factor. \cite{Time-Varying2024} proposes a time-varying
 matrix factor model to investigate a matrix factor model with loading matrices defined as unknown, smooth, time-varying functions. \cite{Online2024}
proposes the online detection of change points in the factor structure of large matrix time series.

In this paper, we generalize the matrix factor model to a Markov-switching dynamic matrix factor (Ms-DMF) model by introducing a state sequence to characterize shifts in factor structures. The proposed model is applicable to numerous cyclical analyses, such as business cycle analysis and stock/bond return studies. This extension accounts for scenarios where factor loadings or processes may synchronously shift due to market conditions (e.g., bull/bear markets), volatility fluctuations, business cycle phases,  or other state variables.


The Ms-DMF model can be viewed as an extension of the Markov-switching dynamic factor (Ms-DF) model for vector-valued observational data to matrix time series data and an extension of the Markov-switching state space (Ms-SS)  model to matrix time series data.
The general Ms-DF framework was first proposed by \cite{1995new} and later applied to business cycle analysis in studies such as \cite{1996Measuring}, \cite{1998business}, and \cite{1998econometric}. For Ms-SS models, Bayesian MCMC estimation was employed by \cite{2014Estimation}, Kalman filter methods by \cite{2012Regime}, and the EM algorithm by \cite{2022Markov}. However, these works focused primarily on estimation methods for vector-valued Ms-DF and Ms-SS frameworks, without establishing theoretical foundations for the models or their estimators.   Recent work by \cite{2024Estimation} proposed quasi-likelihood estimation for high-dimensional regime-switching factor models and derived asymptotic properties for the estimators, though it did not account for serial correlation. Separately, \cite{2025Modelling} developed a two-step estimator for Markov-switching factor models and studied its theoretical properties. Nevertheless, both approaches remain confined to vector factor models.

Compared with existing regime-switching factor models, this paper introduces two primary novelties. First, to the best of our knowledge, it presents the first work incorporating regime switching within a matrix factor model framework. Second, we also built a prediction model (Markov-switching vector autoregression model, Ms-VAR) for the factor process, and the corresponding Ms-VAR parameters can be simultaneously estimated with other parameters.
 Consequently, the proposed model achieves dual purposes: providing structural interpretation and enabling prediction for high-dimensional matrix time series data.

Extending regime-switching models like Ms-DF or Ms-SS to matrix-valued data presents novel computational and theoretical challenges. To estimate the parameters of our proposed Ms-DMF model, we adapt the EM-based method with Kalman filtering developed for the Ms-SS model \citep{2022Markov}. This adaptation achieves high computational efficiency by circumventing explicit high-dimensional matrix operations. Our theoretical foundation is based on recent work in regime-switching factor models \citep{2024Estimation}. However, unlike \cite{2024Estimation}, we address two additional challenges. First,
 the factor dynamic is considered, which makes the log-likelihood  more complex. Second, the distinct Kronecker structures within the factor loadings and autoregression matrices further complicate the estimation equations.
Consequently, rather than analyzing the estimating equations directly, we derive the convergence properties of the parameter estimates directly from the convergence results of the EM iteration equations.



The remainder of this paper is organized as follows. Section 2 introduces the model setup and addresses the identification issue. Section 3 describes the proposed Q-MLE estimation method.  Section 4 presents the technical assumptions and asymptotic results. Section 5 is devoted to numerical studies. A real data example is provided in Section 6.  Finally, Section 7 concludes the paper and discusses possible future work.The technical proofs for the theoretical results and some omitted results Sections 3 and 5 are included in the Supplementary Material.

Throughout the paper, $(p,q,n) \rightarrow \infty$ denotes $p,q$ and $n$ going to infinity jointly. For matrix $A$, let $\|A\|$ and $\|A\|_F$ denote its spectral norm and the Frobenius norm, respectively. Let $P_A=A(A^{\top}A)^{-1}A^{\top}$ denote the projection matrix and $M_A=I-P_A$.

\section{ Markov-switching Dynamic Matrix Factor Model }\label{sec2}

The proposed Markov-switching dynamic matrix factor (Ms-DMF) model comprises two fundamental components: an observation equation governing dimensionality reduction and a state equation characterizing system dynamics through Markov-switching autoregression.

Formally, consider a matrix time series $Y_t=(y_{t,ij}) \in \mathcal{R}^{p\times q}, t \in [n]$. The Ms-DMF model is defined as follows.
\begin{align}\label{model1}
Y_t = &R_{s_t}F_tC_{s_t}^\top+E_t,\nonumber \\
F_t = &B_{s_t}+\Phi_{s_t} F_{t-1}\Gamma_{s_t}^{\top}+\epsilon_t,
\end{align}
where $(s_t)_{t=1,\cdots,n}$ is a stationary and homogeneous hidden Markov chain taking values in a finite set of regimes $\{1,\cdots,M\}$ with transition probabilities
\begin{align}
  p(s_t=j\mid s_{t-1}=i,s_{t-2}=k,\cdots, Y_{t-1},\cdots)=p(s_t=j\mid s_{t-1}=i)\triangleq p_{ij}, \notag
\end{align}
for $i,j,k \in [M]$. The switching variable $s_t$ indicates the regime (or state) in which the system (\ref{model1}) operates at time $t$. The regimes represent different modes or states of the series.

In model (\ref{model1}),
$F_t \in \mathcal{R}^{k_1 \times k_2}$ is the latent factor matrix common to all elements of $Y_t$;
$E_t\in \mathcal{R}^{p \times q}$   represents the idiosyncratic matrix component (measurement errors) with $E(E_t)=0$ and $\text{Cov}[\text{vec}(E_t), \text{vec}(E_t)]=\sigma^2I_{pq}$;
$\epsilon_t \in \mathcal{R}^{k_1 \times k_2}$  corresponds to innovation terms in the autoregressive structure with     $E(\epsilon_t)=0$ and $\text{Cov}[\text{vec}(\epsilon_t),\text{vec}(\epsilon_t)]= \sigma_{\epsilon}^2I_{k_1,k_2} $;
 $R_{s_t} \in \mathcal{R}^{p \times k_1}$ and  $C_{s_t} \in  \mathcal{R}^{q \times k_2}$   constitute state-dependent row and column loading matrices;
 $B_{s_t} \in \mathcal{R}^{k_1 \times k_2}$  contains regime-specific intercept terms;
 $\Phi_{s_t} \in \mathcal{R}^{k_1 \times k_1}, \Gamma_{s_t} \in \mathcal{R}^{k_2 \times k_2}$  encode autoregressive coefficients.
 Under regime $s_t=j$, the core parameters are $\{R_j,C_j,B_j,\Phi_j,\Gamma_j\}$.  For simplicity, we assume that $E_t, s_t, \epsilon_t, t\in [n]$ are mutually independent.

In reality,  the model (\ref{model1}) is just a basic form. There are many expansions, such as the Ms-SS models \citep{2022Markov}. For example, both the parameters in $E_t$ and $\epsilon_t$ can also depend on the regime $s_t$, and the order of the autoregressive process can be greater than 1, and $E_t$ could be weakened to the heteroscedasticity assumption as \cite{2024Quasi}.

\textbf{Remark 1} If the observation equation does not depend on the regime $s_t$, and only the dynamics of the state equation switch with $s_t$, model (\ref{model1}) reduces to
\begin{align}\label{Rmodel1}
Y_t = &RF_tC^\top+E_t,\nonumber \\
F_t = &B_{s_t}+\Phi_{s_t} F_{t-1}\Gamma_{s_t}^{\top}+\epsilon_t.
\end{align}
The dependencies between observations, factors, and regimes under model (\ref{model1}) and (\ref{Rmodel1}) are depicted in Figure \ref{Fig1}. Taking the international trade flow network as an example, model (\ref{Rmodel1}) implies that the export hubs and import hubs (defined by $R$ and $C$) remain unchanged over time, while only the trade volumes from export hubs to import hubs (defined by $F_t$) exhibit different dynamic properties under different trade cycles (defined by $s_t$).

 \begin{figure}[htbp!]
  \centering
  \resizebox{\textwidth}{!}{\includegraphics{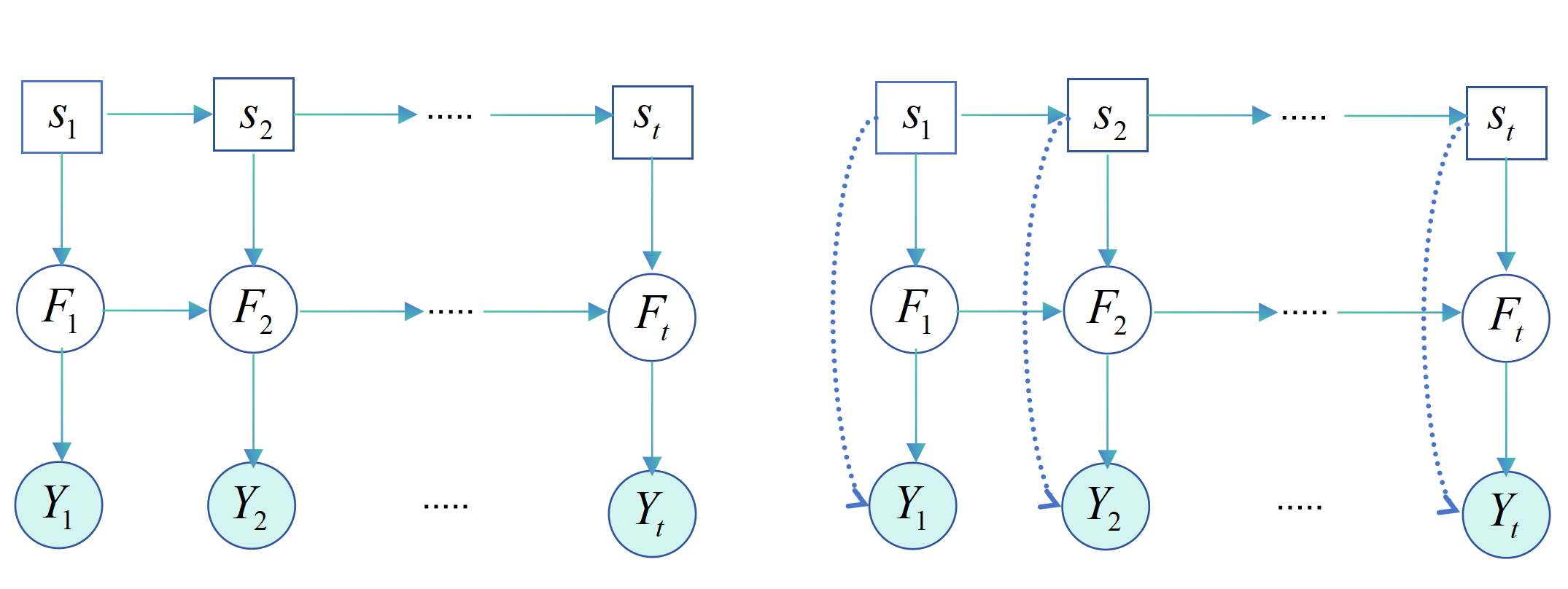}}\\
    \vspace{-15pt} 
  \caption{\small The Bayes network of model (\ref{model1}) (right) and model (\ref{Rmodel1}) (left). Square nodes represent discrete variables and oval ones are continuous variables. Shaded nodes are observed and white ones are hidden. }\label{Fig1}
\end{figure}

\subsection{  The Vector Representation}
Let $\text{vec}(\cdot)$ denote the vec operator that stacks a matrix into a vector by column.
Then, the model (\ref{model1}) can be expressed in the following
vectorized form:
\begin{align}\label{vectorizedmodel}
\boldsymbol {y_t} = & (C_{s_t}\otimes R_{s_t}) \boldsymbol {f_t}+\boldsymbol {\bm{e}_t},\nonumber \\
\boldsymbol {f_t} = &\boldsymbol {\beta_{s_t}}+ (\Gamma_{s_t} \otimes \Phi_{s_t}) \boldsymbol {f_{t-1}}+\boldsymbol {\bm{\varepsilon}_t},
\end{align}
where
$\boldsymbol {y_t}=\text{vec}(Y_t)$, $\boldsymbol {f_t}=\text{vec}(F_t)$, $\boldsymbol{e_t}=\text{vec}(E_t)$, $\boldsymbol {\beta_{s_t}}=\text{vec}(B_{s_t})$,  $\boldsymbol\varepsilon_t=\text{vec}(\epsilon_t)$. It reduces to the Ms-SS model with the exception of a special Kronecker constraint on the loading matrix and the autoregression coefficient matrix.

\textbf{Remark 2}  The processes $\{\bm{y}_t\}$ and $\{\bm{f}_t\}$ are Markovian conditional on the history of regimes $\{s_t\}$,
\begin{align}
    p(\bm{y}_t|\mathcal{Y}_{t-1},\mathcal{S}_t) = p(\bm{y}_t|\bm{y}_{t-1},\mathcal{S}_t ), ~~~~
    p(\bm{f}_t|\mathcal{F}_{t-1},\mathcal{S}_t) &= p(\bm{f}_t|\bm{f}_{t-1},\mathcal{S}_t ),
\end{align}
where $\mathcal{S}_t\triangleq \{s_0,\cdots,s_t\}$ and $\mathcal{Y}_{t-1}\triangleq \{Y_0,\cdots Y_{t-1}\}$. However, the (unconditional) marginal processes $\{\bm{y}_t\}$ and $\{\bm{f}_t\}$ are generally not Markovian. This is caused by the information in the history of the observed variable on the distribution of the regime $s_t$. Only if there is no serial correlation of regimes (mixture of normal), then the Markov property of  $\{\bm{y}_t\}$ and
$\{\bm{f}_t\}$ would be re-established.

The following proposition shows that, under moderate assumptions, $\bm{f_t}$ is a unique stationary and ergodic solution of the state equation (\ref{vectorizedmodel}), if the spectral radius of
 $Q_{\Psi}:= \text{diag} (\Psi_1 \otimes \Psi_1,\cdots,\Psi_M \otimes \Psi_M  )(P^{\top} \otimes I_{(pq)^2})$ is less than one, where $\Psi_k=\Gamma_k\otimes\Phi_k$ for $k\in[M]$.
 \begin{proposition}
      Assume that $(\bm{\epsilon}_t)_{t\in Z}$ is a sequence of i.i.d. random variables with $E(\bm{\epsilon}_0)=0, \sigma_{\epsilon}^2 <\infty$ and $\rho(Q_{\Psi}):= \lim_{k\rightarrow \infty}\| Q_{\Psi} ^{k-1}\|^{1/k} <1.$ Then there is a unique stationary and ergodic solution $\bm{f}=(\bm{f}_t)_{t\in Z}$ to equation (\ref{vectorizedmodel}).
    \end{proposition}


\subsection{ The Identification Problem}
  Maximum likelihood estimation presupposes that the model is at least locally identified. In Ms-DMF, the identification problem may be superficially caused by the factor structure in the observation equation and the interchangeability of the labels of the regime.

To see this, for any invertible $k_1\times k_1$ matrix $H_r$ and $k_2\times k_2$ matrix $H_{c}$, model (\ref{model1}) can be rewritten as
  \begin{align}
      Y_t &= R_{s_t}H_{r}H_{r}^{-1}F_tH_{c}^{-1\top}H_{c}^{\top}C_{s_t}^{\top}+E_t= \widetilde{R}_{s_t}\widetilde{F}_t\widetilde{C}_{s_t}^{\top}+E_t, \notag \\
\widetilde{F}_t &= \widetilde{B}_{s_t} + \widetilde{\Phi}_{s_t}\widetilde{F}_{t-1}\widetilde{\Gamma}_{s_t}^\top+\widetilde{\epsilon}_t ,
       \end{align}
   where $\widetilde{R}_{s_t}= R_{s_t}H_r$, $ \widetilde{C}_{s_t}= C_{s_t}H_c$, $ \widetilde{F}_t= H_{r}^{-1}F_tH_{c}^{-1\top}$, $ \widetilde{B}_{s_t} =H_{r}^{-1} B_{s_t} H_{c}^{-1\top}$, $\widetilde{\Phi}_{s_t}= H_{r}^{-1}\Phi_{s_t}H_r$, $ \widetilde{\Gamma}_{s_t}=H_c^{-1}\Gamma_{s_t}H_{c}$ and $\widetilde{\epsilon}_t= H_{r}^{-1} \epsilon_t H_{c}^{-1\top}$.  We temporarily refer to it as rotational unidentifiability. To ensure that $\widetilde{\epsilon}_t$ satisfies the model assumption that $\text{Cov}[\text{vec}(\widetilde{\epsilon}_t),\text{vec}(\widetilde{\epsilon}_t)] =\sigma_\epsilon^2 I_{k_1k_2}$,   $H_r$ and $H_c$ need to be an orthogonal matrix. Thus, $C_{k_1}^2+k_1+C_{k_2}^2+k_2$ restrictions are imposed on $H_r$ and $H_c$.
    Then there are $k_1^2+k_2^2-C_{k_1}^2-k_1-C_{k_2}^2-k_2=C_{k_1}^2+C_{k_2}^2$ free elements remaining in $H_r$ and $H_c$.
       So we assume that there exists a $k\in [M]$ such that $R_k^{\top}R_k$ and $C_k^{\top}C_k$ are diagonal matrices in (I1), where the number of restrictions is just $C_{k_1}^2+C_{k_2}^2$.


In addition to rotational unidentifiability, the model also faces unidentifiability in order and sign. Specifically, if we swap the $s$-th and $l$-th columns of $R_{s_t}$ while simultaneously swapping the $s$-th and $l$-th rows of $F_t$, the matrix $R_{s_t}F_t C_{s_t}^\top$ does not change. Similarly, swapping the $s$-th and $l$-th columns of $C_{s_t}$ and the $s$-th and $l$-th columns of $F_t$ also leaves the matrix unchanged.
To resolve these order-related issues, we assume in constraint (I1) that the diagonal elements of $R_k^{\top}R_k$ and $C_k^{\top}C_k$ are arranged in descending order. The unidentifiability of the sign arises because changing the sign of the $s$-th column of $R_{s_t}$ and the $s$-th row of $F_t$, or changing the signs of the $l$-th column of $C_{s_t}$ and the $l$-th column of $F_t$
preserves $R_{s_t}F_t C_{s_t}^\top$. This is addressed by constraint (I2).

Finally, we address the issue of regime label interchangeability in the Ms-DMF model. Following the identification strategy proposed by \cite{1997Markov} for Ms-VAR models, this nonidentifiability problem can be resolved by imposing two conditions:
(I3) distinct parameter configurations across regimes, and
(I4) an irreducible and aperiodic Markov chain associated with $P$, which guarantees a unique stationary distribution.
Since the loadings also vary across regimes, condition (I3) can be replaced by the requirement that two regimes differ if their loading spaces differ, and vice versa. A more detailed discussion is provided in \cite{2024Estimation}.

\begin{description}
    \item[(I1)]  $\exists k\in [M]$,  such that $R_k^{\top}R_k=pD_{k}^{(1)}$ and $C_k^{\top}C_k=qD_{k}^{(2)}$, where $D_{k}^{(1)}$ and $D_{k}^{(2)}$ are diagonal matrices with its elements arranged in descending order.
    \item[(I2)] The first rows of $R_k$ and $C_k$ are all non-negative for each $k\in [M]$.
    \item[(I3)] A monotonic constraint is imposed on the intercept terms of $F_{t,11}$, such that $\beta_{1,11}>\beta_{2,11}>\cdots>\beta_{M,11}$,where $\beta_{k,11}$ denotes the $(1,1)$-th
element of $B_k$.
 \item[(I4)] The transition probability matrix $P=(p_{ij})$ is irreducible and aperiodic.
\end{description}

\section{Estimation}

In this section, an EM (expectation maximization)  algorithm
will be used to obtain the quasi-maximum likelihood estimation (QMLE) of the Ms-DMF model, where all parameters
are estimated jointly. When computing the posterior expectation corresponding to
the latent regimes and factors, smoothed regime probabilities and factor scores are calculated using a recursive algorithm.

Let \(\theta _1=\{ (R_k,C_k):1\leq k\leq M,\sigma^2\}\) be the collection of parameters in the observation equation. Let \(\theta _2=\{(B_k,\Phi_k,\Gamma_k):1\leq k\leq M,\sigma_\epsilon^2\} \) be the collection of parameters in the state equation that describes the system dynamics. Let $\rho =\text{vec}(P)$. Let \(\theta= \{\theta_1,\theta_2,\rho\}\) be the collection of all parameters in model (\ref{model1}). For brevity, define \(\mathcal{Y}_t = \{Y_{t - j}\}_{j = 0}^{t - 1},~\mathcal{F}_t = \{F_{t - j}\}_{j = 0}^t\), and \(\mathcal{S}_t=\{s_{t-j}\}_{j = 0}^t\).

\subsection{The EM algorithm}\label{Sec3.1}

The complete log quasi-likelihood is given by
\begin{align*}
L^c(\theta)=&\log p( \mathcal{Y}_n,\mathcal{F}_n,\mathcal{S}_n;\theta)
=\log p(\mathcal{Y}_n\mid \mathcal{F}_n,\mathcal{S}_n;\theta)+\log p(\mathcal{F}_n\mid \mathcal{S}_n;\theta)+ \log p(\mathcal{S}_n;\theta)\\
=&\sum_{t = 1}^n\log\psi\left(y_t;(C_{s_t}\otimes R_{s_t}) f_t,\sigma^2I_{pq}\right)+\sum_{t = 1}^n\log\psi\left(f_t;\beta _{s_t}+(\Gamma_{s_t}\otimes\Phi_{s_t})f_{t-1},\sigma_\epsilon^2I_r\right)\\
&+\sum_{t=1}^n\log P_r(s_t\mid s_{t-1};\rho)+\log p(f_0\mid \mathcal{S}_n;\theta)+\log P_r(s_0;\rho),
\end{align*}
where \(\psi (x;m,v)\) is the density function of a multivariate normal distribution with mean vector $m$ and covariant matrix $v$, and $r=k_1k_2$.

Given a current estimate \( \widetilde{\theta}\) of \(\theta\), the E step consists in taking the conditional expectation of \(L^c(\theta)\) with respect to the observed data \(\mathcal{Y}_n\) while assuming that \(\widetilde{\theta}\) is the true model parameter. This produces the Q-function as
\begin{align*}
Q(\theta;\widetilde{\theta})=& \mathbb{E}[ L^c(\theta)\mid \mathcal{Y}_n;\widetilde{\theta}]\\
=&-\frac{npq}{2}\ln \sigma^2 -\frac{1}{2\sigma^2}\sum_{t = 1}^n \mathbb{E}\left[\left\|Y_t-R_{s_t}F_tC_{s_t}^{\top} \right\|_F^2\mid \mathcal{Y}_n;\widetilde{\theta}\right]\\
&-\frac{nr}{2}\ln\sigma_\epsilon^2-\frac{1}{2\sigma_\epsilon^2}\sum_{t=1}^n \mathbb{E}\left[\left\| F_t- B_{s_t}-\Phi_{s_t}F_{t-1}\Gamma_{s_t}^{\top}\right\|_F^2\mid \mathcal{Y}_n;\widetilde{\theta}\right]\\
&+\sum_{t=1}^n\sum_{i=1}^M\sum_{j=1}^M\log P_r[s_t=j\mid s_{t-1}=i;\rho]P_r[s_{t-1}=i,s_t=j\mid \mathcal{Y}_n;\widetilde{\theta}]+Cons \\
\triangleq &Q_1(\theta_1;\widetilde{\theta})+Q_2(\theta_2;\widetilde{\theta})+Q_3(\rho;\widetilde{\theta})+Cons,
\end{align*}
where $\mathit{Cons}$ denotes a constant term that does not depend on any unknown parameters. To compute the Q-function, we require the elements listed in Table 1. These elements are obtained via the filtering and smoothing methods described in Section 3.2 and Section A.5 of the Supplementary Material. More details on filtering and smoothing algorithms for Ms-SS models can be found in \cite{Kim1994}.

\begin{table}[h]
  \centering
  \caption{ The posterior expectations required for computing the Q-function.}
\label{Ta1}
  \begin{tabular}{ll}
    \toprule
    For \(Q_1(\theta_1;\widetilde{\theta})\) & For \(Q_2(\theta_2;\widetilde{\theta})\) \\
    \midrule
    \(w_{t\mid n}^{(k)}\triangleq P_r[s_t=k\mid \mathcal{Y}_n;\widetilde{\theta}]\) & \(f_{t-1\mid n}^{*(k)}\triangleq \mathbb{E}[f_{t-1}\mid s_t=k,\mathcal{Y}_n;\widetilde{\theta}]\) \\
    \(f_{t\mid n}^{(k)}\triangleq \mathbb{E}[f_t\mid s_t=k,\mathcal{Y}_n;\widetilde{\theta}]\) & \(p_{t-1\mid n}^{*(k)}\triangleq \mathbb{E}[f_{t-1}f_{t-1}^\top\mid s_t=k,\mathcal{Y}_n;\widetilde{\theta}]\) \\
    \(p_{t\mid n}^{(k)}\triangleq \mathbb{E}[f_tf_t^\prime\mid s_t=k,\mathcal{Y}_n;\widetilde{\theta}]\) &\(p_{t,t-1\mid n}^{(k)}\triangleq \mathbb{E}[f_tf_{t-1}^\top\mid s_t=k,\mathcal{Y}_n;\widetilde{\theta}]\) \\
 \midrule
For \(Q_3(\rho;\widetilde{\theta})\) &
    \(w_{t-1,t\mid n}^{(i,j)}=P_r[s_{t-1}=i,s_t=j\mid \mathcal{Y}_n;\widetilde{\theta}]\) \\
\bottomrule
\end{tabular}
\end{table}

\subsubsection{ The iterative equation of $\theta_1$}

Consider $Q_1(\theta_1;\widetilde{\theta})$.
 With some matrix operations, we have
\begin{align*}
&\left\|Y_t-R_{s_t}F_tC_{s_t}^{\top} \right\|_F^2=\text{Tr}(R_{s_t}^\top R_{s_t}F_tC_{s_t}^\top C_{st}F_t^\top)-2\text{Tr}(R_{s_t}^\top Y_tC_{s_t}F_t^\top)+\text{Tr}(Y_t^\top Y_t).
\end{align*}
Then by the law of total expectation
\begin{align*}
    Q_1(\theta_1;\widetilde{\theta})& =-\frac{npq}{2}\ln\sigma^2+Cons\\
    &-\frac{1}{2\sigma^2}\sum_{t=1}^n\sum_{k=1}^M w_{t\mid n}^{(k)}\left[\text{Tr}(R_k^\top R_kP_{t\mid n}^{c(k)})-2\text{Tr}(R_k^\top Y_tC_kF_{t\mid n}^{(k)\top})
   +\text{Tr}(Y_t^\top Y_t) \right],
\end{align*}
where \(w_{t\mid n}^{(k)}\) is defined in Table \ref{Ta1}, and
\begin{align*}
P_{t\mid n}^{c(k)}\triangleq & \mathbb{E}[F_tC_k^\top C_kF_t^\top\mid s_t=k,\mathcal{Y}_n;\widetilde{\theta}]
=\sum_{j=1}^q(e_q^{(j)\top}C_k\otimes
I_{k_1})p_{t\mid n}^{(k)}(C_k^\top e_q^{(j)}\otimes I_{k_1}),
\end{align*}
with \(e_q^{(j)}\) denoting the $j$-th column vector of \(I_q\) and \(p_{k\mid n}^{(k)}\) is defined in Table \ref{Ta1}. \(F_{t\mid n}^{(k)}\) denotes the \( k_1 \times k_2 \) matrix version of \( f_{t\mid n}^{(k) } \) given in Table \ref{Ta1}. Similarly, define
 $P_{t\mid n}^{r(k)}= \sum_{i=1}^p(I_{k_2}\otimes e_p^{(i)\top}R_k)p_{t\mid n}^{(k)}(I_{k_2}\otimes R_k^\top e_p^{(i)}).$
Maximizing \(Q_1(\theta_1;\widetilde{\theta})\) with respect to \(R_k,C_k\) and \(\sigma^2\) , We can get the iterative equations of \(\theta_1\) as follows.
\begin{align}\label{es-th1}
    \widehat{R}_k=&\left(\sum_{t=1}^nw_{t\mid n}^{(k)}Y_t\widehat{C}_kF_{t\mid n}^{(k)\top}\right)\left(\sum_{t=1}^nw_{t\mid n}^{(k)}P_{t\mid n}^{c(k)}\right)^{-1}, \notag\\
    \widehat{C}_k=&\left(\sum_{t=1}^nw_{t\mid n}^{(k)}Y_t^\top\widehat{R}_kF_{t\mid n}^{(k)}\right)\left(\sum_{t=1}^nw_{t\mid n}^{(k)}P_{t\mid n}^{r(k)}\right)^{-1},\notag\\
    \hat{\sigma}^2=&\frac{1}{npq}\sum_{t=1}^n\sum_{k=1}^Mw_{t\mid n}^{(k)}\left[\text{Tr}(R_k^\top R_k P_{t\mid n}^{c(k)})-2\text{Tr}(R_k^\top Y_t C_k F_{t\mid n}^{(k)\top})+\text{Tr}(Y_t^\top Y_t)\right].
\end{align}

\subsubsection{ The iterative equation of \(\theta_2\)}
By the law of total expectation and some matrix operations,
\begin{align}\label{es1}
Q_2(\theta_2;\widetilde{\theta})=&
-\frac{nr}{2}\ln\sigma_\epsilon^2-\frac{1}{2\sigma_\epsilon^2}\sum_{t=1}^n \mathbb{E}\left[\left\| F_t- B_{s_t}-\Phi_{s_t}F_{t-1}\Gamma_{s_t}^{\top}\right\|_F^2\mid \mathcal{Y}_n;\widetilde{\theta}\right]\notag\\
=&-\frac{nr}{2}\ln\sigma_\epsilon^2-\frac{1}{2\sigma_\epsilon^2}\sum_{t=1}^n\sum_{k=1}^M w_{t\mid n}^{(k)} \text{Tr}\left[P_{t\mid n}^{(k)}-2F_{t\mid n}^{(k)} B_k^\top-2P_{t,t-1\mid n}^{(2k)}\Phi_k^{\top} \right.\notag\\
&\left.+B_kB_k^\top+2B_k\Gamma_kF_{t-1\mid n}^{*(k)\top}\Phi_k^{\top}+\Phi_kP_{t-1\mid n}^{*(2k)}\Phi_k^{\top}\right],
\end{align}
where
\begin{align*}
  & P_{t\mid n}^{(k)}\triangleq  \mathbb{E}[F_t F_t^\top\mid s_t=k,\mathcal{Y}_n;\widetilde{\theta}]= \sum_{d=1}^{k_2}(e_{k_2}^{(d)\top}\otimes I_{k_1})p_{t|n}^{(k)}(e_{k_2}^{(d)}\otimes I_{k_1}),\\
  &  P_{t,t-1\mid n}^{(2k)}\triangleq  \mathbb{E}[F_t \Gamma_kF_{t-1}^\top\mid s_t=k,\mathcal{Y}_n;\widetilde{\theta}]=\sum_{d=1}^{k_2}(e_{k_2}^{(d)\top}\otimes I_{k_1})p_{t,t-1\mid n}^{(k)}(\Gamma_k^{\top}e_{k_2}^{(d)}\otimes I_{k_1}),\\
 &   P_{t-1\mid n}^{*(2k)}\triangleq  \mathbb{E}[F_{t-1} \Gamma_k^{\top}\Gamma_kF_{t-1}^\top\mid s_t=k,\mathcal{Y}_n;\widetilde{\theta}]=\sum_{d=1}^{k_2}(e_{k_2}^{(d)\top}\Gamma_k\otimes I_{k_1})p_{t-1\mid n}^{*(k)}(\Gamma_k^{\top}e_{k_2}^{(d)}\otimes I_{k_1}),
\end{align*}
and $F_{t-1|n}^{*(k)}$ denotes the matrix version of $f_{t-1|n}^{*(k)}$ in Table \ref{Ta1}.   Following the matrix differentiation framework, the parameter update rules for \(B_k\),\(\Phi_k\) and \(\sigma_\epsilon^2\) derived from (\ref{es1}) are formulated as follows:
\begin{align}\label{es-th2}
    \widehat{B}_k=&\frac{1}{\sum_{t=1}^nw_{t\mid n}^{(k)}}\sum_{t=1}^nw_{t\mid n}^{(k)}\left(F_{t\mid n}^{(k)}-\widetilde{\Phi}_kF_{t-1\mid n}^{*(k)}\widetilde{\Gamma}_k^{\top}\right), \notag \\
    \widehat{\Phi}_k=&\left[\sum_{t=1}^nw_{t\mid n}^{(k)}(P_{t,t-1\mid n}^{(2k)}-\widetilde{B}_k \widetilde{\Gamma}_kF_{t-1\mid n}^{*(k)\top})\right ]\left(\sum_{t=1}^nw_{t\mid n}^{(k)}P_{t-1\mid n}^{*(2k)}\right)^{-1},\notag\\
      \widehat{\sigma}_\epsilon^2=&\frac{1}{nr}\sum\limits_{t=1}^n\sum\limits_{k=1}^M w_{t\mid n}^{(k)} \text{Tr}\left[P_{t\mid n}^{(k)}-2F_{t\mid n}^{(k)} \widetilde{B}_k^\top-2P_{t,t-1\mid n}^{(2k)}\widetilde{\Phi}_k^{\top}
+\widetilde{B}_k\widetilde{B}_k^\top \right. \notag \\
&\left.   +2\widetilde{B}_k\widetilde{\Gamma}_kF_{t-1\mid n}^{*(k)\top}\widetilde{\Phi}_k^{\top}+\widetilde{\Phi}_kP_{t-1\mid n}^{*(2k)}\widetilde{\Phi}_k^{\top}\right].
\end{align}

To establish the update rule of \(\Gamma_k\), we first transform the objective function (\ref{es1}) into an equivalent form as follows.
\begin{align}\label{es2}
    Q_2(\theta_2;\widetilde{\theta})=&-\frac{nr}{2}\ln\sigma_\epsilon^2-\frac{1}{2\sigma_\epsilon^2}\sum_{t=1}^n\sum_{k=1}^M w_{t\mid n}^{(k)}\left[\text{Tr}(P_{t\mid n}^{(k)}-2F_{t\mid n}^{(k)}B_k^\top)-2\text{Tr}(P_{t,t-1\mid n}^{(1k)}\Gamma_k) \right . \nonumber\\
    &\left. +\text{Tr}(B_kB_k^\top)+2\text{Tr}(F_{t-1\mid n}^{*(k)\top}\Phi_k^{\top}B_k\Gamma_k)+\text{Tr}(\Gamma_kP_{t-1\mid n}^{*(1k)}\Gamma_k^{\top})\right],
\end{align}
where
\begin{align*}
    P_{t,t-1\mid n}^{(1k)}=&\sum_{d=1}^{k_1}(I_{k_2}\otimes e_{k_1}^{(d)\top}\Phi_k)p_{t,t-1\mid n}^{(k)\top}(I_{k_2}\otimes e_{k_1}^{(d)}),\\
    P_{t-1\mid n}^{*(1k)}=&\sum_{d=1}^{k_1}(I_{k_2}\otimes e_{k_1}^{(d)\top}\Phi_k)p_{t-1\mid n}^{*(k)\top}(I_{k_2}\otimes \Phi_k^{\top} e_{k_1}^{(d)}).
\end{align*}
Maximizing (\ref{es2}) with respect to \(\Gamma_k\), we can get
\begin{align}\label{es-th22}
    \widehat{\Gamma}_k=\left[\sum_{t=1}^n w_{t\mid n}^{(k)}\left(P_{t,t-1\mid n}^{(1k)\top}-\widetilde{B}_k^\top\widetilde{\Phi}_kF_{t-1\mid n}^{*(k)}\right)\right]\left(\sum_{t=1}^n w_{t\mid n}^{(k)}P_{t-1\mid n}^{*(1k)}\right)^{-1}.
\end{align}

\subsubsection{ The iterative equation of $P$}

Since \({Q_3(\rho;\widetilde{\theta})}= \mathop{\sum }_{t = 1}^n\mathop{\sum}_{i = 1}^M\mathop{\sum }_{j = 1}^M(\log p_{ij})w_{t-1,t\mid n}^{(i,j)}\), then \( \partial Q_3(\rho;\widetilde{\theta})/\partial\rho =\hat{\xi} \oslash \rho,\) where \(\rho \triangleq \text{vec}(P)\) and \( \hat{\xi} = \mathop{\sum}_{t = 1}^n \text{vec}(W_t)\) with \(W_t=(w_{t-1,t\mid n}^{(i,j)})_{i=1,j=1}^{M~M}\). The maximization problem is constrained by the $M$ adding restriction \(P1_M=1_M,\) that is \((1_M^\top\otimes I_M)\rho=1_M\). Then we will consider the Lagrange argument function given by
$$Q_3^\lambda(\rho;\widetilde{\theta})=Q_3(\rho;\widetilde{\theta})-\lambda^\top[(1_M^\top\otimes I_M)\rho-1_M],$$
where \( \lambda \) is the vector of the corresponding Lagrange multipliers. The derivation of \( {Q}^{\lambda } \) with respect to \( \rho \) is given by
\begin{align}\label{es3}
    \frac{\partial Q_3^\lambda(\rho;\widetilde{\theta})}{\partial\rho}=\hat{\xi}\oslash\rho-(1_M\otimes I_M)\lambda=0.
\end{align}
The solution of equation (\ref{es3}) for \(\rho\) yields \(\rho=\hat{\xi}\oslash(1_M\otimes\lambda)\) . Then applying the adding restriction \((1_M^\top\otimes I_M)\rho=1_M\) to the above equation results in
\begin{align*}
    (1_M^\top\otimes I_M)[\hat{\xi}\oslash(1_M\otimes\lambda)]=1_M,
\end{align*}
and hence \(\lambda  = (1_M^\prime\otimes I_M)\hat{\xi}\). Inserting the result of \( \lambda \) into the estimation equation of \( \rho\) yields
\begin{align}\label{es_P}
    \hat{\rho}=\hat{\xi}\oslash\left\{1_M\otimes \left[(1_M^\prime\otimes I_M)\hat{\xi}\right]\right\}.
\end{align}

\subsection{  The filtering and smoothing of \(F_t\) and \(s_t\)}\label{Sec3.2}

\subsubsection{ The Kalman filtering algorithm of \(F_t\) conditioned on \(s_t\) }

For each \(m,k\in [M]\) and \(t,s\in [n]\), define
\begin{align*}
    f_{t\mid s}^{(m,k)}:=&\mathbb{E}[f_t\mid \mathcal{Y}_s,s_{t-1}=m,s_t=k;\theta],&V_{t\mid s}^{(m,k)}:=\text{Cov}[f_t\mid \mathcal{Y}_s,s_{t-1}=m,s_t=k;\theta],\\
    y_{t\mid s}^{(m,k)}:=&\mathbb{E}[y_t\mid \mathcal{Y}_s,s_{t-1}=m,s_t=k;\theta],&{\Sigma}_{t\mid s}^{(m,k)}:=\text{Cov}[y_t\mid \mathcal{Y}_s,s_{t-1}=m,s_t=k;\theta],\\
    f_{t\mid s}^{(m)}:=&\mathbb{E}[f_t\mid \mathcal{Y}_s,s_t=m;\theta],&V_{t\mid s}^{(m)}:=\text{Cov}[f_t\mid \mathcal{Y}_s,s_t=m;\theta].
\end{align*}

Under normal assumptions on \(E_t\) and \(\epsilon_t\), the conditional mean and covariance mentioned above can be obtained through the following Kalman filter recursions. The proof of these recursions follows \cite{1979Anderson} (pp. 39-41) and \cite{Helmut2005} (pp. 630-632).

(1) Initialization : \(f_{0|0}^{(m)},V_{0|0}^{(m)}.\)

(2) Prediction steps (\(1\leq t\leq n\)):
\begin{align*}
    f_{t\mid t-1}^{(m,k)}=&\beta_k+(\Gamma_k\otimes\Phi_k)f_{t-1\mid t-1}^{(m)},\\
    V_{t\mid t-1}^{(m,k)}=&(\Gamma_k\otimes\Phi_k)V_{t-1\mid t-1}^{(m)}(\Gamma_k\otimes\Phi_k)^\top+\sigma_\epsilon^2I_r,\\
    y_{t\mid t-1}^{(m,k)}=&(C_k\otimes R_k)f_{t\mid t-1}^{(m,k)},\\
    {\Sigma}_{t\mid t-1}^{(m,k)}=&(C_k\otimes R_k)V_{t\mid t-1}^{(m,k)}(C_k\otimes R_k)^\top+\sigma^2I_{pq}.
\end{align*}

(3) Correction step (\(1\leq t\leq n\)) :
\begin{align*}
    f_{t\mid t}^{(m,k)}=&f_{t\mid t-1}^{(m,k)}+\mathcal{K}_t^{(m,k)}\left(y_t-y_{t\mid t-1}^{(m,k)}\right),\\
    V_{t\mid t}^{(m,k)}=&V_{t\mid t-1}^{(m,k)}-\mathcal{K}_t^{(m,k)}{\Sigma}_{t\mid t-1}^{(m,k)}\mathcal{K}_t^{(m,k)\top},
\end{align*}
 where $\mathcal{K}_{t\mid t}^{(m,k)}=V_{t\mid t-1}^{(m,k)}(C_k\otimes R_k)^\top({\Sigma}_{t\mid t-1}^{(m,k)})^{-1}$
is the Kalman filter gain.

(4) Reduce the (\(M\times M\)) posteriors (\( f_{t\mid t}^{(m,k)}\) and \(V_{t\mid t}^{(m,k)}\)) to the $M$ posteriors (\( f_{t\mid t}^{(k)}\) and \(V_{t\mid t}^{(k)}\)). See Section A.2 of the Supplementary Material for a detailed calculation.
\begin{align*}
     f_{t\mid t}^{(k)}=&\frac{\sum\limits_{m=1}^M P_r[s_{t-1}=m,s_t=k\mid \mathcal{Y}_t;\theta]f_{t\mid t}^{(m,k)}}{P_r[s_t=k\mid \mathcal{Y}_t;\theta]}=\frac{\sum\limits_{m=1}^M w_{t-1,t\mid t}^{(m,k)}f_{t\mid t}^{(m,k)}}{w_{t\mid t}^{(k)}},\\
     V_{t\mid t}^{(k)}=&{\sum\limits_{m=1}^M w_{t-1,t\mid t}^{(m,k)}\left[V_{t\mid t}^{(m,k)}+(f_{t\mid t}^{(k)}-f_{t\mid t}^{(m,k)})(f_{t\mid t}^{(k)}-f_{t\mid t}^{(m,k)})^\top\right]}\big/{w_{t\mid t}^{(k)}}.
\end{align*}

Using matrix algebra, the above Kalman filter recursions  can be further simplified within the matrix factor model framework to avoid high-dimensional matrix operations (see Lemma \ref{lema1}). The proof of Lemma \ref{lema1} appears in Section A.1 of the Supplementary Material. The simplified Kalman filter recursions are presented in Algorithm 2 of the  Supplementary Material.

\begin{lemma}\label{lema1}
 Define \({\Lambda}_k=C_k\otimes R_k\). For each \(m,k \in [M],\)  \(t \in   [n]\), the correction step given in Subsection 3.2.1 can be simplified as,
\begin{align*}
    V_{t\mid t}^{(m,k)} =& \left(I_r + \frac{1}{\sigma^2}V_{t\mid t-1}^{(m,k)}\Lambda_k^\top\Lambda_k\right)^{-1}V_{t\mid t-1}^{(m,k)},\\
    f_{t\mid t}^{(m,k) }=&f_{t\mid t-1}^{(m,k)}+\frac{1}{\sigma^2}V_{t\mid t}^{(m,k)}\left({\Lambda_k^\top y_t-\Lambda_k^\top\Lambda_kf_{t\mid t-1}^{(m,k)}}\right).
\end{align*}
\end{lemma}

The last thing to be considered to complete this filtering is to calculate \(P_r[s_{t-1} = m,s_t = k \mid  \mathcal{Y}_t;\theta ]\triangleq  w_{t-1,t\mid t}^{( m,k) }\) and \(w_{t \mid t}^{(k)}:=P_r[w_t = k \mid \mathcal{Y}_t;\theta ]\). Following  \cite{1989Hamilton}  and \cite{Kim1994} with slight modifications for the high-dimensional computational problem, the detailed computational process is outlined in Algorithm 2. See Supplementary Material A.3 for a detailed derivation process.

\subsubsection{ The smoothing algorithm for \(F_t\) and $s_t$}

As pointed by \cite{Kim1994}, we can first calculate smoothed probabilities \(w_{t,t + 1\mid n}^{(j,k) }:= P_r[s_t=j ,s_{t + 1} = k \mid   \mathcal{Y}_n;\theta ]\) and \(w_{t\mid n}^{(j)}:=P_r[s_t = j \mid \mathcal{Y}_n;\theta ]\) , and then these smoothed probabilities can be used to get smoothed values of \(f_t\) .

(1) For \(t = n - 1,\cdots 1\) , get the smoothed probabilities \(w_{t,t + 1\mid n}^{(i,k)}\) and \(w_{t\mid n}^{(j) }\) by
\begin{align*}
    w_{t,t+1\mid n}^{(j,k)}={w_{t+1\mid n}^{(k)}w_{t\mid t}^{(j)}p_{jk}}/{w_{t+1\mid t}^{(k)}},~~~~
    w_{t\mid n}^{(j)}=\sum_{k=1}^Mw_{t,t+1\mid n}^{(j,k)}.
\end{align*}

(2) For \(t=n-1,\cdots ,1\) , the smoothing algorithm for \(F_t\) given \(s_t=j\) and \(s_{t + 1} = k\) is as follows:
\begin{align*}
    f_{t\mid n}^{(j,k)}&=f_{t\mid t}^{(j)}+G_t^{(j,k)}({f_{t+1\mid n}^{(k)}-f_{t + 1\mid t}^{(j,k)}}),\\
    V_{t\mid n}^{(j,k)}& =V_{t\mid t}^{(j)}+G_t^{(j,k)}( {V_{t+1\mid  n}^{( k)}-V_{t+1\mid t}^{(j,k)}})G_t^{(j,k)\top},\\
    f_{t\mid n}^{(j)}&={\sum\limits_{k=1}^{M}{w}_{t,t+1\mid n}^{(j,k)}f_{t\mid n}^{(j,k)}}/{w_{t\mid n}^{(j)}},\\
    V_{t\mid n}^{(j)}&= \sum_{k=1}^M w_{t,t+1\mid n}^{(j,k)}\left[{V_{t\mid n}^{(j,k)}+\left({f_{t\mid n}^{(j)}-f_{t\mid n}^{(j,k)}}\right){\left(f_{t\mid n}^{(j)}-f_{t\mid n}^{(j,k)}\right)}^{\top}}\right] \big/w_{t\mid n}^{(j)},
\end{align*}
 where $G_t^{(j,k)}= V_{t\mid t}^{(j)}( {\Gamma_k\otimes \Phi_k}) (V_{t+1\mid t}^{(j,k)})^{-1}.$

The filtering and smoothing process is organized in Algorithm 2 in Section A.4 of the Supplementary Materials.

\subsection{ The EM with filtering and smoothing}
We summarize the above processes given in Sections 3.1 and 3.2 in the following algorithm (Algorithm 1). Given that the EM algorithm is only guaranteed to converge to a stationary point of the likelihood function, choosing good starting points is essential to increase the chances of convergence to a global maximum. The detailed procedure for determining initial values is outlined in Section A.6 of the Supplementary Materials. In the R function "MsDMF.R", the initialization process is carried out by the sub-function "ini".

\begin{algorithm}[h]
\label{EM}
\caption{  EM Algorithm with Filtering and Smoothing}
\KwIn{Data $(Y_t)_{t=1}^n$, factor dimensions $(k_1,k_2)$, number of states $M$, converge level $\epsilon$, and maximum number of iterations $n_{max}$.}
\KwOut{$\widehat{\theta}, ~(\widehat{f}_t)_{t=1}^n,~ (\widehat{s}_t)_{t=1}^n$}
\textbf{Initialization:} $\theta ^{(m)}, ~dis \leftarrow 1, ~ m \leftarrow 0,$ \;
     \While{  ($dis > \epsilon $ and $m \leq n_{max}$) }
            {
     Given $\theta ^{(m)}$, obtain $ \{f_{t|n}^{(k)}, V_{t|n}^{(k)},w_{t|n}^{(k)},w_{t-1,t|n}^{(i,k)}: t \in [n], i,k\in [K]\}$ by the filtering and smoothing algorithm ( Algorithm 2)\;

     Given $ \{f_{t|n}^{(k)}, V_{t|n}^{(k)},w_{t|n}^{(k)},w_{t-1,t|n}^{(i,k)}: ~t \in [n], ~i,k\in [K]\}$, update $\theta ^{(m)}$ to $\theta ^{(m+1)}$ by equations (\ref{es-th1}), (\ref{es-th2}) and (\ref{es-th22})\;

     $dis \leftarrow \|\theta ^{(m+1)}-\theta ^{(m)} \|_F^2$\;

     $m\leftarrow m+1 $

            }
    return  $\widehat{\theta}= \theta ^{(m+1)}, ~ \widehat{f}_t= \sum_{k=1}^K w_{t|n}^{(k)}f_{t|n}^{(k)},  ~ s_t= \arg \max_{k \in [K]} w_{t|n}^{(k)},~ t \in [n]$.

    \end{algorithm}

\addtolength{\textheight}{-.2in}%

\section{ Assumptions and asymptotic results}
In this section, we establish the theoretical properties of the proposed estimators. First, we present a set of regularity conditions under which the asymptotic properties are derived. Subsequently, the consistency and convergence rates of the estimators are established. Throughout the remainder of the paper, the symbol $c_0$ denotes a positive constant that may vary across different contexts.
\subsection{ Assumptions}

\vspace{0.2cm}

~ \textbf {Assumption A.}
The Markov chain associated with \( P = {\left( {p}_{ij}\right) }_{i = 1,j = 1}^{M,M} \) is irreducible and aperiodic, and denote its unique stationary distribution as $\pi=(\pi_1^0,\cdots,\pi_M^0)$, where $\pi_k^0>0$ for all $k\in[M]$.

\vspace{0.2cm}

\textbf {Assumption B.}
     For any \( k,\ell \in  \left\lbrack  M\right\rbrack \) ,\(i \in  \left\lbrack  p\right\rbrack \) and \( j \in  \left\lbrack  q\right\rbrack \) ,

(1) \( {R}_{k}^{0\top}{R}_{k}^{0} = pD_k^{(1)},{C}_{k}^{0\top}{C}_{k}^{0} = qD_k^{(2)},\begin{Vmatrix}{\gamma }_{k,i\cdot}^{0}\end{Vmatrix}_F \leq  c_0 \) and \( \begin{Vmatrix}{c}_{k, j\cdot}^{0}\end{Vmatrix}_F \leq  c_0\) for some

\( c_0 > 0 \), where $D_k^{(1)}$ and $D_k^{(2)}$ are defined in (I1).

(2) For $k\neq \ell$, \( \mathop{\min }\limits_{t}{f}_{t}{}^{\top }{\Lambda }_{k}^{{0}^{\top }}{M}_{{\Lambda }_{\ell}^{0}}{\Lambda }_{k }^{0}f_t /(pq)\geq  c_0 \) with $\Lambda_k= C_k \otimes R_k.
$

\vspace{0.2cm}

\textbf{Assumption C.} The Ms - VAR  parameters $\{(B_k,\Phi_k,\Gamma_k): k\in[M]\}$  satisfy \( {\left| \right| B_k \left| \right| }_{F}^{2} \leq  c_0 \),
\( \parallel \Phi_k {\parallel }_{F}^{2} \leq  c_0 \)  and $\parallel \Gamma_k {\parallel }_{F}^{2} \leq  c_0 $ for some $c_0>0$, and make the factor process satisfy the following moment conditions.

(1) For some \( \alpha > {16} \) , there exists \( c_0  > 0 \) such that \( \mathbb{E}\left( {\left| \right| F_t{\left| \right| }_{F}^{\alpha }}\right)  \leq  {\ c_0} \).

(2) For $k \in [M],$ and  some positive definite $\Sigma_{F_k},\Sigma_{F_k}^{(1)}$ and $\Sigma_{F_k}^{(2)}$,

(2.1) $ \dfrac{1}{n\pi_k^0} \sum\limits_{t=1}^n{I}_{(s_t  = k)}{\bm{f}}_{t}\bm{f}_{t}^{\top}\overset{{p}}{ \longrightarrow  }{\Sigma }_{{F}_{k}},$

(2.2)
 $ \dfrac{1}{n\pi_k^0} \sum\limits_{t=1}^n {I}_{(s_t  = k)}(D_k^{(1)})^{1/2}{F}_{t}D_k^{(2)}F_{t}^{\top}(D_k^{(1)})^{1/2}\overset{{p}}{ \longrightarrow  }{\Sigma }_{{F}_{k}}^{(1)}, $

 (2.3) $\dfrac{1}{n\pi_k^0} \sum\limits_{t=1}^n {I}_{(s_t  = k)}(D_k^{(2)})^{1/2}{F}_{t}^{\top}D_k^{(1)}F_{t}(D_k^{(2)})^{1/2}\overset{{p}}{ \longrightarrow  }{\Sigma }_{{F}_{k}}^{(2)}$,

(2.4) \( p\lim \dfrac{1}{\left| {A}_{k}\right| }  \mathop{\sum }\limits_{{t \in  {A}_{k}}}{\bm{f}}_{t}{\bm{f}}_{t}^{\top } \) is also positive definite, where \( {A}_{k} \) denotes
any subset    ~~~~of  $\{ t:{s}_{t} = k \}$  with Cardinality \( \left| {A}_{k}\right| \) and \( \lim \left| {A}_{k}\right| /n > 0 \). \vspace{0.2
cm}

\vspace{0.2cm}

\textbf {Assumption D.}
 For some $c_0>0$,

(1) for  $i\in[p],j\in[q]$ and $t\in[n]$, \( \mathbb{E}\left( {{e }_{t,ij}}\right)  = 0,  \mathbb{E}\left( {e }_{t,{ij}}^{\alpha}\right)  \leq  c_0 \) for some \( \alpha  > {16} \);

(2) for  $i\in[p],j\in[q]$ and $t\in[n]$,

(2.1) $\mathop{\sum }\limits_{{i_1 = 1}}^{p}\mathop{\sum }\limits_{{j_1 = 1}}^{q}\left| { \mathbb{E}\left( {{e }_{t,ij},{e}_{t,i_1j_1}}\right) }\right|  \leq  c_0, $

(2.2) $\mathop{\sum }\limits_{{i_1 = 1}}^{p}\mathop{\sum }\limits_{{j_1 = 1}}^{q}\left| { \mathbb{E}\left( {{e}_{t,ij_1},{e}_{t,i_1j}}\right) }\right|  \leq  c_0, $

(2.3)
\( \mathop{\sum }\limits_{{d = 1}}^{n}\mathop{\sum }\limits_{{i_1 = 1}}^{p}\mathop{\sum }\limits_{{j_1 = 1}}^{q}\left| { \mathbb{E}\left\lbrack  {{e}_{t,ij},{e}_{d,i_1j_1}}\right\rbrack  }\right|  \leq  c_0 \).

(3) for any \(i,i_1   \in \left\lbrack  p\right\rbrack, j,j_1 \in\left\lbrack  q\right\rbrack \),  \( {k \in}\left\lbrack  M\right\rbrack \), and any bounded constant
\( {a }_{t,{i}_{1}{j}_{1}}, a_{t,i_1}, a_{ij} \),

(3.1) \(  \mathbb{E}{\begin{Vmatrix}\dfrac{1}{\sqrt{n}}\mathop{\sum }\limits_{t = 1}^{n}I_{( {s}_{t} = k)} \left( {e}_{t,ij}{e }_{t,i_1j_1} -  \mathbb{E}\left( {e }_{t,ij}{e}_{t,i_1j_1}\right) \right) \end{Vmatrix}}_{F}^{2} \leq  c_0 \),

(3.2) \(  \mathbb{E}{\begin{Vmatrix}\dfrac{1}{\sqrt{npq}} \sum \limits_{t=1}^n\mathop{\sum }\limits_{{i_1 = 1}}^{p}\mathop{\sum }\limits_{{j_1 = 1}}^{q}{I}_{(s_t=k)}a_{t,i_1j_1}\left( {e }_{t,ij_1}{e }_{t,i_1j} -  \mathbb{E}\left( {e }_{t,ij_1}{e }_{t,i_1j}\right) \right) \end{Vmatrix}}_{F}^{2} \leq  c_0 \),

(3.3) \(  \mathbb{E}{\begin{Vmatrix}\dfrac{1}{\sqrt{np}} \sum \limits_{t=1}^n\mathop{\sum }\limits_{{i_1 = 1}}^{p}{I}_{(s_t=k)}a_{t,i_1}\left( {e }_{t,ij_1}{e }_{t,i_1j} -  \mathbb{E}\left( {e }_{t,ij_1}{e }_{t,i_1j}\right) \right) \end{Vmatrix}}_{F}^{2} \leq  c_0 \),

(3.4) \(  \mathbb{E}{\begin{Vmatrix}\dfrac{1}{\sqrt{pq}} \mathop{\sum }\limits_{{i = 1}}^{p}\mathop{\sum }\limits_{{j = 1}}^{q}{I}_{(s_t=k)}a_{ij}\left( {e }_{t,ij}{e }_{t,i_1j_1} -  \mathbb{E}\left( {e }_{t,ij}{e }_{t,i_1j_1}\right) \right) \end{Vmatrix}}_{F}^{2} \leq  c_0 \).

(4) For some $\beta \geq 2$,
\(  \mathbb{E}{\begin{Vmatrix}\dfrac{1}{\sqrt{pq}} \mathop{\sum }\limits_{{i = 1}}^{p}\mathop{\sum }\limits_{{j = 1}}^{q}r_{k,is}c_{k,jl}e_{t,ij} \end{Vmatrix}}_{F}^{\beta} \leq  c_0 \) for $k\in [M], s\in[k_1],l\in[k_2]$ and  $t\in[n]$.

\vspace{0.2cm}

\textbf {Assumption E.}
For any deterministic vector \( v \) and \( w \) satisfying \( \parallel v{\parallel }_{F}^{2} = 1 \) and \( \parallel w{\parallel }_{F}^{2} = 1 \)
with suitable dimensions, assume that \( E\begin{Vmatrix}\frac{1}{\sqrt{n}}\sum_{t=1}^n I_{(s_t=k)}F_tv^\top E_tw\end{Vmatrix}_F^2  \leq  c_0 \) for all $k\in[M]$.

\vspace{0.2cm}

\textbf {Assumption F.} For $i\in[p],j\in[q],l\in[k_1],h\in[k_2]$ and $k\in[M]$,

(1)  $\dfrac{1}{\sqrt{qn}}\sum\limits_{t=1}^nI_{(s_t=k)}(D_k^{(1)})^{\frac{1}2}F_tC_k^{0\top}e_{t,i\cdot} \xrightarrow[]{d} N(0, V_{1i})$ for some positive definite $V_{1i}$;


(2) $\dfrac{1}{\sqrt{pn}}\sum\limits_{t=1}^nI_{(s_t=k)}(D_k^{(2)})^{\frac{1}{2}}F_t^{\top} R_k^{0\top}e_{t,\cdot j} \xrightarrow[]{d} N(0, V_{2i})$ for some positive definite $V_{2i}$.






Assumption A is a standard Markov chain assumption. Assumption B(1) imposes orthogonality and boundedness on the factor loadings, which is common in factor models \citep{2012Statistical,2024Quasi}. Assumption B(2) requires that the loading spaces across different regimes differ $-$ the regime-distinctiveness criterion serving as an identification condition. A more detailed discussion is provided in \cite{2024Estimation}.
Assumptions C(1), C(2.1)-C(2.3) specify that the factors have bounded moments. Assumption C(2.4) corresponds to Assumption 1(1) in \cite{2024Estimation}, ruling out the possibility that for any regime $k$, the subsample $\{t: s_t=k\}$ could be decomposed into sub-regimes.
Assumptions D, E, and F are slight modifications of corresponding assumptions in \cite{2021Projected} and \cite{2024Quasi}.

\subsection{ Asymptotic results}

Theorems 1-6 establish the asymptotic properties of the proposed estimators, with corresponding proofs provided in Sections B to G of the Supplementary Material, respectively. We now define: $\delta_{pqn} = \min \{\sqrt{pq},\sqrt{n}\}$, $\widetilde{R}_k^0={R}_k^0({D}_k^{(1)})^{-1/2}$,  $\widetilde{C}_k^0={C}_k^0({D}_k^{(2)})^{-1/2},$ and $\widetilde{R}_k,\widetilde{C}_k$ as the QMLE of $\widetilde{R}_k^0$ and $\widetilde{C}_k^0$, respectively.

\begin{theorem}\label{Th1}
    \textbf{(Consistency of the estimated loading space)} Under Assumptions (A)-(D), $\dfrac{1}{pq}\left\|M_{\widehat{\Lambda}_k}\Lambda_{k}^{0}\right\|_F^2=O_p( 1/ \sqrt{\delta_{pqn}})$ for each $k\in [M]$ as $(p,q,n)\xrightarrow{} \infty$, where $\widehat{\Lambda}_k= \widehat{C}_k\otimes \widehat{R}_k$.
\end{theorem}
Theorem \ref{Th1} establishes the consistency of the estimated loading space without observation of the state variable $s_t$. Furthermore, it implies:
$\frac{\widetilde{C}_k^{0\top}\widetilde{C}_k}{q} \frac{\widetilde{C}_k^{\top}\widetilde C_k^{0}}{q} \otimes \frac{\widetilde{R}_k^{0\top}\widetilde{R}_k}{p} \frac{\widetilde{R}_k^{\top}\widetilde{R}_k^{0}}{p}=I_{r}+O_p({1}/{\sqrt{\delta_{pqn}}})$ where $r=k_1k_2$. This asymptotic representation is essential for deriving properties of loadings and factors, including the QMLE for matrix factor models in \cite{2024Quasi} and \cite{2023Two-way}.

\begin{theorem}\label{Th2}
    \textbf{(Consistency of the estimated states)} Under Assumptions (A)-(D), as $(p,q,n) \xrightarrow{} \infty$, for each $k\in[M]$ and any fixed $\eta>0$,

    (1) $\sup\limits_{t}|\widehat{w}_{t|n}^{(k)}-I_{s_t=k} |=o_p( \frac{1}{(pq)^\eta})$, if $n^{\frac{16}{\alpha}}/(pq)$ and $n^{\frac{2}{\alpha}+\frac{2}{\beta}}/(pq) \rightarrow0$;

    (2) $|\widehat{w}_{t|n}^{(k)}-I_{s_t=k} |=o_p(\frac{1}{(pq)^\eta})$.
\end{theorem}

Note that $\eta$ may be large but remains fixed as $(p, q, n) \to \infty$.
Theorem \ref{Th2} establishes the consistency of $\widehat{w}_{t|n}^{(k)}$ as $(p, q, n) \to \infty$, with uniform consistency holding if both $n^{16/\alpha}/(pq) \to 0$ and $n^{(2/\alpha) + (2/\beta)}/(pq) \to 0$.
When $\alpha$ and $\beta$ are large, these rate conditions are unrestrictive, but Assumptions C(1), D(1), and D(4) become stronger than those in \cite{2024Quasi} and \cite{2023Two-way}. Conversely, when $\alpha$ and $\beta$ are small, Assumptions C(1), D(1), and D(4) are more readily satisfied, while the rate conditions become restrictive.

We emphasize that the rate conditions $n^{16/\alpha}/(pq) \to 0$ and $n^{(2/\alpha) + (2/\beta)}/(pq) \to 0$ are sufficient but not necessary. Similar trade-off analyses appear in \cite{2024Estimation}. Theorems 3-6 inherit these requirements, as they are derived from Theorem \ref{Th2}.

\begin{theorem}\label{Th3}
    \textbf{(Convergence rates)} Under Assumptions (A)-(E),   $n^{\frac{16}{\alpha}}/(pq)$ and $n^{\frac{2}{\alpha}+\frac{2}{\beta}}/(pq) \rightarrow0$, as $(p,q,n) \xrightarrow{} \infty$, for each $k\in[M]$
\begin{align}
    (1)~~~~ & \frac{1}{p}\left\|\widetilde{R}_k-\widetilde{R}_k^0\widetilde{H}_{1k}\right\|_F^2=O_p\left( \frac{1}{nq}+\frac{1}{p^2q^2}+\frac{1}{n^2p}\right);\notag\\
     (2)~~~~ & \frac{1}{q}\left\|\widetilde{C}_k-\widetilde{C}_k^0\widetilde{H}_{2k}\right\|_F^2=O_p\left( \frac{1}{np}+\frac{1}{p^2q^2}+\frac{1}{n^2q}\right),\notag
\end{align}
where $ \widetilde{H}_{jk}$  satisfy $  \widetilde{H}_{jk}^{\top}\widetilde{H}_{jk}= \widetilde{H}_{jk}\widetilde{H}_{jk}^{\top}+o_p(1)=I_{k_j}+o_p(1), j=1,2.$
\end{theorem}

Theorem \ref{Th3} establishes that QMLEs for the  factor spaces achieve convergence rates that match those of \cite{2024Quasi}, although marginally slower than those of \cite{2021Projected}. The latter reports rates of $O_p\left( \frac{1}{nq} + \frac{1}{p^2q^2} + \frac{1}{n^2p^2} \right)$ for $R$ and $O_p\left( \frac{1}{np} + \frac{1}{p^2q^2} + \frac{1}{n^2q^2} \right)$ for $C$.

 \begin{theorem}\label{Th4}
    \textbf{(Asymptotic normality)} Under Assumptions (A)-(F),   $n^{\frac{16}{\alpha}}/(pq)$ and $n^{\frac{2}{\alpha}+\frac{2}{\beta}}/(pq) \rightarrow0$, as $(p,q,n) \xrightarrow{} \infty$, for each $k\in[M], i\in[p]$ and $j\in[q]$,

    (1) when $qn=o_p(\min \{p^2q^2,n^2p\})$,
    \begin{align}
        \sqrt{qn}\widetilde{H}_{1k}\left({ \widetilde{\gamma}_{k,i}. - \widetilde{H}_{1k}^{\top}{\gamma}^{0}_{k,i\cdot}}\right) \overset{L}{ \rightarrow  }{N}_{{k}_{1}}\left( {0},\left( {\pi }_{k}^{0}\right)^{-2}\left({\Sigma}_{F_k}^{\left( 1\right) }\right)^{-1}{V}_{1i}\left({\Sigma}_{F_k}^{\left( 1\right) }\right)^{-1}\right);\notag
    \end{align}

        (2) when $pn=o_p(\min \{p^2q^2,n^2q\})$,
    \begin{align}
       \sqrt{pn}\widetilde{H}_{2k}\left({ \widetilde{c}_{k,j}. - \widetilde{H}_{2k}^{\top}{c}^{0}_{k,j\cdot}}\right) \overset{L}{ \rightarrow  }{N}_{{k}_{2}}\left( {0},\left( {\pi }_{k}^{0}\right)^{-2}\left({\Sigma}_{F_k}^{\left( 2\right) }\right)^{-1}{V}_{2i}\left({\Sigma}_{F_k}^{\left( 2\right) }\right)^{-1}\right), \notag
    \end{align}
    where $ \widetilde{H}_{1k}$ and $\widetilde{H}_{2k}$ are defined in Theorem \ref{Th3}.
\end{theorem}

\begin{theorem}\label{Th5}
    Under Assumptions (A)-(F),   $n^{\frac{16}{\alpha}}/(pq)$ and $n^{\frac{2}{\alpha}+\frac{2}{\beta}}/(pq) \rightarrow0$, as $(p,q,n) \xrightarrow{} \infty$,
    $\widehat{p}_{ij} \overset{p}{ \rightarrow} p_{ij}^{0}$ for each $i,j \in[M]$.
\end{theorem}

\begin{theorem}\label{Th6}
    Under the Assumptions  of Theorem   \ref{Th5},  for each $k\in[M]$, when $B_k^0=0, \frac{1}{n}\sum_{t=1}^n I_{(s_t=k)}F_t= O_p(n^{-\frac{1}{2}})$, $ \widehat{R}_k^{\top}R_k^0 /p= D_k^{(1)}+ o_p(n^{-\frac{1}{2}})$,  and  $\widehat{C}_k^{\top}C_k^0 /q= D_k^{(2)}+ o_p(n^{-\frac{1}{2}}    )$, we have
    \begin{align*}
        \|\widehat{B}_k-B_k^0\|_F^2 & =O_p\left(\frac{1}{n}+\frac{1}{pq}\right), \\
        \|\widehat{\Phi}_k-\Phi_k^0\widehat{P}_{k2}\|_F^2 & =O_p\left(\frac{1}{n}+\frac{1}{pq}\right), \\
         \|\widehat{\Gamma}_k-\Gamma_k^0\widehat{P}_{k1}\|_F^2 & =O_p\left(\frac{1}{n}+\frac{1}{pq}\right),
    \end{align*}
    where $ \widehat{P}_{k2}=\left( \frac{1}{n}\sum_{t=1}^nF_{t-1}\Gamma_k^{0\top}\widehat{\Gamma}_k F_{t-1}^{\top}\right)\left( \frac{1}{n}\sum_{t=1}^nF_{t-1}\widehat{\Gamma}_k^{\top}\widehat{\Gamma}_k F_{t-1}^{\top}\right)^{-1}$ and $ \widehat{P}_{k1}=\left( \frac{1}{n}\sum_{t=1}^nF_{t-1}^{\top}\Phi_k^{0\top}\widehat{\Phi}_k \right.$ $ \left. F_{t-1} \right)\left( \frac{1}{n}\sum_{t=1}^nF_{t-1}^{\top}\widehat{\Phi}_k^{\top}\widehat{\Phi}_k F_{t-1}\right)^{-1}.$
\end{theorem}

According to the definition of $\widehat{P}_{k2}$ in Theorem \ref{Th6}, convergence of $\widehat{\Gamma}_k$ to $\Gamma_k^0$ implies $\widehat{P}_{k2} \to I_{k_1}$. Consequently, $\widehat{\Phi}_k \to \Phi_k^0$. This rotation matrix $\widehat{P}_{k2}$ arises from the rotational invariance in the model structure $\Phi_k F_{t-1} \Gamma_k^{\top}$, which  prevents unique identification of $\Phi_k$ and $\Gamma_k$.

\section{ Simulation Studies}
In this section, we examine the finite sample performance of the proposed estimation method for the Ms-DMF model from the following two aspects: (a) the estimation accuracy; (b) the robustness of the estimation method, specifically its fitting performance compared to alternative approaches under varying data structures and error distribution settings.
\subsection{ Simulation settings}
The simulation studies utilize data generated from an Ms-DMF model with $k_1 = k_2 = 2$ and $M = 2$. The data generation follows this procedure:
\begin{description}

        \item[Step 1: {\bf [State process]} ] The state variable $s_t$ follows a Markov chain with transition probabilities $p_{11} = p_{22} = 0.95$.

            \item[Step 2: \textbf{[Factors]}] The factor matrix $F_t$ is generated as:
$
F_t = B_{s_t} + \Phi_{s_t} F_{t-1} \Gamma_{s_t}^{\top} + \epsilon_t
$
where $\text{vec}(\epsilon_t) \stackrel{\text{i.i.d.}}{\sim} \mathcal{N}(\mathbf{0}, I_{k_1k_2})$. Parameters are specified as: $
\text{vec}(B_1) = b\beta, \quad \text{vec}(B_2) = 0.1\beta,  \quad
\text{vec}(\Phi_1) = \text{vec}(\Gamma_1) = (0.9, 0, 0, 0.7)^\top, \quad
\text{vec}(\Phi_2) = \text{vec}(\Gamma_2) = (0.7, 0, 0, 0.5)^\top,
$
with elements of $\beta$ independently drawn from $\mathcal{U}(0,1)$.

\item[Step 3: \textbf{[Loadings]}] For each regime $m \in \{1,2\}$:
     generate $\widetilde{R}_m$ ($p \times k_1$) and $\widetilde{C}_m$ ($q \times k_2$) with entries from $\mathcal{U}(2,4)$;
     define binary matrices $Q_{rm}$ and $Q_{cm}$ with mutually exclusive column supports such that their first and second columns have disjoint supports: positions containing 1 in the first column must be 0 in the second column (and vice versa). Here, \(\odot\) denotes the Hadamard product;
   compute normalized loadings:
    \begin{align*}
    R_m &= \left[p (Q_{rm} \odot \widetilde{R}_m)^\top (Q_{rm} \odot \widetilde{R}_m)\right]^{-1/2} (Q_{rm} \odot \widetilde{R}_m) \\
    C_m &= \left[q (Q_{cm} \odot \widetilde{C}_m)^\top (Q_{cm} \odot \widetilde{C}_m)\right]^{-1/2} (Q_{cm} \odot \widetilde{C}_m)
    \end{align*}

\item[Step 4: \textbf{[Errors]}] The error matrix follows a vector autoregressive process:
$
\text{vec}(E_t) = \psi \cdot \text{vec}(E_{t-1}) + \sqrt{1 - \psi^2} \cdot \text{vec}(U_t),
$
where $\text{vec}(U_t) \stackrel{\text{i.i.d.}}{\sim} \mathcal{N}(\mathbf{0}, \sigma^2 I_{pq})$.
\item[Step 5: \textbf{[Observations]}] The observed data is generated as:
$
Y_t = R_{s_t} F_t C_{s_t}^\top + E_t.$
\end{description}
We consider various $(p, q, n)$ combinations to evaluate estimation accuracy and perform 200 replications per parameter setting.

\subsection{ Estimation accuracy}
 We first assess the parameter recovery performance of the proposed EM method with filtering and smoothing for the Ms-DMF model. For the loading matrices $(R, C)$, we assess its performance by computing the distance between the estimated and true loading spaces using the metric established in \citet{2021Projected}. Regarding the latent factors $F_t$, we calculate regime-specific $R^2$ values through regressions of true factors on their estimated counterparts within each state. State recovery accuracy for $s_t$ is quantified via the Rand Index between the estimated and true state sequences, where values approaching 1 indicate stronger agreement. Finally, the estimation precision for all other parameters is measured by mean squared error (MSE).
		\begin{table}[H]
			\centering
             \setlength{\intextsep}{6pt} 
			\caption{\label{EA1}
				Averaged distance metrics of $R,C$, $R^2$ of $F_t$, Rand Index of $S_t$, and averaged MSE of other parameters over 200 replicates.
			}
               \renewcommand{\arraystretch}{0.7} 
                \resizebox{\textwidth}{!}{
				\begin{tabular}{ccccccccc}
					\toprule
					$n$ & $R_{s_t=1}$ & $R_{s_t=2}$ & $C_{s_t=1}$ & $C_{s_t=2}$ & $F_{s_t=1}$ & $F_{s_t=2}$& $P$ & $s_t$\\
					\midrule
 100 & 0.021  & 0.030  & 0.021  & 0.030  & 0.909  & 0.950  & 0.0011  & 0.9975  \\
         200 & 0.014  & 0.020  & 0.014  & 0.020  & 0.902  & 0.958  & 0.0006  & 0.9998  \\
      300 & 0.012  & 0.016  & 0.012  & 0.016  & 0.903  & 0.964  & 0.0005  & 0.9998  \\
         500 & 0.009  & 0.012  & 0.009  & 0.012  & 0.889  & 0.968  & 0.0002  & 0.9998  \\
						\cmidrule(lr){1-9}
					$n$ & $\sigma^2$ & $\sigma_{\varepsilon}^2$ & $\beta_{S_t=1}$ & $\beta_{S_t=2}$ & $\Phi_{S_t=1}$ & $\Phi_{S_t=2}$&  $\Gamma_{S_t=1}$  & $\Gamma_{S_t=2}$ \\
					\cmidrule(lr){1-9}
					 100 & 0.0048  & 0.015  & 0.118  & 0.166  & 0.058  & 0.087  & 0.074  & 0.043  \\
         200 & 0.0001  & 0.008  & 0.081  & 0.089  & 0.028  & 0.053  & 0.055  & 0.036  \\
         300 & 6.45$\times 10^{-5}$ & 0.007  & 0.073  & 0.066  & 0.019  & 0.039  & 0.051  & 0.031  \\
         500 & 7.17$\times 10^{-5}$ & 0.006  & 0.065  & 0.053  & 0.009  & 0.034  & 0.044  & 0.027  \\
					\bottomrule
				\end{tabular}
			}
		\end{table}
Table \ref{EA1} shows the averaged estimation errors  of $R$ and $C$, $R^2$ values for   $F_t$, Rand Index for $s_t$, and the averaged MSEs for other parameters under $p=q=10$, $n=100,200,300,500, \psi= 0.1, b=0.5, \sigma^2 =1 $. From Table \ref{EA1}, we can see that all parameters benefit from large dimension $n$, indicating that the EM with filtering and smoothing method possesses good estimation performance. Table  \ref{EA1} reveals that the estimation performance of $R_1$ and $C_1$ outperforms that of $R_2$ and $C_2$, respectively. This occurs because, compared to state 2, the factor series under state 1 exhibits larger intercept terms and autoregressive coefficients, resulting in a higher signal-to-noise ratio within the factor structure. This enhanced ratio facilitates more accurate estimation of the loading matrices. On the other hand, the larger intercept terms and variances of the factor series also imply that $F_t$ and its related autoregressive parameters are estimated less effectively under state 1 than under state 2.

		Figure \ref{Est_Acu} presents the average estimation errors of the loading parameters $R$ and $C$ under different matrix dimensions $(p, q)$, comparing the EM with filtering and smoothing method with the initial estimates given in Section A.6 of the Supplementary Material. The results indicate that: (a) the estimation accuracy of $\widehat{R}$  improves significantly with increasing dimension $q$, demonstrating better performance under larger $q$;  (b) the accuracy of $\widehat{C}$ is closely related to the dimension $p$, with estimation errors decreasing as $p$ increases, suggesting that higher dimensionality contributes to improved estimation accuracy of the loading matrices; and (c) the estimates from EM with filtering and smoothing are superior to the initial values.

\begin{figure}[htbp!]
			\centering
			\resizebox{\textwidth}{!}{\includegraphics{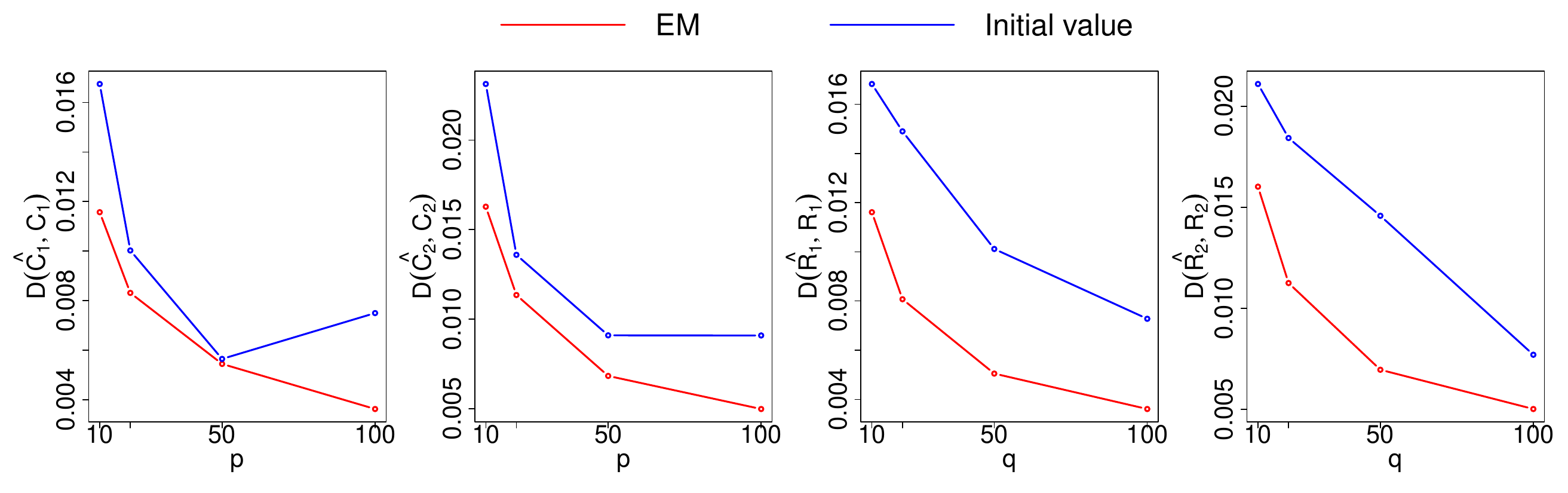}}\\
             \vspace{-18pt} 
			\caption{\small Averaged estimation errors of $R,C$ under different setting of $p$ and $q$. }\label{Est_Acu}
			\normalsize
		\end{figure}

		Figure \ref{Hete} illustrates the performance of the estimates under different levels of factor intercepts ($b = 0.2,0.4,0.6,0.8,1$), with fixed dimensions $p=q=10$ and a time length of $n = 300$.
		The results indicate that: (1) the estimation accuracy of the factor loadings $R$, $C$, and the observation error variance $\sigma^2$ (i.e., parameters in the observation equation) benefits from larger intercept values $b$, as we explained before; (2) the estimation errors of the state equation parameters--including $B$, $\Phi$, $\Gamma$, and $\sigma_\epsilon^2$--tend to increase with larger values of $b$; (3) since $b$ only affects the intercept in the first state, the estimation accuracy of the latent factors $F_t$ in the first state deteriorates as $b$ grows, while the estimation performance in the second state remains relatively stable.

		\begin{figure}[htbp!]
			\centering
			\resizebox{\textwidth}{!}{\includegraphics{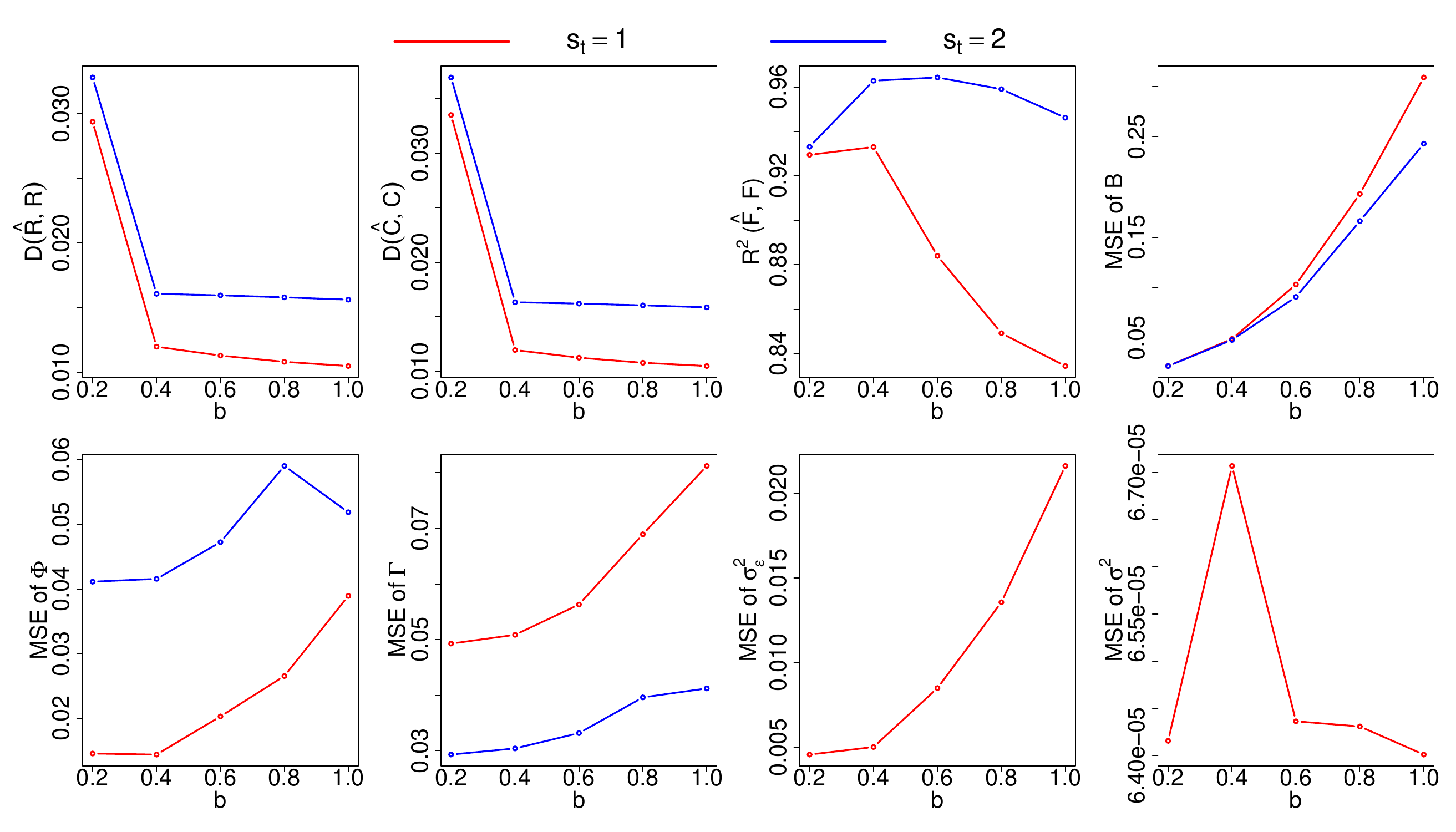}}\\
              \vspace{-20pt} 
			\caption{\small Averaged estimation errors  of  $R,C,F,B,\Phi,\Gamma,\sigma_{\varepsilon}^2,\sigma^2$ under different setting of $b$ ($p=q=10,n=300$).}\label{Hete}
			\normalsize
		\end{figure}
		
		\subsection{ Robustness Analysis}
In this subsection, we evaluate the robustness of the proposed method across varying data structures and error distributions by comparing its model-fitting performance against two established matrix factor model estimators: the PE method  \citep{2021Projected} and the Q-MLE approach \citep{2024Quasi}. We generate data under three distinct model specifications and two error distribution assumptions. Estimation performance is systematically compared using the mean squared error (MSE) of the common  $X_t=R_{s_t}F_tC_{s_t}'$:
 $\text{MSE} = \frac{1}{npq}\sum_{t=1}^n\|\widehat{X}_t - X_t\|_F^2$. The data-generating models are:

(a) \textbf{Ms-DMF}: Both observation and state equations depend on $s_t$ (Model 1);

(b) \textbf{State-only switching Ms-DMF}: Only state equations are regime-dependent while observation equations remain static (Model 2);

(c) \textbf{Static matrix factor model}: The matrix factor model without regime switching \cite{2024Quasi}.

Models (a) and (b) are generated following the simulation settings described in subsection 5.1, while Model (c) follows the data-generating process described in \cite{2024Quasi}. To evaluate methodological robustness under non-Gaussian conditions, we generate observation errors $E_t$ for all models under two distinct distributions: standard normal and chi-square with 1 degree of freedom.

 Figure \ref{robust} presents the average MSEs of the common component $X_t$ under a fixed time length $n = 300$, across different matrix dimensions ($p = q = 10, 20, 50, 100$) and various model specifications, comparing the performance of the three estimation methods. The results show that:
		(1) The fitting accuracy of $X_t$ using the one-step EM method is influenced by the matrix dimensions $p$ and $q$; in all model settings, the estimation improves as $p$ and $q$ increase;
		(2) Under model setting (a), the one-step EM method significantly outperforms the other two methods, while under settings (b) and (c), the performance of all three methods is comparable, with only minor differences.
\begin{figure}[!htbp]
			\centering
             \setlength{\intextsep}{6pt} 
		{\includegraphics[width=\textwidth, height=7.5cm]{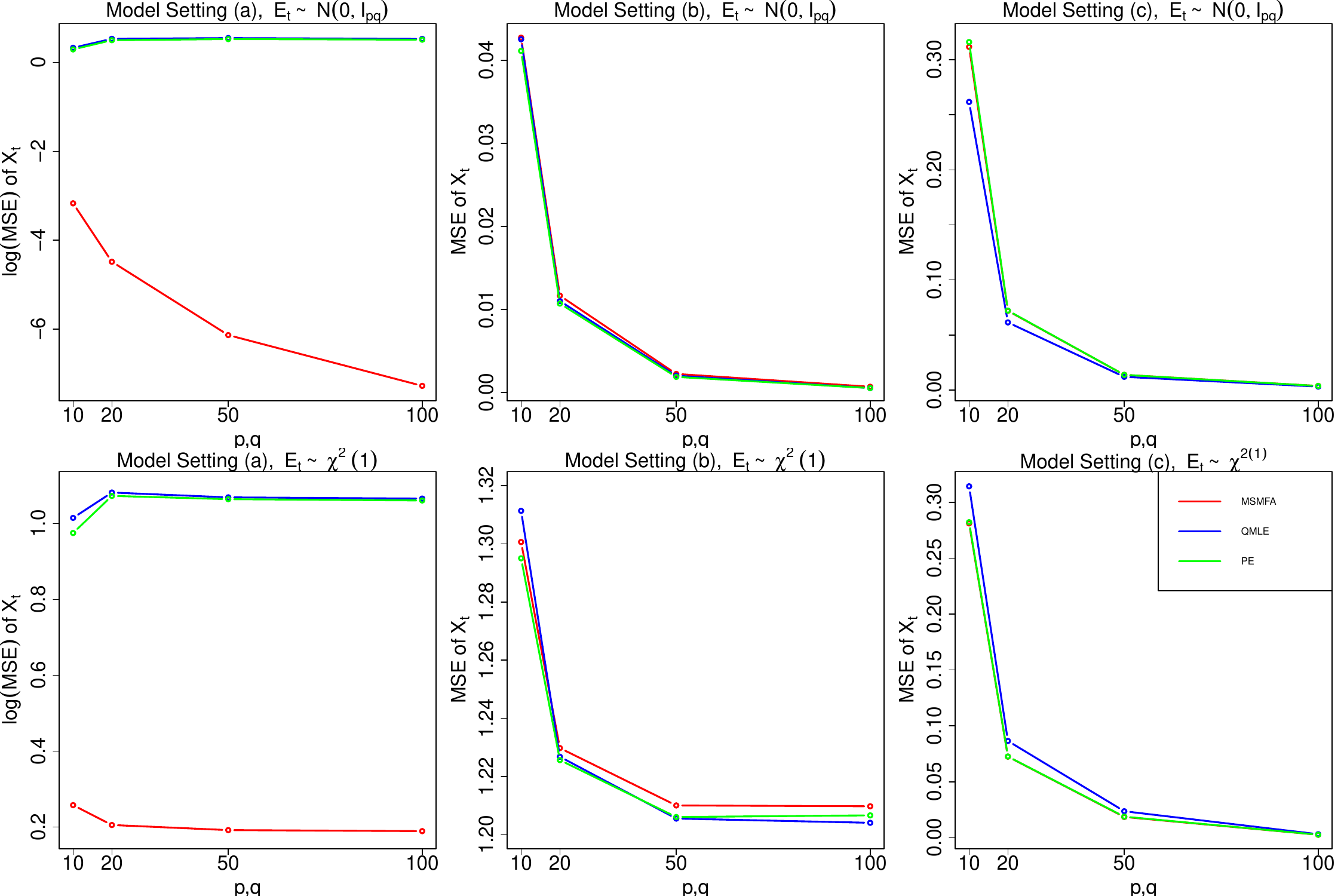}}\\
          \vspace{-18pt} 
			\caption{\small  The average MSEs of the common component $X_t$ under a fixed time length $n = 300$, across different matrix dimensions ($p = q = 10, 20, 50, 100$) and various model specifications.}\label{robust}
			\normalsize
		\end{figure}

\section{ Analysis of the International Trade Flow Data}

We analyze monthly multilateral import and export volumes of commodity goods between 22 economies from 2000 to 2024. Data are sourced from the International Monetary Fund's Direction of Trade Statistics (DOTS), which reports monthly bilateral trade values. The study uses total import values (CIF basis, USD-denominated). Missing import data for certain economies are imputed using corresponding export records from trading partners, consistent with IMF methodologies.
Economies included (Alphabetical Order):
Australia, Canada, China mainland, Denmark, Finland, France, Germany, Indonesia, Ireland, Italy, Japan, Korea, Malaysia, Mexico, Netherlands, New Zealand, Singapore, Spain, Sweden, Thailand, United Kingdom and United States.

Let $X_t=(x_{t,ij})_{i=1, j=1}^{22,22}, t=1,\cdots,300$ be the observed trade flow network, where $x_{t,ij}$ denotes the trade volume from country $i$ (exporter) to country $j$ (importer) in month $t$.  For each time series $(x_{t,ij})_{t=1}^{300}$, the seasonal component and trend component are removed by the Loess smoothing method with the "stl" function in R, and the irregular components are denoted as $y_{t,ij}$.  In the following, we will fit $Y_t=(y_{t,ij})_{i=1, j=1}^{22,22}$ by the proposed Ms-DMF model.

The initial step involves determining the number of factors in rows and columns. Consistent with the model's assumption of state-invariant factor dimensions, we employ the iterative algorithm from \cite{2021Projected}, which yields $k_1=2$ (row factors) and $k_2=1$ (column factors). For interpretability, we fix the number of regimes at $M=2$.

Table \ref{RC} reports the estimated factor loading matrices across regimes. Crucially, these matrices exhibit significant state-dependent variation. For example, China emerges as an export hub exclusively in State 2. The dominant export and import hubs under each regime are summarized in Table \ref{hub}.

\begin{table}[htbp]
  \centering
  \caption{The estimated factor loading matrices across different states.}
  \label{RC}
    \scriptsize
    \renewcommand{\arraystretch}{0.8} 
    \resizebox{\textwidth}{!}{
  \begin{tabular}{crrrrrr}
    \toprule
    & \multicolumn{2}{c}{$s_t=1$} & \multicolumn{2}{c}{$s_t=2$} & $s_t=1$ &$s_t=2$  \\
    \cmidrule(lr){2-3} \cmidrule(lr){4-5} \cmidrule(lr){6-7}
Country    &$R_{11}$ & $R_{12}$ & $R_{21}$ & $R_{22}$ & $C_{11}$ & $C_{21}$ \\
    \midrule
  Australia&	\cellcolor{sunshine}-5.92&	0.59&	0.12&	0.60&	-2.83&	12.77\\
Canada&	-0.34&	1.97&	\cellcolor{mint}-7.81&	5.17&	-10.34&	18.28\\
China&	0.08&	0.62&	-0.48&	\cellcolor{mint}43.03&	\cellcolor{lavender}-233.12&	8.58\\
Denmark&	-0.01&	0.22&	-0.08&	-0.19&	-0.10&	1.68\\
Finland&	-0.01&	0.14&	-0.11&	0.13&	-1.00&	1.19\\
France&	-0.02&	1.79&	-1.00&	-0.39&	-4.40&	11.5\\
Germany&	0.46&	7.30&	-2.07&	-0.72&	-4.27&	18.04\\
Indonesia&	-0.5&	2.47&	-0.29&	0.35&	-5.10&	8.97\\
Ireland&	-0.03&	0.44&	-0.29&	-0.27&	-0.70&	3.15\\
Italy&	-0.07&	1.35&	-1.03&	-0.02&	-1.84&	10.6\\
Japan&	-0.04&	\cellcolor{sunshine}14.35&	-2.13&	0.44&	0.40&	20.16\\
Korea&	0.59&	\cellcolor{sunshine}17.05&	-0.71&	0.25&	-12.66&	14.6\\
Malaysia&	0.74&	5.61&	-0.70&	0.31&	-4.05&	4.35\\
Mexico&	-0.01&	0.69&	\cellcolor{mint}-6.97&	-0.96&	-11.75&	19.04\\
Netherlands&	0.06&	0.86&	-0.60&	-0.45&	-2.59&	22.52\\
New Zealand&	-0.27&	0.10&	-0.04&	0.00&	-1.22&	1.51\\
Singapore&	-0.44&	1.68&	0.14&	0.49&	-9.53&	7.68\\
Spain&	-0.09&	0.46&	-0.35&	0.02&	-4.08&	9.31\\
Sweden&	-0.06&	0.42&	-0.22&	0.18&	-2.14&	2.31\\
Thailand&	-0.75&	2.52&	-0.38&	0.72&	-5.43&	7.44\\
United Kingdom&	0.04&	1.13&	-1.03&	0.09&	-5.73&	14.94\\
United States&	-0.48&	\cellcolor{sunshine}11.24&	-1.27&	0.41&	-10.92&	\cellcolor{lavender}203.87\\
    \bottomrule
  \end{tabular}
  }
\end{table}

\begin{table}[htbp]
  \centering
  \caption{The dominant export and import hubs under each regime.}
  \label{hub}
     \renewcommand{\arraystretch}{0.8} 
  \resizebox{\textwidth}{!}{
  \scriptsize
  \begin{tabular}{lll}
    \toprule
\textbf{State 1:}  & & \\
 Export Hub 1  & ($R_{11}$):  & Primarily comprises Australia \\
Export Hub 2& ($R_{12}$): & Dominated by Japan, Korea, and the United States\\
Import Hub &($C_{11}$): & Centered on China\\
    \midrule
\textbf{State 2:}&&\\
Export Hub 1 &($R_{21}$):& Led by Canada and Mexico\\
Export Hub 2& ($R_{22}$):& Driven primarily by China\\
Import Hub &($C_{21}$): & Led by the United States\\
    \bottomrule
  \end{tabular}
  }
\end{table}

Figure \ref{SS} displays the monthly regime distribution over the 25-year period (300 months). State 2 predominated, occurring in 199 months (66.3\% of observations), while State 1 occurred in 101 months. Consequently, during the State 2 periods, which represent approximately two-thirds of the sample, Canada, Mexico, and China consistently functioned as export hubs, while the United States (US) served as the main import hub. It further illustrates that during the 25-year period, State 2 predominantly occurred from March to November, with a particularly strong prevalence in March. In contrast, State 1 was primarily observed during December and January.

{\footnotesize
    \begin{figure}[H]
  \centering
 \includegraphics[width=\textwidth, height=5.5cm]{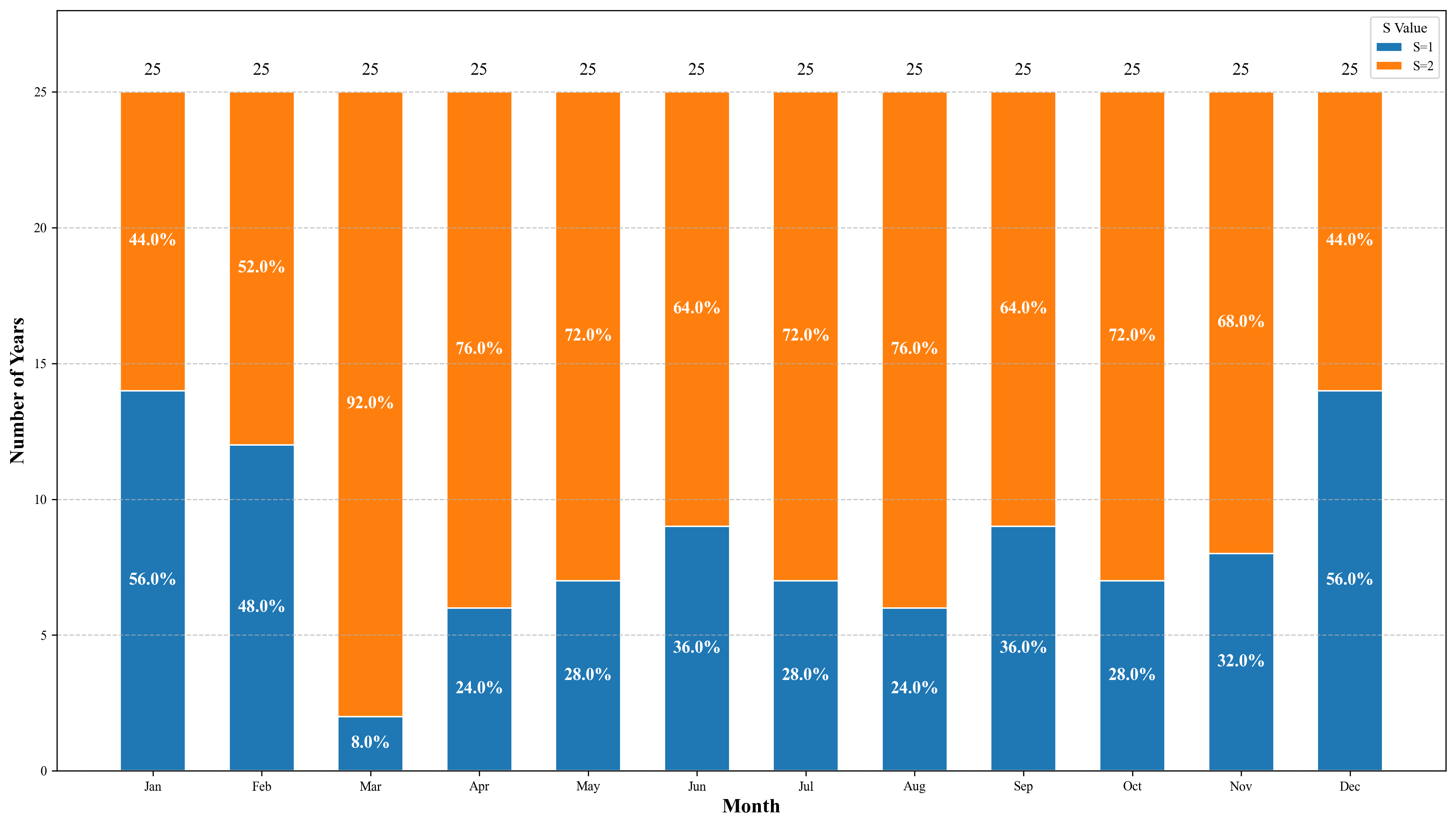}\\
   \vspace{-15pt} 
  \caption{\small  The monthly regime distribution over the 25-year period (300 months).}\label{SS}
\end{figure}
}

Figure \ref{China_out_S} displays the joint trajectory of the state variable and China's detrended/deseasonalized total exports, revealing a systematic association between the regimes and China's export cycles.
For the 2000-2005 period (upper panel), state 2 consistently coincided with declining export phases, while state 1 aligned with rising export phases.
For the 2020-2024 period (lower panel), this cyclical relationship reversed, with state 2 now corresponding to export expansions and state 1 to contractions.

{\footnotesize
    \begin{figure}[H]
  \centering
 \includegraphics[width=\textwidth, height=\textwidth/2]{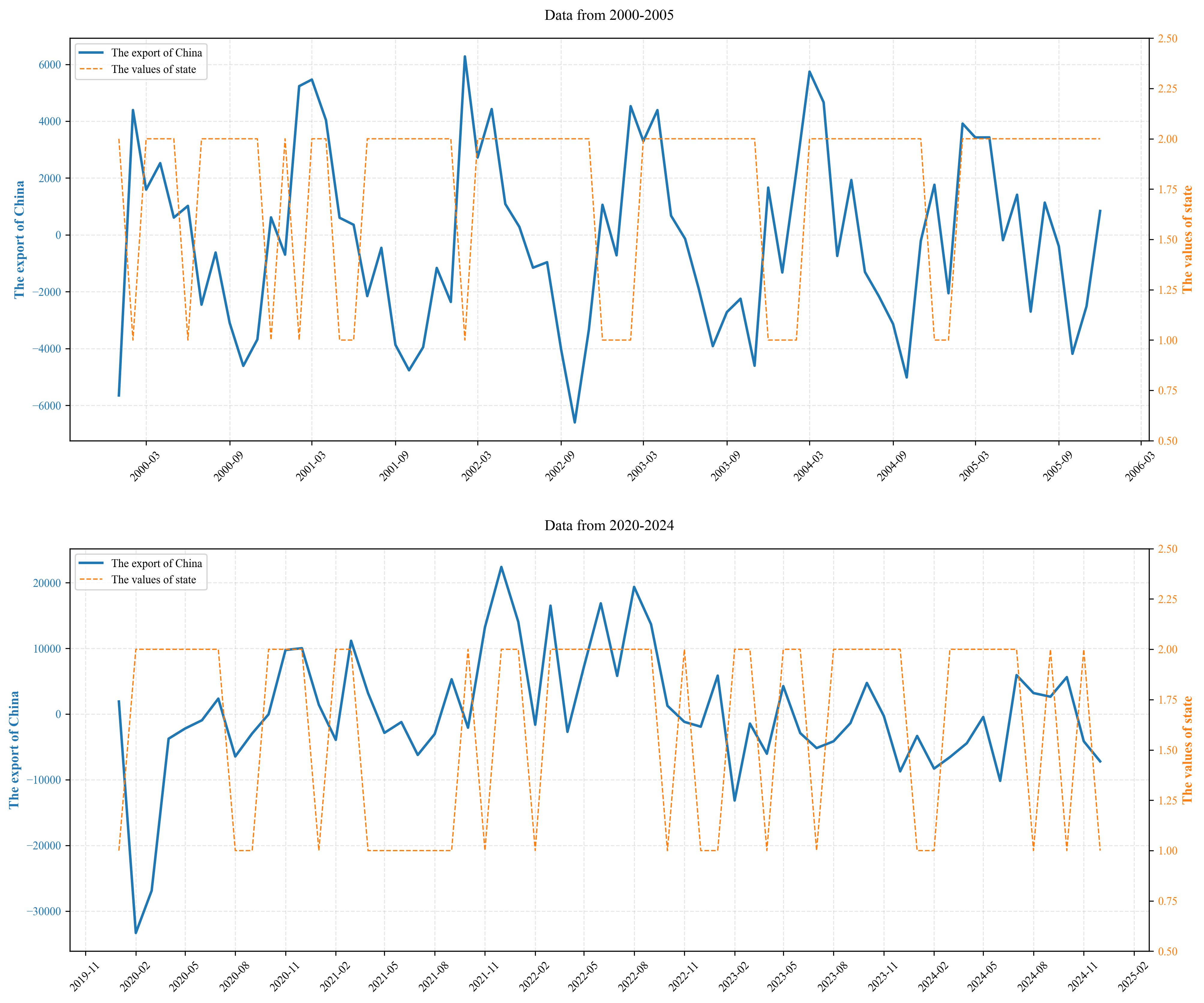}\\
   \vspace{-15pt} 
  \caption{\small  The joint trajectory of the state variable and China's total exports (detrended and deseasonalized).}\label{China_out_S}
\end{figure}
}

\cite{Modeling2024} interprets the matrix factor model applied to trade data by treating the factor process $F_t$ as representing trade volumes between latent hubs, analogous to trade between country clusters. To empirically validate this interpretation, Figure \ref{Fseries} (Upper panel) compares $F_{t,11}$ with Canada-to-U.S. exports. This comparison is made because Canada represents Export Hub 1 while the U.S. represents Import Hub 1 under State 2, as shown in Table \ref{hub}. Correspondingly, the lower panel of Figure \ref{Fseries} compares $F_{t,21}$ with China-to-U.S. exports, since China dominates Export Hub 2 under State 2. Both panels reveal a strong alignment between the trajectories of the latent factors and their corresponding bilateral trade series, empirically confirming the hub interpretation hypothesis.

{\footnotesize
    \begin{figure}[H]
  \centering
 \includegraphics[width=\textwidth, height=7.5cm]{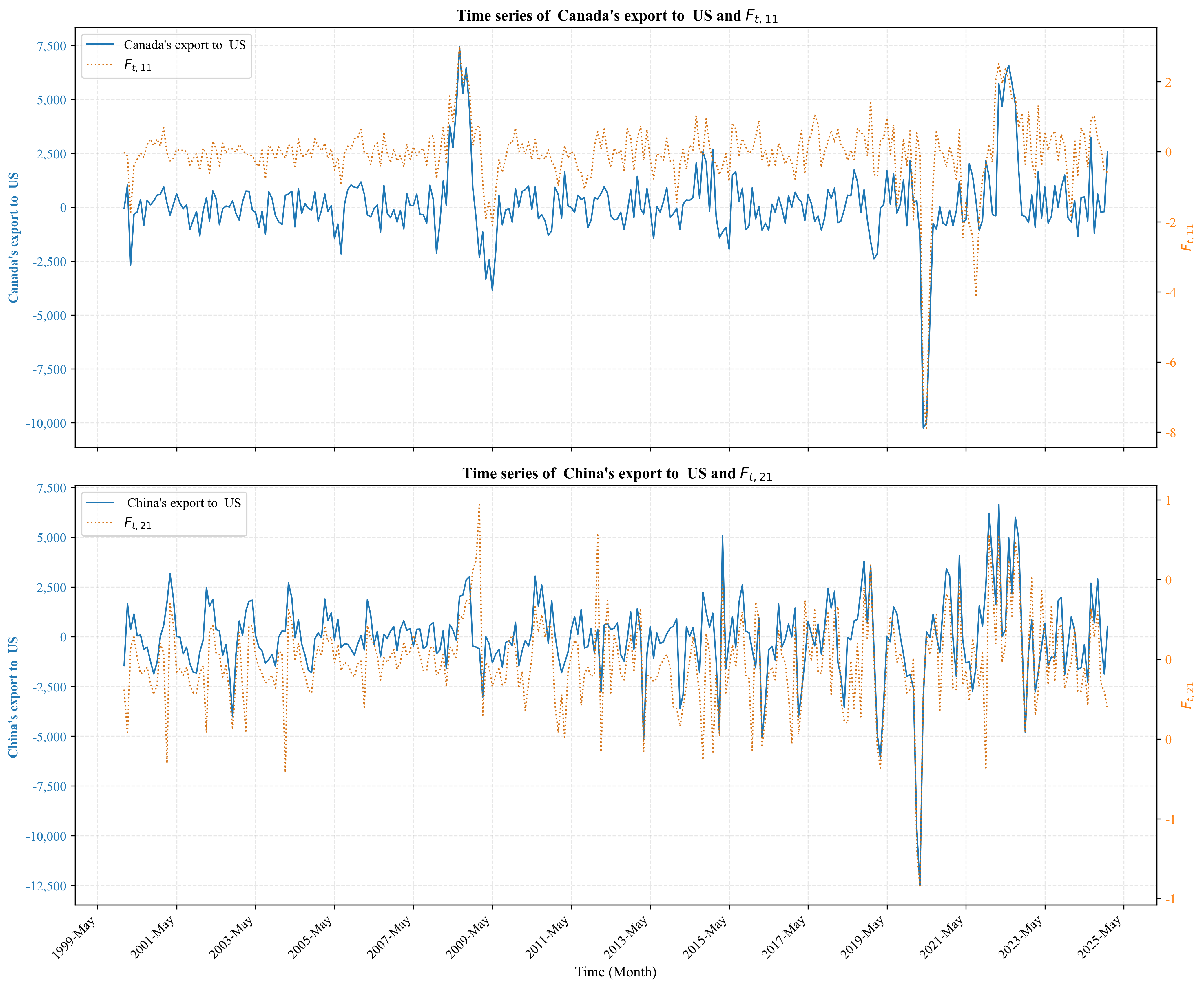}\\
   \vspace{-15pt} 
  \caption{\small Upper panel: The joint time series plots of
$F_{t,11}$ and Canada-to-U.S. exports. Lower panel: the joint time series plots of $F_{t,21}$ and China-to-U.S. exports.}\label{Fseries}
\end{figure}
}

A rolling forecast procedure is also applied to evaluate the forecast accuracy of the Ms-DMF model. For comparison, we consider both the standard matrix factor model with VAR applied to the factor matrix (MFM-VAR) and a simple autoregressive AR(1) model fitted to each individual time series. For each month $t$ from May 2023 to November 2024, we train the models using the 280 most recent observations preceding $t$ and generate one-step-ahead forecasts. Figure \ref{Forecast} displays the mean absolute forecast errors (MAPE) for these methods. The results show that during March, April, July, and August of 2024, the Ms-DMF model significantly outperformed the other two approaches. At all other time points, all three methods demonstrated comparable performance.

{\footnotesize
    \begin{figure}[htbp!]
  \centering
 \includegraphics[width=0.95\textwidth]{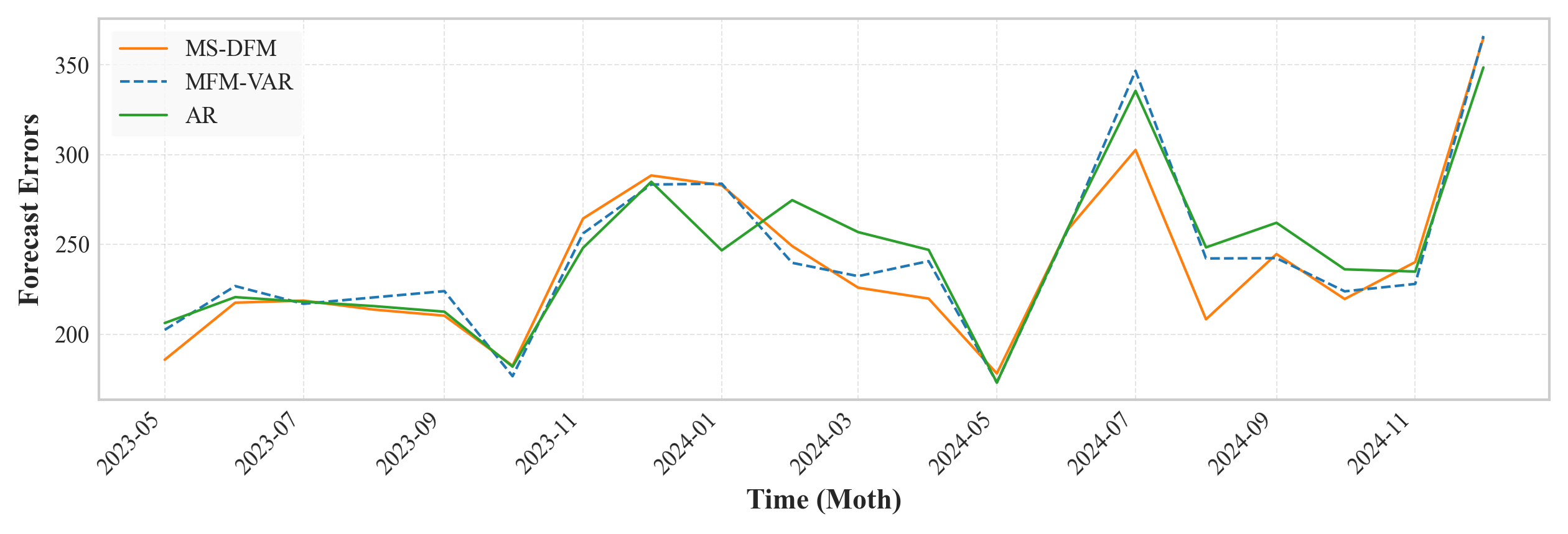}\\
   \vspace{-18pt} 
  \caption{\small The rolling forecast accuracy of the three methods (Ms-DMF, MFM-VAR and AR).} \label{Forecast}
\end{figure}
}

\section{ Conclusions and Discussions}
This study proposes a Markov-switching dynamic matrix factor (Ms-DMF) model designed to simultaneously achieve structural interpretation and dynamic prediction for high-dimensional matrix-valued time series data. By incorporating latent Markov regime states, the model allows both factor loadings and factor processes to switch across regimes, providing a flexible framework for capturing heterogeneity and dynamic patterns. For parameter estimation, the EM algorithm  with the filtering and smoothing process yields QMLEs. Under mild regularity conditions, we establish the consistency, convergence rates, and asymptotic distributions of the proposed estimators. Finally, we apply the Ms-DMF model to international trade flow data. Compared with existing matrix factor models, our approach not only identifies key import/export countries as trade hubs but also reveals regime switches and cyclical patterns among them. This provides new insights into the structural dynamics of global trade networks and highlights the model's broad applicability in economics and finance.

Notably, this study does not address automatic selection of the number of regimes or factors, nor does it consider cases where the number of factors varies across regimes. These important issues remain for future research.

\bigskip

\begin{center}

{\large\bf SUPPLEMENTARY MATERIAL}

\end{center}

\begin{description}
\item[Title:]
 This document provides Supplementary Information for the manuscript. It comprises seven sections: Section A contains omitted  results referenced in Section 3 of the main text, while Sections B through G present the proofs of Theorems 1 to 6, respectively.
\end{description}

\spacingset{1}


\newpage
\spacingset{1.8} 
\begin{center}
\large
\textbf{ Supplement to ``A Markov-switching  dynamic matrix factor model for the high-dimensional matrix-valued time series'' }
\end{center}

\vspace{2cm}
This document provides Supplementary Information for the manuscript titled ``A Markov-Switching Dynamic Matrix Factor Model for High-Dimensional Matrix Time Series". The supplement comprises seven sections: Section A contains omitted estimation results referenced in Section 3 of the main text, while Sections B through G present the proofs of Theorems 1 to 6, respectively.

 \setcounter{equation}{0}
 \setcounter{subsection}{0}
 \setcounter{lemma}{0}
 \renewcommand{\theequation}{A.\arabic{equation}}
 \renewcommand{\thesubsection}{A.\arabic{subsection}}
  \renewcommand{\thelemma}{A.\arabic{lemma}}
\section*{A ~~ Details for estimation}
\subsection{ Proof of Lemma 1}
\begin{proof}

 To simplify the presentation, we introduce the following notations in the proof: \(V \mathrel{\text{:= }} V_{t \mid t- 1}^{\left(m,k\right)}\), \({\Lambda }_k = C_k \otimes  R_k\) and  \( f_{t \mid t- 1}^{\left(m,k\right)}:= f_{t \mid t- 1}\). Since
\begin{align}\label{A1.1}
{\left({\Sigma}_{t\mid t-1}^{(m,k) }\right) }^{-1} = &{\left\lbrack\sigma^2I_{pq} + \Lambda _kV{\Lambda }_{k}^{\top }\right\rbrack  }^{-1} \notag\\
=& \frac{1}{{\sigma}^2}I_{pq} - \frac{1}{\sigma^2}\Lambda_{k}{\lbrack  {\sigma }^2V^{-1} +\Lambda_{k} ^{\top }{\Lambda }_{k}\rbrack}^{-1}\Lambda _k ^\top,
\end{align}
then
\begin{align*}
V_{t \mid  t}^{\left( m,k\right) }=&V-V{\Lambda }_k^\top{\left({\Sigma}_{t\mid  t-1}^{( m,k) }\right) }^{-1 }\Lambda_kV\\
=& V - \frac{1}{\sigma^2}V   {\Lambda }_k^{\top }\Lambda_kV + \frac{1}{\sigma^2}{V}\Lambda_{k}^\top {\Lambda }_k{\lbrack \sigma^2V^{-1} +\Lambda_k^\top\Lambda _k\rbrack  }^{-1}\Lambda _k^\top \Lambda _kV\\
=& V - \frac{1}{\sigma^2}V\Lambda_k^\top\Lambda_kV + \frac{1}{\sigma^2}V{\lbrack \sigma^2V^{-1 }{(\Lambda _k^\top\Lambda_k) }^{-1} +I_r\rbrack  }^{-1}\Lambda_k^\top\Lambda _kV.
\end{align*}

By the equality \({( I + AB) }^{-1}=I-A{( I+BA)}^{-1}B\) , we have
\begin{align*}
V_{t\mid t}^{(m,k)}=&V-\frac{1}{\sigma ^2}V\left\{ {I_r- {\left[ \sigma^2V^{-1}{\left(\Lambda_k^\top\Lambda_k\right) }^{-1} +I_r\right]  }^{-1}}\right\}\Lambda_k^\top\Lambda_kV\\
=& {V} - \frac{1}{\sigma^2}V\left\{{I_r-I_r+V^{-1}{(\Lambda_k^\top\Lambda _k) }^{-1}{\left[I_r+\sigma^2V^{-1}{(\Lambda_k^\top\Lambda_k) }^{-1}\right] }^{-1}\sigma^2}\right\}\Lambda_k^\top\Lambda_k{V}\\
=& V - {(\Lambda _k^\top\Lambda_k)}^{-1}{\left\lbrack  I_r +\sigma^2V^{-1}{\left( \Lambda_k^\top\Lambda_k\right) }^{-1}\right\rbrack}^{-1}{\Lambda }_k^\top\Lambda_{k}V\\
=& V - {(\Lambda_k^\top\Lambda_k)}^{-1}{\left\{I_r-V^{-1}({\Lambda_k^\top\Lambda_k)}^{-1}{\left\lbrack I_r+\sigma^2V^{ -1}{(\Lambda_k^\top\Lambda_k) }^{-1}\right\rbrack}^{-1}\sigma^2\right\}}  \Lambda_k^\top\Lambda_kV
\\
=& {(\Lambda_k^\top\Lambda_k)}^{-1}V^{-1 }{(\Lambda_k^\top\Lambda_k) }^{-1}{\left[ I_r + {\sigma}^2V^{-1}{(\Lambda_k^\top\Lambda_k) }^{-1}\right]  }^{-1}\Lambda_k^\top\Lambda_kV\sigma^2\\
=& (\Lambda_k^\top\Lambda_k)^{-1}{\left[\Lambda_k^\top\Lambda_kV+\sigma^2I_r\right]}^{-1} \sigma^2\Lambda_k^\top\Lambda_{k}V\\
=& {\left[\Lambda_k^\top\Lambda_kV\Lambda_k^\top\Lambda _k+\sigma^2\Lambda_k^\top\Lambda_k\right]}^{-1} \sigma^2\Lambda_k^\top\Lambda_kV\\
=&{\left[ I_r+ \frac{1}{\sigma^2}V\Lambda_k^\top\Lambda_k\right]}^{-1}  V.
\end{align*}
This establishes the first conclusion of Lemma 1.

Next, consider the result of \({f}_{t \mid t}^{\left( m,k\right) }\), where
\begin{align}\label{A1.2}
   f_{t\mid t}^{(m,k)}=f_{t\mid t-1}^{(m,k)}+ V_{t\mid t-1}^{(m,k)}\Lambda_k^\top{\left({\Sigma}_{t \mid t-1}^{(m,k)}\right)}^{-1}\left({y_t-\Lambda_kf_{t\mid t - 1}^{(m,k) }}\right) .
\end{align}
Substituting Equation (\ref{A1.1}) into Equation (\ref{A1.2}) gives
\begin{align*}
f_{t \mid t}^{(m,k)} =& f_{t \mid t- 1} + \frac{1}{\sigma^2}V\Lambda_k^\top({y_t-\Lambda_kf_{t \mid t-1}}) \notag\\
& ~~~~~~~~~- V\frac{1}{\sigma^2}\Lambda_k^\top\Lambda_k{(\sigma^2V^{-1}+\Lambda_k^\top\Lambda_k) }^{-1}\Lambda_k^\top({y_t-\Lambda_kf_{t \mid t-1}})\\
=&f_{t\mid t-1} + \frac{1}{\sigma^2}V\Lambda_k^\top y_t-\frac{1}{\sigma^2}V\Lambda_k^\top\Lambda_kf_{t\mid t- 1} - \frac{1}{\sigma^2}V\Lambda_k^\top\Lambda_k{(\sigma^2V^{-1}+\Lambda_k^\top\Lambda_k)}^{-1}\Lambda_k^\top y_t\\
 &~~~~~~~~~+\frac{1}{\sigma^2}V\Lambda_k^\top\Lambda_k{(\sigma^2V^{-1}+\Lambda_k^\top\Lambda_k)}^{-1}\Lambda_k^\top\Lambda_kf_{t \mid t-1}\\
=& f_{t \mid t-1} +\frac{1}{\sigma^2}V\left\lbrack {I_r-\Lambda_k^{\top }{\Lambda}_{k}{(\sigma^2V^{-1}+\Lambda_k^\top\Lambda_k) }^{-1}}\right\rbrack\Lambda_k^\top y_t\\
&~~~~~~~~-\frac{1}{\sigma^2}V\left\lbrack{I_r-\Lambda_k^\top\Lambda_k{(\sigma^2V^{-1}+\Lambda_k^\top\Lambda_k)}^{ -1}}\right\rbrack\Lambda_k^\top\Lambda_kf_{t\mid t-1}\\
=& {f}_{t \mid t-1} + \frac{1}{\sigma^2}V\left\lbrack{I_r-\left( \sigma^2V^{-1}(\Lambda_k^\top\Lambda_k) ^{-1}+I_r\right) }^{-1}\right\rbrack\Lambda_k^\top y_t \\
&~~~~~~~-\frac{1}{\sigma^2}V\left\lbrack{I_r-{\left(\sigma^2V^{-1}{(\Lambda_k^\top\Lambda_k)}^{-1}+I_r\right)}^{-1}}\right\rbrack\Lambda_k^\top\Lambda_kf_{t\mid t-1}
\end{align*}

Applying the identity \({(I+AB) }^{-1}=I-A{(I+BA) }^{-1}B,~f_{t\mid t}^{(m,k)}\) can be further simplified to
\begin{align*}
f_{t \mid  t}^{\left( m,k\right) } =&f_{t\mid t-1}+\frac{1}{\sigma^2}V\left\lbrack{V^{-1}{(\Lambda_k^\top\Lambda_k)}^{-1}{\left(I_r+\sigma^2V^{-1}{(\Lambda_k^\top\Lambda_k)}^{-1}\right) }^{-1}\sigma^2}\right\rbrack\Lambda_k^\top y_t\\
&~~~~~~~~-\frac{1}{\sigma^2}V\left\lbrack{V^{-1}{(\Lambda_k^\top\Lambda_k)}^{-1}{\left(I_r+\sigma^2V^{-1}{(\Lambda_k^\top\Lambda_k)}^{-1}\right)}^{-1}\sigma^2}\right\rbrack\Lambda_k^\top\Lambda_kf_{t\mid t-1}\\
=& f_{t\mid t-1} + \frac{1}{\sigma^2}{(\Lambda_k^\top\Lambda_k)}^{-1}{\left\lbrack I_r+\sigma^2V^{-1}{(\Lambda_k^\top\Lambda_k)}^{-1}\right\rbrack}^{-1}\sigma^2\Lambda_k^\top y_t\\
&~~~~~~~- \frac{1}{\sigma^2}{(\Lambda_k^\top\Lambda_k) }^{-1}\left\lbrack I_r+\sigma^2V^{-1}{(\Lambda_k^\top\Lambda_k) }^{-1 }\right\rbrack ^{-1} \sigma^2\Lambda_k^\top\Lambda_kf_{t\mid t-1}.
\end{align*}
Since
\begin{align}\label{A1.3}
{(\Lambda_k^\top\Lambda_k)}^{-1}{\left\lbrack I_r+\sigma^2V^{-1}{(\Lambda_k^\top\Lambda_k)}^{-1}\right\rbrack}^{-1}=&{(\Lambda_k^\top\Lambda_k) }^{-1}{\left\lbrack I_r+\frac{1}{\sigma^2}\Lambda_k^\top\Lambda_kV\right\rbrack }^{-1}\Lambda_k^\top\Lambda_kV\frac{1}{\sigma^2}\notag\\
=& {\left\lbrack\Lambda_k^\top\Lambda_k+\frac{1}{\sigma^2}\Lambda_k^\top\Lambda_kV\Lambda_k^\top\Lambda_k\right\rbrack}^{-1}\Lambda_k^\top\Lambda_kV\frac{1}{\sigma^2}\notag\\
=& {\left\lbrack I_r+\frac{1}{\sigma^2}V\Lambda_k^\top\Lambda_k\right\rbrack}^{-1}V\cdot\frac{1}{\sigma^2} \notag\\
=& V_{t \mid t}^{(m,k)}\frac{1}{\sigma^2},
\end{align}
where the first equality  holds due to \({(I+A^{-1})}^{-1} = {(A+I)}^{-1}A\), then we have

\[
f_{t\mid t}^{(m,k)}=f_{t\mid t-1}^{(m,k)}+\frac{1}{\sigma^2}V_{t\mid t}^{(m,k)}\left( {\Lambda_k^\top y_t-\Lambda_k^\top\Lambda_kf_{t\mid t-1}^{(m,k)}}\right) .
\]
\end{proof}

\subsection{ The derivation of the reduced posteriors \(f_{t |  t}^{(k) }\) and \(V_{t |  t}^{(k) }\)}

Consider $ f_{t \mid  t}^{(k) }$.
Since
\begin{align*}
f_{t \mid  t}^{(k)} &= \mathbb{E}[f_t \mid  \mathcal{Y}_t,s_t = k;\theta] = \mathop{\sum }\limits_{i = 1}^M \mathbb{E}[f_t \mid  \mathcal{Y}_t,s_{t-1} = i,s_t = k] \cdot  P_r[s_{t-1} = i \mid  \mathcal{Y}_t,s_t = k] \\
&= \mathop{\sum }\limits_{i = 1}^M f_{t|t}^{( i,k) }\frac{P_r( s_{t-1} = i,s_t= k \mid \mathcal{Y}_t) }{P_r( s_t = k \mid  \mathcal{Y}_t) },
\end{align*}
then we have
$$
f_{t \mid  t}^{(k) } = \mathop{\sum }\limits_{i = 1}^M f_{t \mid t}^{(i,k)}{w_{t-1, t \mid  t}^{(i,k) }}/{w_{t \mid t}^{( k) }}.
$$

Consider \(V_{t|t}^{(k)}\).
\begin{align*}	
V_{t|t}^{(k) } &= \mathbb{E}\left[(f_t - f_{t|t}^{(k)}) ( f_t - f_{t |t}^{(k)})^{\top} \mid  s_t = k,\mathcal{Y}_t\right] \\
&= \mathop{\sum }\limits_{i = 1}^M \mathbb{E}\left[  (f_t - f_{t | t}^{(k)}) (f_t - f_{t | t}^{(k)})^{\top } \mid  s_t = k,s_{t - 1} = i , \mathcal{Y}_t\right] P_r \left\lbrack s_{t - 1} = i \mid s_t = k,\mathcal{Y}_t \right\rbrack \\
&= \mathop{\sum }\limits_{i = 1}^M \mathbb{E}\left[ ( f_t - f_{t | t}^{(i,k)} + f_{t | t}^{(i,k) } - f_{t | t}^{(k)}) (f_t - f_{t| t}^{(i,k)} + f_{t | t}^{( i,k)} - f_{t | t}^{(k)})^{\top }\right. \\
& \quad \left.\mid s_{t - 1} =i,s_t = k, \mathcal{Y}_t\right] \frac{w_{t -1,t|t}^{( i,k)}}{w_{t|t}^{(k)}} \\
&= \mathop{\sum }\limits_{i = 1}^M\left[ V_{t|t}^{(i,k)} + ( f_{t |t}^{(i,k)} - f_{t|t}^{(k)}) ( f_{t|t}^{( i,k)} - f_{t|t}^{(k)})^{\top }\right]  \frac{w_{t - 1,t|t}^{(i,k)}}{w_{t|t}^{(k)}}.
\end{align*}

\subsection{ The filtering of \(s_t\)}
The argument below follows \cite{Kim1994}.

Step (1). Calculate
\begin{align*}
w_{t-1, t \mid t - 1}^{(i,k)} &\triangleq  P_r\left\lbrack s_{t - 1} =i,s_t = k \mid \mathcal{Y}_{t - 1}\right\rbrack  \\
&=P_r\left\lbrack s_t = k \mid s_{t - 1} = i\right\rbrack  \mathop{\sum }\limits_{i' = 1}^M P_r\left\lbrack s_{t - 2} = i',s_{t - 1} = i \mid \mathcal{Y}_{t - 1}\right\rbrack  \\
&= p_{ik}\mathop{\sum }\limits_{i^{' } = 1}^M w_{t - 2,t - 1|t - 1}^{(i^{'},i)}.
\end{align*}

Step(2). Calculate the joint density function of \(y_t\) and \(( s_{t - 1},s_t)\) :
\begin{align*}
f(y_t, s_{t - 1} = i,s_t = k \mid \mathcal{Y}_{t-1}) &= f(y_t \mid s_{t-1} = i,s_t = k,\mathcal{Y}_{t - 1}) P_r\left\lbrack s_{t - 1} = i,s_t = k \mid \mathcal{Y}_{t - 1}\right\rbrack  \\
&= w_{t-1,t \mid t-1}^{(i,k)}f(y_t \mid s_{t-1}=i,s_t = k,\mathcal{Y}_{t - 1}),
\end{align*}
where $$f(y_t \mid s_{t-1} = i,s_t =k,\mathcal{Y}_{t - 1}) = (2\pi)^{-\frac{pq}{2}}\left| \Sigma_{t \mid t-1}^{(i,k)}\right|^{-\frac{1}{2}}\exp\left\{-\frac{1}{2}\eta_{t|t - 1}^{( i,k)}\left( \Sigma_{t | t - 1}^{(i,k)}\right)^{-1} \eta_{t \mid t - 1}^{( i,k)}\right\},$$  and its detailed calculation is given in Lemma \ref{lemaA1}.

Step (3). Calculate
\begin{align}
    w_{t-1,t|t}^{(i,k)} &= P_r\left\lbrack s_{t-1} = i,s_t = k \mid \mathcal{Y}_t\right\rbrack = \frac{f(y_t,s_{t-1} = i,s_t=k \mid \mathcal{Y}_{t - 1})}{f(y_t \mid \mathcal{Y}_{t-1})} \notag\\
    &= \frac{f(y_t,s_{t-1} = i,s_t=k \mid \mathcal{Y}_{t-1})}{\mathop{\sum }\limits_{i = 1}^M\mathop{\sum }\limits_{k = 1}^M f( {y_t,s_{t-1}=i,s_t=k} \mid \mathcal{Y}_{t-1}) }.
\end{align}

Step(4).  \(w_{t | t}^{(k) } = P_r\left\lbrack  s_t=k \mid \mathcal{Y}_t\right\rbrack\) can be calculated by
$$w_{t | t}^{(k) } = \mathop{\sum }\limits_{i = 1}^M P_r\left\lbrack s_{t - 1} = i,s_t=k \mid \mathcal{Y}_t\right\rbrack = \mathop{\sum }\limits_{i = 1}^M w_{t-1,t|t}^{(i,k)}.$$

\begin{lemma}\label{lemaA1}
    For \(k \in  \lbrack M\rbrack\), we have
\begin{align}
   &f(y_t \mid s_{t-1} = m,s_t = k,\mathcal{Y}_{t - 1}) \nonumber\\
   &={(2\pi)}^{-\frac{pq}{2}}(\sigma^2)^{-\frac{pq}{2}}\mathop{\prod }\limits_{s = 1}^{k_1}\mathop{\prod }\limits_{\ell = 1}^{k_2}(1+\frac{1}{\sigma^2}{\operatorname{d_{s\ell}^v}})^{-\frac{1}{2}} \exp\{-\frac{1}{2}(A_1 - A_2 + A_3 - B_1 + B_2 - B_3) \}, \nonumber
\end{align}
where \(d_{s\ell}^v\) denotes the eigenvalues of \(\Lambda _k^{\top}\Lambda_k V_{t \mid t-1}^{(m,k)}\) with $\Lambda_k=C_k\otimes R_k$, and
\begin{align}
A_1 & = \frac{1}{\sigma^2} \text{Tr}( Y_t^{\top }Y_t), \nonumber\\
A_2 &= \frac{2}{\sigma^2}\text{Tr}( Y_t^{\top}R_kF_{t|t-1}^{(m,k)}C_k^{\top}), \notag\\
A_3 &= \frac{1}{\sigma^2}\text{Tr}\left(C_k F_{t|t-1}^{({m,k})\top}R_k^{\top}R_k F_{t|t-1}^{({m,k})}C_k^{\top}\right),\nonumber\\
B_1 &= \frac{1}{\sigma^4}\text{Vec}^{\top}( R_k^{\top }Y_t C_k) V_{t|t}^{(m,k)}\text{Vec}(R_k^{\top }Y_t C_k), \nonumber\\
B_2 &= \frac{2}{\sigma^4}{\text{Vec}}^{\top}(R_k^{\top }Y_t C_k) V_{t|t}^{(m,k) }\Lambda_k^{\top }\Lambda_k f_{t \mid t-1}^{(m,k)},\nonumber\\
B_3 &= \frac{1}{\sigma^4}f_{t \mid t-1}^{(m,k)\top}\Lambda_k^{\top }\Lambda_k V_{t \mid  t}^{(m,k) }\Lambda_k^{\top }\Lambda_k f_{t \mid t-1}^{(m,k)}. \notag
\end{align}
\end{lemma}
\begin{proof}

Firstly, consider \(|\mathop{\Sigma }_{t\mid t-1}^{(m,k)}|\) , where \(\mathop{\Sigma }_{t\mid t-1}^{(m,k)} = \Lambda_k V_{t|t-1}^{(m,k)}\Lambda_k^{\top } + \sigma^2 I_{pq}.\) By the Sylvester theorem,
\begin{align}\label{A3.1}
\left| {\Sigma}_{t\mid t-1}^{(m,k)}\right| &= (\sigma^2)^{pq}\left| I_{pq} + \frac{1}{\sigma^2}\Lambda_k V_{t|t-1}^{(m,k)}\Lambda_k^{\top} \right|\notag \\
&=(\sigma^2)^{pq}\left| I_{r} + \frac{1}{\sigma^2}\Lambda_k^{\top}\Lambda_k V_{t|t-1}^{(m,k)}\right| \notag\\
&=(\sigma^2)^{pq}\mathop{\prod }\limits_{s = 1}^{k_1}\mathop{\prod }\limits_{l = 1}^{k_2}(1+\frac{1}{\sigma^2}d_{sl}^v).
\end{align}
Next, Consider \(\eta_{ t|t-1 }^{(m,k)\top}(\mathop{\Sigma }_{t \mid t-1}^{(m,k)})^{-1}{\eta}_{t|t-1}^{(m,k)}\). For convenience, we will omit the superscripts $(m, k)$  in the symbols. Define \(V = V_{t|t-1}^{(m,k)}\).   Since \(\Sigma _{t |t-1}^{-1} = \frac{1}{\sigma^2}I_{pq} - \frac{1}{\sigma^2}\Lambda_k{\left\lbrack  \sigma^2 V^{-1} + \Lambda_k^{\top}\Lambda _k\right\rbrack}^{-1}\Lambda_k^{\top}\) , then
\begin{align*}
&\eta_{t|t-1}^{\top}({\Sigma}_{t|t-1})^{-1}\eta_{t|t-1}\notag\\
=& (y_t - \Lambda_k f_{t|t-1})^{\top}\left( \frac{1}{\sigma^2}I_{pq} - \frac{1}{\sigma^2}\Lambda_k{\left\lbrack  \sigma^2 V^{-1} + \Lambda_k^{\top}\Lambda_k\right\rbrack  }^{-1}\Lambda_k^{\top }\right) (y_t - \Lambda_k f_{t|t - 1}) \\
=&\frac{1}{\sigma^2}y_t^{\top}y_t - \frac{2}{\sigma^2}y_t^{\top }\Lambda_k f_{t|t-1} + \frac{1}{\sigma^2}f_{t|t-1}^{\top }\Lambda_k^{\top}\Lambda_k f_{t|t-1} \\
-&\frac{1}{\sigma^2}y_t^{\top}\Lambda_k{\left\lbrack \sigma^2 V^{-1} + \Lambda_k^{\top}\Lambda_k\right\rbrack}^{-1}\Lambda_k^{\top}y_t \\
+& \frac{2}{\sigma^2}y_t^{\top }\Lambda_k{\left\lbrack \sigma^2 V^{-1} + \Lambda_k^{\top}\Lambda_k\right\rbrack  }^{-1}\Lambda_k^{\top}\Lambda_k f_{t|t-1} \\
-&\frac{1}{\sigma^2}f_{t \mid t-1}^{\top}\Lambda_k^{\top}\Lambda_k{\left\lbrack \sigma^2 V^{-1} + \Lambda_k^{\top}\Lambda_k\right\rbrack  }^{-1}\Lambda_k^{\top}\Lambda_k f_{t \mid t-1} \\
\triangleq &A_1 - A_2 + A_3 - B_1 + B_2 - B_3,
\end{align*}
where
\begin{align*}
A_1 &= \frac{1}{\sigma^2}y_t^{\top}y_t = \frac{1}{\sigma^2} \text{Tr}(Y_t^{\top}Y_t), \\
A_2 &= \frac{2}{\sigma^2}y_t^{\top}\Lambda_k f_{t|t-1} = \frac{2}{\sigma^2}y_t^{\top}(C_k \otimes  R_k)f_{t|t-1} = \frac{2}{\sigma^2}y_t^{\top}\text{Vec}(R_k F_{t|t-1}C_k^{\top}) \\
 &= \frac{2}{\sigma^2}\text{Tr}(Y_t^{\top}R_k F_{t| t-1}C_k^{\top}), \\
A_3 &= \frac{1}{\sigma^2}f_{t|t-1}^{\top}\Lambda_k^{\top}\Lambda_k f_{t|t-1} = \frac{1}{\sigma^2}f_{t|t-1}^{\top}(C_k^{\top} \otimes R_k^{\top})(C_k \otimes R_k)f_{t|t-1} \\
&= \frac{1}{\sigma^2}\text{Vec}^{\top}(R_k F_{t|t-1}C_k^{\top})\text{Vec}(R_k F_{t|t-1}C_k^{\top}) \\
&= \frac{1}{\sigma^2}T_r\left\lbrack C_k F_{t|t-1}^{\top}R_k^{\top}R_k F_{t|t- 1}C_k^{\top}\right\rbrack,
\end{align*}
and
\begin{align*}
B_1 &=\frac{1}{\sigma^2}y_t^{\top}\Lambda_k{\left\lbrack \sigma^2 V^{-1} + \Lambda_k^{\top}\Lambda_k\right\rbrack}^{-1}\Lambda_k^{\top}y_t \\
&= \frac{1}{\sigma ^2}y_t^{\top}\Lambda_k(\Lambda_k^{\top}\Lambda_k)^{-1}{\left\lbrack \sigma^2 V^{-1}(\Lambda_k^{\top}\Lambda_k)^{-1}+ I_r\right\rbrack}^{-1}\Lambda_k^{\top}y_t \\
&= \frac{1}{\sigma^2}y_t^{\top}\Lambda_k \cdot \frac{1}{\sigma^2}V_{t|t}^{(m,k)}\Lambda_k^{\top}y_t~~~~~~( \text{ See equation (\ref{A1.3}) }) \\
&= \frac{1}{\sigma^4}\text{Vec}(R_k^{\top}Y_t C_k)V_{t|t}^{(m,k)}\text{Vec}(R_k^{\top}Y_t C_k),
\end{align*}
\begin{align*}
B_2 &= \frac{2}{\sigma^2}y_t^{\top}\Lambda _k{\left\lbrack \sigma^2 V^{-1} + \Lambda_k^{\top}\Lambda_k\right\rbrack }^{-1}\Lambda_k^{\top}\Lambda_k f_{t|t- 1} \\
&=\frac{2}{\sigma^2}y_t^{\top}\Lambda_k  \frac{1}{\sigma^2}V_{t|t}^{(m,k) }\Lambda_k^{\top}\Lambda_k f_{t|t - 1} \\
&=\frac{2}{\sigma^4}\text{Vec}^{\top}(R_k^{\top}Y_t C_k)V_{t|t}^{(m,k)}\Lambda_k^{\top}\Lambda_k f_{t|t-1}, \\
B_3 &=\frac{1}{\sigma^2}f_{t \mid t-1}^{\top}\Lambda_k^{\top}\Lambda_k{\left\lbrack \sigma^2 V^{-1} + \Lambda_k^{\top}\Lambda_k\right\rbrack  }^{-1}\Lambda_k^{\top}\Lambda_k f_{t \mid t-1} \\
&= \frac{1}{\sigma^2}f_{t|t- 1}^{\top}{\Lambda }_{k}^{\top }\Lambda_k(\Lambda_k^{\top}\Lambda_k)^{-1}{\left\lbrack \sigma^2 V^{-1}(\Lambda_k^{\top}\Lambda_k)^{-1} + I_r\right\rbrack}^{-1}\Lambda_k^{\top }\Lambda_k f_{t|t-1} \\
&= \frac{1}{\sigma^2}f_{t|t-1}^{\top }\Lambda_k^{\top}\Lambda_k\frac{1}{\sigma^2}V_{t|t}^{(m,k)}\Lambda_k^{\top}\Lambda_kf_{t|t - 1}.
\end{align*}
With the above results, we finally get the results of Lemma \ref{lemaA1}.

\end{proof}

\subsection{ The filtering and smoothing algorithm }
Algorithm 2 summarizes the filtering and smoothing process for
$f_t$ and $s_t$, as detailed in Section 3.2 of the main text.

\begin{algorithm}
\caption{The  Filtering  and  Smoothing Algorithm}
\begin{algorithmic}
\scriptsize
\STATE
\label{filter}

\begin{description}
\item[Input:]   Data $(Y_t)_{t=1}^n$,  and parameter $ ~\theta$.
\item[Output:] $\{f_{t|n}^{(k)}, V_{t|n}^{(k)}, w_{t|n}^{(k)}, w_{t-1,t|n}^{(i,k)}:  t\in[n], i, k\in[K]\}$.
\item[Initialization:]  For $i,k \in [K],~f_{0|0}^{(i)}=0, ~V_{0|0}^{(i)}=\textbf{0},~ w_{-1,0|0}^{(i,k)}=0.$
\item[Filtering:] ~
\begin{description}
\item[Conditional Kalman filtering of $f_t$ $(t = 1,  \cdots, n)$ ]~

    \(	{f}_{t \mid  t - 1}^{( i,k) } = {\beta }_{k} + \left( {{\Gamma }_{k} \otimes  {\Phi }_{k}}\right) f_{t - 1 \mid  t - 1}^{( i) }\)

\(	{V}_{t |t- 1}^{( i,k) } = \left( {{\Gamma }_{k}\otimes  {\Phi }_{k}}\right) V_{t -1|t- 1}^{( i) }\left( {{\Gamma }_{k} \otimes  {\Phi }_{k}}\right)  + {\sigma }_{\epsilon}^{2}I_r\)

\(	{V}_{t \mid  t}^{( i,k) } = {\left\lbrack  I_r + \frac{1}{{\sigma }^{2}}V_{t \mid  t - 1}^{( i,k) }{\Lambda }_{k}^{\top }{\Lambda }_{k}\right\rbrack  }^{-1}V_{t \mid  t - 1}^{( i,k ) }\)

\(  {f}_{t \mid t}^{( i,k) } = f_{t \mid t - 1}^{( i,k) } + \frac{1}{{\sigma }^{2}}V_{t \mid t}^{( i,k) }\left\lbrack  {\text{vec}\left( {R_{k}^{\top }Y_t C_k}\right)  - {\Lambda }_k^{\top}{\Lambda }_kf_{t \mid t - 1}^{(i,k) }
}\right\rbrack\)

\item[Filtering process of $s_t$ $(t = 1,  \cdots, n)$ ]~

\(	{w}_{t - 1,t \mid t - 1}^{( i,k) } =p_{ik}\mathop\sum \limits_{{{i'} = 1}}^Mw_{t - 2,t - 1 \mid  t - 1}^{( {i'},i) }\)

\(	{w}_{t \mid  t - 1}^{( k) } = \mathop\sum \limits_{i = 1}^Mw_{t - 1,t \mid  t - 1}^{( i,k) }\)

\(	{fy}_{t-1,t \mid t- 1}^{( i,k) } = {( 2\pi ) }^{-\frac{pq}{2}}{\left( {\sigma}^{2}\right) }^{-\frac{pq}{2}}\mathop\prod \limits_{s = 1}^{k_1}\mathop\prod \limits_{\ell = 1}^{k_2}\left( {1 +{ \frac{1}{{\sigma}^{2}}d_{se}^{v}}}\right) \exp  ^{-\frac{1}{2}\left( {A_1 - A_2 + A_3 - B_1 + B_2 - B_3}\right) },\)
where \(d_{se}^v\), \({A}_1,A_2,A_3,B_1,B_2\) and \(B_3\) are given in Section A.3 of the Supplementary.

\({w}_{t - 1,t \mid  t}^{( i,k) } = \frac{f{y}_{t - 1,t \mid  t - 1}^{( i,k) } {w}_{t - 1,t \mid t - 1}^{( i,k) }}{\mathop\sum\limits_{{i' = 1}}^M\mathop\sum \limits_{{k' = 1}}^Mf{y}_{t - 1,t \mid  t - 1}^{(i',k') } {w}_{t - 1,t \mid t - 1}^{( i',k') }}
\)

\(
{w}_{t \mid t}^{( k) } = \mathop\sum \limits_{{i = 1}}^Mw_{t - 1,t \mid t}^{( i,k) }
\)

 \item[Collapsed filtering of $f_t$ $(t = 1,  \cdots, n)$]~

\(
{f}_{t \mid  t}^{( k) } = \mathop\sum \limits_{{i = 1}}^Mw_{t-1,t \mid  t}^{( i,k) }f_{t \mid  t}^{( i,k) }/{w}_{t \mid  t}^{( k) }
\)

\(
{V}_{t \mid  t}^{( k) } = \mathop\sum \limits_{{i = 1}}^Mw_{t-1,t \mid  t}^{( i,k) }{\left[  V_{t \mid t}^{( i,k) } +\left( f_{t \mid  t}^{( k) }-f_{t \mid  t}^{( i,k) }\right)\left( f_{t \mid  t}^{( k) }-{f}_{t \mid  t}^{( i,k) }\right)^{\top} \right]  } /w_{t \mid  t}^{( k) }
\)
\end{description}
\item[Smoothing:] ~
\begin{description}
\item[Smoothing of $s_t$ ( $t = n - 1, n - 2, \cdots, 1$)]~

    \(w_{t,t + 1 \mid  n}^{(j,k)} = {w_{t + 1 \mid  n}^{(k)}w_{t \mid  t}^{(j)}p_{jk}}/{w_{t + 1 \mid  t}^{(k)}}, \text{where}~~ w_{t + 1 \mid  t}^{(k)} = \mathop\sum \limits_{{i = 1}}^Mw_{t,t + 1 \mid  t}^{(i,k)} \)

\(w_{t | n}^{(j)} = \mathop\sum \limits_{k = 1}^Mw_{t,t + 1 \mid  n}^{(j,k)}\)

\item[Smoothing of $f_t$ ( $t = n - 1, n - 2, \cdots, 1$)]~

 \(f_{t \mid  n}^{( j,k) } = f_{t \mid  t}^{(j) } + G_{t}^{( j,k) }\left( {f_{t + 1 \mid  n}^{(k)} - f_{t +1 \mid  t}^{( j,k) }}\right) , \text{where}~ G_{t}^{( j,k) } = V_{t \mid  t}^{(j)}(\Gamma_k \otimes \Phi_k){\left(  V_{t +1 \mid  t}^{( j,k) }\right)^{-1}  }\)

\(
V_{t | n}^{( j,k) } = V_{t | t}^{(j)} + G_{t}^{( j,k) }\left(  {V_{t + 1|n}^{(k)} - V_{t + 1|t}^{( j,k) }}\right)  {G_{t}^{( j,k) }}^{\top }
\)

\(
f_{t \mid  n}^{(j)} = {\mathop\sum \limits_{k = 1}^Mw_{t,t + 1 \mid  n}^{(j,k)}f_{t \mid  n}^{( j,k) }}/{w_{t \mid  n}^{(j)}}
\)

\(
V_{t | n}^{(j)} = {\mathop\sum \limits_{{k = 1}}^Mw_{t,t + 1 \mid  n}^{(j, k) }\left\{ V_{t | n}^{( j,k) } +\left( f_{t \mid  n}^{(j)}-f_{t \mid  n}^{( j,k) }\right)\left( f_{t \mid  n}^{(j)}-f_{t \mid  n}^{( j,k) }\right)^{\top } \right\}  }/{w_{t | n}^{(j)}}
\)
\end{description}
 \end{description}
\end{algorithmic}
\end{algorithm}

\subsection{  The detailed calculation of the elements in Table 1 }

(1)  Consider $f_{t-1|n}^{*(k)} $.
\begin{align*}
f_{t - 1|n}^{*( k) } &= \mathbb{E}\left\lbrack  {f_{t - 1} \mid  \mathcal{Y}_n,s_t = k;\theta}\right\rbrack \\
&= \mathop\sum \limits_{i = 1}^M \mathbb{E} \left\lbrack  {f_{t - 1} \mid  \mathcal{Y}_n,s_{t - 1} = i,s_t = k;\theta}\right\rbrack  \;P_r\left\lbrack  {s_{t - 1} = i \mid  \mathcal{Y}_n,s_t = k;\theta}\right\rbrack \\
&= \mathop\sum \limits_{i = 1}^Mf_{t - 1|n}^{( i) }\frac{w_{t -1, t | n}^{( i,k) }}{w_{t | n}^{( k) }},
\end{align*}
where \(f_{t -1 |n}^{( i) },w_{t -1, t|n}^{( i,k) }\) and \(w_{t|n}^{( k) }\) are given by the filtering and smoothing algorithm in Section \(A. 4\) .

(2) Consider $p_{t-1|n}^{*(k)}$.
\begin{align*}
p_{t - 1|n}^{*( k) } &= \mathbb{E}\left\lbrack  {f_{t - 1}f_{t - 1}^{\top } \mid \mathcal{Y}_n,s_t = k;\theta}\right\rbrack \\
&= \mathop\sum \limits_{i = 1}^ME\left\lbrack  {f_{t - 1}f_{t - 1}^{\top} \mid  \mathcal{Y}_n,s_{t - 1} = i,s_t = k;\theta}\right\rbrack  P_r\left( {s_{t - 1} = i \mid  \mathcal{Y}_n,s_t  = k;\theta}\right) \\
&= \mathop\sum \limits_{i = 1}^Mp_{t - 1|n}^{( i) }   \frac{w_{t - 1,t|  n}^{( i,k) }}{w_{t | n}^{( k) }} = \mathop\sum \limits_{i = 1}^M\left\lbrack  {V_{t -1|n}^{( i) } + f_{t-1|n}^{( i) }f_{t - 1|n}^{( i) \top} }\right\rbrack    \frac{w_{t -1, t |n}^{( i,k) }}{w_{t |n}^{( k) }},
\end{align*}
where \(V_{t - 1|n}^{( i) },f_{t - 1|n}^{( i) },w_{t - 1,t|n}^{( i,k) }\) and \(w_{t |n}^{( k ) }\) are given by the filtering and smoothing algorithm in Section A.4.

(3) Consider $p_{t,t-1|n}^{(k)}$.
\begin{align*}
p_{t,t - 1 \mid n}^{( k) } &= \mathbb{E}\left\lbrack f_tf_{t - 1}^{\top} \mid \mathcal{Y}_n,s_t = k;\theta\right\rbrack \\
&= {\operatorname{Cov}\left\lbrack f_t,~f_{t - 1} \mid \mathcal{Y}_n,s_t=k;\theta\right\rbrack} + f_{t |n}^{( k)}f_{t - 1 \mid n}^{*( k)^{\top}} \\
&= {\operatorname{Cov}\left\lbrack {\beta}_{k} + \left( {\Gamma}_{k}\otimes{\Phi}_{k}\right) f_{t - 1} + {\varepsilon}_{t},~f_{t - 1} |\mathcal{Y}_n,s_t = k;\theta\right\rbrack} + f_{t | n}^{( k)}f_{t - 1|n}^{*(k)\top} \\
&= \left( {\Gamma}_{k}\otimes{\Phi}_{k}\right)\operatorname{Cov}\left\lbrack f_{t - 1},~f_{t - 1} \mid \mathcal{Y}_n,s_t=k;\theta\right\rbrack + f_{t | n}^{( k)}f_{t - 1|n}^{*(k)\top} \\
&= \left( {\Gamma}_{k} \otimes {\Phi}_{k}\right) \left\lbrack p_{t - 1|n}^{*( k)} - f_{t - 1|n}^{*( k)}f_{t - 1|n}^{*( k) \top}\right\rbrack + f_{t |n}^{( k)}f_{t - 1|n}^{*( k)\top},
\end{align*}
where
\(p_{t - 1|n}^{*( k) }\)
  and
\(f_{t - 1|n}^{*( k) }\)
  are defined in results (2) and (1), respectively, and
 \(f_{t |n}^{( k) }\)
  is given by the filtering and smoothing algorithm in Section A.4.
\subsection{  Initialization }
\begin{enumerate}
    \item Partition $\{1,\cdots,n\}$ into intervals $\mathcal{T}_1,\cdots, \mathcal{T}_a$ of approximately equal length for some $a>0$. For $l = 1,..., a$, fit a matrix factor model to $(Y_t)_{t \in \mathcal{T}_l}$ using the QMLE estimation method \citep{2024Quasi}. Let $\widehat{R}^{(l)}, \widehat{C}^{(l)}$ and $\widehat F_{t}^{(l)}$ be the estimated loading matrices and the factor score matrix.
    \item  Calculate the distance matrices as follows: $D_R=\left(D(R^{(i)}, R^{(j)})\right)_{i,j=1}^a$, where $D(\cdot,\cdot)$ is defined in \cite{2021Projected}, which is used to measure the distance of two matrices. Similarly,  $D_C=\left(D(C^{(i)}, C^{(j)})\right)_{i,j=1}^a$. Let $\mu^{(l)}$ and $\Sigma^{(l)}$ be the sample mean and sample covariance of $(f_t^{(l)})_{t\in \mathcal{T}_l}, l=1,\cdots,a$.
    \item  Partition the $a$ intervals into $M$ clusters by hierarchical clustering based on the distance matrix $D_{R}$, and let $c_l^R$ denote the cluster label of $R^{(l)}$.  Similarly,    perform hierarchical clustering in $D_{C}$ to partition $C^{(1)},\dots,\widehat C^{(a)}$ into $M$ clusters and denote by $c_l^C$ the clustering label of $C^{(l)}$.   In parallel, apply $k$-means to $\{\mu^{(l)},\Sigma^{(l)}\}_{l=1}^a$ and let $c_l^F$ be the clustering label of $(\mu^{(l)},\Sigma^{(l)})$.

\item  if a given interval $\mathcal{T}_l$ is assigned to the same cluster by at least two
of these three labels $(c_l^R, c_l^C, c_l^F )$ , it is assigned to this cluster and denotes the cluster label as $c_l$. For $t\in \mathcal{T}_l$, let $s_t^{(0)} = c_l$.

\item  For each $k=1,\cdots,M$, fitting the matrix factor model on $\{Y_t: t\in \cup _{l=1}^a \mathcal{T}_lI_{(c_l=k)} \}$. Then get the initial value of the loading matrices $R_k^{(0)}, C_k^{(0)}$ and the factor scores $F_t^{(0)}$. The initial values of $\sigma^{(0)2} = \sum_t\|Y_t- R_{s_t^{(0)}}^{(0)} F_t^{(0)} C_{s_t^{(0)}}^{(0)'}\|_F^2/(pqn)$.
Let $\theta_1^{(0)}=\{(R_k^{(0)}, C_k^{(0)} ): k \in [K], \sigma^{(0)2}\}$.

\item   For each $k=1,\cdots,M$, fitting the VAR(1) model on $\{f_t: t\in \cup _{l=1}^a \mathcal{T}_lI_{(c_l=k)} \}$. Let $\beta_k^{(0)}$ and $\Psi_k^{(0)}$  be the intercept vector and autocoefficient matrix, respectively.

\item  Finally, using the Kronecker product approximation on $\Psi_k^{(0)}$ via singular value decomposition \citep{1993Linear} to get $\Phi_k^{(0)}$ and $\Gamma_k^{(0)}$.

\end{enumerate}

	
	


 \setcounter{equation}{0}
 \setcounter{subsection}{0}
 \renewcommand{\theequation}{B.\arabic{equation}}
 \renewcommand{\thesubsection}{B.\arabic{subsection}}
\section*{B ~~ Details for Theorem 1}
\begin{proof}

Step (1): The log likelihood function of \({\mathcal{Y}}_{n}\) is given by
\begin{align}
\ell\left( \theta \right)  & = \log \left\lbrack  {\mathop{\sum }\limits_{{{s}_{0} = 1}}^{M}\cdots \mathop{\sum }\limits_{{{s}_{n} = 1}}^{M}p\left( {{y}_{1},\cdots ,{y}_{n} \mid  {s}_{n},\cdots ,{s}_{0};\theta }\right) {P}_{r}\left( {{s}_{n},\cdots ,{s}_{0};\theta }\right) }\right\rbrack
 \\
 &= \log \left\lbrack  {\mathop{\sum }\limits_{{{s}_{0} = 1}}^{M}\cdots \mathop{\sum }\limits_{{{s}_{n} = 1}}^{M}\mathop{\prod }\limits_{{t = 1}}^{n}p\left( {{y}_{t} \mid  \mathcal{Y}_{t - 1},\mathcal{S}_{n};{\theta }_{1},{\theta }_{2}}\right) } \mathop{\prod }\limits_{{t = 1}}^{n}{P}_{r}\left( {{s}_{t} \mid  {s}_{t - 1};P}\right) {P}_{r}\left( {{s}_{0};{P}_{}}\right)\right\rbrack  \notag
\end{align}
where \(p\left( {{y}_{t} \mid  \mathcal{Y}_{t - 1},\mathcal{S}_{n};{\theta }_{1},{\theta }_{2}}\right)  = {\left( 2\pi \right) }^{-\frac{pq}{2}}{\left| {\Sigma }_{t \mid  t - 1}\right| }^{-\frac{1}{2}}{e}^{-\frac{1}{2}{\left( {y}_{t} - {y}_{t \mid  t - 1}\right) }^{\top}{\Sigma }_{t \mid  t - 1}^{-1}\left( {{y}_{t} - {y}_{t \mid  t - 1}}\right) ,{}}\) and
\begin{align}\label{Th1_1}
    {y}_{t \mid  t - 1} & \triangleq  E\left\lbrack  {{y}_{t} \mid  \mathcal{Y}_{t - 1},\mathcal{S}_{n};{\theta }_{1},{\theta }_{2}}\right\rbrack   = \left( {{C}_{s_t} \otimes  {R}_{s_t}}\right)\underbrace{ E\left\lbrack  {{f}_{t} \mid  \mathcal{Y}_{t - 1},\mathcal{S}_{n};{\theta }_{1},{\theta }_{2}}\right\rbrack }_{:= f_{t|t-1}},\notag\\
    {\Sigma }_{t \mid  t - 1} &\triangleq  \operatorname{Cov}\left\lbrack  {{y}_{t} \mid  \mathcal{Y}_{t - 1},\mathcal{S}_{n};{\theta }_{1},{\theta }_{2}}\right\rbrack \notag\\
    &= \left( {{C}_{s_t} \otimes  {R}_{s_t}}\right)
 \underbrace{\text{Cov}\left\lbrack  {{f}_{t} \mid  \mathcal{Y}_{t - 1},\mathcal{S}_{n};{\theta }_{1},{\theta }_{2}}\right\rbrack}_{:=V_{t|t-1}}  {\left( {C}_{s_t} \otimes  {R}_{s_t}\right) }^{\top } + {\sigma}^{2}{I}_{pq},
\end{align}
 \({f}_{t \mid  t - 1}\) and \({V}_{t\mid {t} - 1}\) can be given by the kalman filter recursions (See Section A of the supplement).

In the following, define \({G}_{t} = \left\{  {\mathcal{Y}_{t},\mathcal{S}_{t}}\right\}\). Let
$${m}_{t} = \arg \mathop{\max}\limits_{{j}} p\left( {{y}_{t} \mid  {G}_{t - 1},{s}_{t} = j;{\widehat{\theta }}_{1},{\widehat{\theta }}_{2}}\right), $$
that is, \(p\left( {{y}_{t} \mid  {G}_{t - 1},{s}_{t} = j;{\widehat{\theta}}_{1},{\widehat{\theta }}_{2}}\right)\) take the maximum when \(j = m_{t}\). Under the model setting, we have
\(p\left( {{y}_{t} \mid  \mathcal{Y}_{t - 1},\mathcal{S}_{n};{\theta }_{1},{\theta }_{2}}\right)  = p\left( {{y}_{t} \mid  {G}_{t - 1},{s}_{t};{\theta }_{1},{\theta }_{2}}\right).\)
Since \(\mathop{\sum }\limits_{{{s}_{t} = 1}}^{M}{P}_{r}\left( {{s}_{t} \mid  {s}_{t - 1};P}\right)  = 1\), for any \({s}_{t - 1}\),
$$\mathop{\sum }\limits_{{s_{t} = 1}}^{M}p\left( {{y}_{t} \mid  {G}_{t - 1},{s}_{t};\widehat{{\theta }}_{1},\widehat{{\theta} }_{2}}\right) {P}_{r}\left( {{s}_{t} \mid  {s}_{t - 1};P}\right) \le p\left( {{y}_{t} \mid  {G}_{t - 1},{s}_{t} = {m}_{t};\widehat{{\theta }}_{1},\widehat{{\theta} }_{2}}\right), $$
then
\begin{align}\label{Th1_0}
    l\left( {\widehat{{\theta }}_{1},\widehat{{\theta }}_{2},P}\right) & = \log \left\lbrack{\mathop{\sum }\limits_{{{s}_{0} = 1}}^{M} \cdots \mathop{\sum }\limits_{{{s}_{n - 1} = 1}}^{M}\mathop{\prod }\limits_{{t = 1}}^{{n - 1}}p\left( {{y}_{t} \mid  {G}_{t - 1},{s}_{t};\widehat{{\theta }}_{1},\widehat{{\theta }}_{2}}\right) \mathop{\prod }\limits_{{t = 1}}^{n-1}{P}_{r}\left( {{s}_{t} \mid  {s}_{t - 1};P}\right)  } \right.  \notag\\
   & ~~~~~{P}_{r}\left( {{s}_{0};P}\right) \left. \left( {\mathop{\sum }\limits_{{{s}_{n} = 1}}^{M}p\left( {{y}_{n} \mid  {G}_{n - 1},{s}_{n};{\widehat{\theta} }_{1},{\widehat{\theta }}_{2}}\right) {P}_{r}\left( {{s}_{n} \mid  {S}_{n-1};P}\right) }\right) \right\rbrack \notag\\
   &\leq  \log \left\lbrack { \mathop{\sum }\limits_{{{s}_{0} = 1}}^{M}\cdots \mathop{\sum }\limits_{{{s}_{n - 1} = 1}}^{M}\mathop{\prod }\limits_{{t = 1}}^{{n - 1}}p\left( {{y}_{t} \mid  {G}_{t - 1},\mathcal{S}_{t};\widehat{{\theta }}_{1},\widehat{{\theta }}_{2}}\right) \mathop{\prod }\limits_{{t = 1}}^{{n - 1}}P_{r}\left( {{s}_{t} \mid  {s}_{t - 1};P}\right) } \right.\notag\\
& ~~~~~   P_{r}\left({s_{0};P}\right)p  \left.\left( {{y}_{n} \mid  {G}_{t - 1},{s}_{n} = {m}_{n};\widehat{{\theta }}_{1},\widehat{{\theta} }_{2}}\right) \right\rbrack\notag\\
& \leq  \cdots  \leq  \mathop{\sum }\limits_{{t = 1}}^{n}\log p\left( {{y}_{t} \mid  {G}_{t - 1},{s}_{t} = {m}_{t};\widehat{{\theta} }_{1},\widehat{{\theta} }_{2}}\right) .
\end{align}
It follows that
\begin{align}\label{Th1_11}
    \ell\left( {\widehat{{\theta }}_{1},\widehat{{\theta} }_{2},P}\right)   & \leq  \mathop{\sum }\limits_{{t = 1}}^{n}\log \left\lbrack  {{\left( 2\pi \right) }^{-\frac{pq}{2}}{\left| {\widehat{\Sigma }}_{t \mid  t - 1}\right| }^{-\frac{1}{2}}{e}^{-\frac{1}{2}{\left({y}_{t} - {\widehat{y}}_{t \mid  t - 1}\right)}^{\top}{\widehat{\Sigma }}_{t \mid  t - 1}^{-1}\left( {{y}_{t} - {\widehat{y}}_{t \mid  t - 1}}\right) }}\right\rbrack \notag\\
    & =  - \frac{n{pq}}{2}\log \left( {2\pi }\right)  - \frac{1}{2}\mathop{\sum }\limits_{{t = 1}}^{n}\log \left| {\widehat{\Sigma }}_{t \mid  t - 1}\right|  - \frac{1}{2}\mathop{\sum }\limits_{{t = 1}}^{n}{{\left( {y}_{t} - {\widehat{y}}_{t \mid  t - 1}\right)}}^{\top}\widehat{{\Sigma }}_{t \mid  t - 1}^{-1} \notag\\
    &~~~~
    \left( {{y}_{t} - {\widehat{y}}_{t - 1}}\right),
\end{align}
where \({\widehat{y}}_{t \mid  t - 1} = {\widehat{\Lambda }}_{m_t}{\widehat{f}}_{t \mid  t - 1}\) and \({\widehat{\Sigma }}_{t \mid  t - 1} = {\widehat{\Lambda }}_{m_t}{\widehat{V}}_{t \mid  t - 1}{{\widehat{\Lambda }}_{m_t}^{\top}} + {\widehat{\sigma }}^{2}{I}_{pq}, {\widehat{\Lambda }}_{m_t} = {\widehat{C}}_{m_t} \otimes  {\widehat{R}}_{m_t}\).
Consider the last term on the right hand side of equation (\ref{Th1_11}). By woodburg identity,
\({\left( {\widehat{\Lambda }}_{m_t}{\widehat{V}}_{t \mid t- 1}{{\widehat{\Lambda }}_{m_t}^{\prime}} + {\widehat{\sigma }}^{2}{I}_{pq}\right) }^{-1} = {\widehat{\sigma }}^{-2}{I}_{pq} - {\widehat{\sigma }}^{-2}{\widehat{\Lambda }}_{m_t}{\widehat{V}}_{t \mid t - 1}^{\frac{1}{2}}\left( {{\widehat{\sigma }}^{2}{I}_{r} + {\widehat{V}}_{t \mid t - 1}^{\frac{1}{2}}{{\widehat{\Lambda }}_{m_t}^{\prime}}{\widehat{\Lambda }}_{m_t}{\widehat{V}}_{t \mid  t-1}^{\frac{1}{2}}}\right)^{-1}  {\widehat{V}}_{t\mid t-1}^{\frac{1}{2}}{\widehat{\Lambda}}_{m_t}\).
Define \({\widetilde{\Lambda }}_{m_t} = {\widehat{\Lambda }}_{m_t}{\widehat{V}}_{t \mid t-1}^{\frac{1}{2}}\), then
\begin{align}\label{Th1_2}
    \mathop{\sum }\limits_{{t = 1}}^{n}\left( {{y}_{t} - {\widehat{y}}_{t \mid  t - 1}}\right) {}^{\prime }{\widehat{\Sigma}}_{t \mid  t - 1}^{-1}\left( {{y}_{t} - {\widehat{y}}_{t \mid  t - 1}}\right)  = {\widehat{\sigma }}^{-2}\mathop{\sum }\limits_{{t = 1}}^{n}{\left( {y}_{t} - {\widehat{y}}_{t \mid  t - 1}\right) }^{\prime }\left( {{y}_{t} - {\widehat{y}}_{t \mid  t - 1}}\right) \notag\\
    - {\widehat{\sigma }}^{-2}\mathop{\sum }\limits_{{t = 1}}^{n}{\left( {y}_{t} - {\widehat{y}}_{t \mid t -1}\right) }^{\prime }{\widetilde{\Lambda }}_{m_t}{\left( {\widehat{\sigma} }^{2}{I}_{r} + {\widetilde{\Lambda }}_{m_t}^{\prime }{\widetilde{\Lambda}}_{m_t}\right) }^{-1}{\widetilde{\Lambda }}_{m_t}^{\prime }\left( {{y}_{t} - {\widehat{y}}_{t \mid t - 1}}\right).
\end{align}
Since \({\left( {\widehat{\sigma }}^{2}{I}_{r} + {\widetilde{\Lambda }}_{m_t}^{\prime }{\widetilde{\Lambda }}_{m_t}\right) }^{-1} = {\left( {\widetilde{\Lambda }}_{m_t}^{\prime }{\widetilde{\Lambda }}_{m_t}\right) }^{-1} - {\widehat{\sigma}}^{2}{\left( {\widehat{\sigma}}^{2}{I}_{r} + {\widetilde{\Lambda }}_{m_t}^{\prime }{\widetilde{\Lambda }}_{m_t}\right) }^{-1}{\left( {\widetilde{\Lambda }}_{m_t}^{\prime }{\widetilde{\Lambda }}_{m_t}\right) }^{-1},\) then (\ref{Th1_2}) equals to
\begin{align}\label{Th1_2_1}
  & \mathop{\sum }\limits_{{t = 1}}^{n}{\left( {y}_{t} - {\widehat{y}}_{t \mid t- 1}\right) }^{\prime }\widehat{\mathop{\Sigma }}_{{t \mid t- 1 }}^{-1}\left( {{y}_{t} - {\widehat{y}}_{t \mid t- 1}}\right)  \notag\\
  =&
   \widehat{\sigma}^{-2}\mathop{\sum }\limits_{{t = 1}}^{n}{\left( {y}_{t} - {\widehat{y}}_{t \mid  t - 1}\right) }^{\prime }\left( {{y}_{t} - {\widehat{y}}_{t \mid  t - 1}}\right)  - {\widehat{\sigma }}^{-2}\mathop{\sum }\limits_{{t = 1}}^{n}{\left( {y}_{t} - {\widehat{y}}_{t \mid  t - 1}\right) }^{\prime }{P}_{\widetilde{{\Lambda}}_{m_t}}\left( {{y}_{t} - {\widehat{y}}_{t \mid  t - 1}}\right) +I_{\Lambda} \notag\\
   =&  \widehat{{\sigma }}^{-2}\mathop{\sum }\limits_{{t = 1}}^{n}{\left( {y}_{t} - {\widehat{y}}_{t \mid  t - 1}\right) }^{\prime }{M}_{{\widetilde{\Lambda }}_{m_t}}\left( {{y}_{t} - {\widehat{y}}_{t \mid  t - 1}}\right)  + {I}_{\Lambda} \notag\\
   =& \widehat{ \sigma}^{-2} \mathop{\sum }\limits_{{t = 1}}^{n}{\begin{Vmatrix}{M}_{{\widetilde{\Lambda }}_{m_t}}\left( {y}_{t} - {\widehat{y}}_{t \mid t- 1}\right) \end{Vmatrix}}_{F}^{2} + {I}_{\Lambda},
\end{align}
where $I_{\Lambda}= \mathop{\sum }\limits_{{t = 1}}^{n}{\left( {y}_{t} - {\widehat{y}}_{t \mid  t - 1}\right) }^{\prime }{\widetilde{\Lambda }}_{m_t}{\left( {\widehat{\sigma }}^{2}{I}_{r} + {\widetilde{\Lambda }}_{m_t}^{\prime }{\widetilde{\Lambda }}_{m_t}\right) }^{-1}\left( {{\widetilde{\Lambda }}_{m_t}^{\prime }{\widetilde{\Lambda }}_{m_t}}\right)^{-1} {\widetilde{\Lambda }}_{m_t}^{\prime }\left( {{y}_{t} - {\widehat{y}}_{t \mid  t - 1}}\right)$.

Step(2) : Consider \(\ell\left( {{R}^{\circ},{C}^{\circ},{\widehat{\sigma }}^{2},{\widehat{\theta }}_{2},P}\right)\), where $ {R}^{\circ}$ and ${C}^{\circ}$ denote the true value of $R \triangleq
 \{R_1, \cdots,R_M\}$ and $C \triangleq
 \{C_1, \cdots,C_M\}$.

Since
$
{P}_{r}\left( {{s}_{t} \mid  {s}_{t - 1};P}\right)  \geq  \mathop{\min }\limits_{{j,k}}{p}_{jk}
$, then
\begin{align}
  &  \mathop{\sum }\limits_{{{s}_{t} = 1}}^{M}p\left( {{y}_{t}|{G}_{t - 1},{s}_{t};{R}^{\circ},{C}^{\circ},{\widehat{\sigma }}^{2},{\widehat{\theta }}_{2}}\right) {P}_{r}\left( {{s}_{t}|{s}_{t - 1};P}\right) \notag\\
   \geq&  p\left( {{y}_{t}|{G}_{t - 1},{s}_{t};{R}^{\circ},{C}^{\circ},{\widehat{\sigma }}^{2},\widehat{{\theta }}_{2}}\right) \mathop{\min }\limits_{{j,k}}{p}_{jk}.
\end{align}
Similar to inequality (\ref{Th1_0}),
\begin{align}
  \ell \left( {{R}^{\circ},{C}^{\circ},{\widehat{\sigma }^{2}},{\widehat{\theta }}_{2},P}\right)  & \geq  \mathop{\sum }\limits_{{t = 1}}^{n}\log \left\lbrack  {p\left( {{y}_{t} \mid  {G}_{t - 1},{s}_{t};{R}^{\circ},{C}^{\circ},{\widehat{\sigma }}^{2},\widehat{{\theta }}_{2}}\right) \mathop{\min }\limits_{{j,k}}\{ {p}_{jk}\} }\right\rbrack   \notag\\
  & = n\log \mathop{\min }\limits_{{j,k}}{p}_{jk} - \frac{pqn}{2}\log \left( {2\pi }\right)  - \frac{1}{2}\mathop{\sum }\limits_{{t = 1}}^{n}\log \left| {{\Lambda }_{s_t}^{\circ}{\widehat{V}}_{t \mid t- 1}^{\circ}{{\Lambda }_{s_t}^{\circ\prime}}+ {\widehat{\sigma }}^{2}{I}_{pq}}\right| \notag\\
  &- \frac{1}{2}\mathop{\sum }\limits_{{t = 1}}^{n}{\left( {y}_{t} - {\widehat{y}}_{t \mid  t - 1}^{\circ}\right) }^{\prime }\left( {{\Lambda }_{s_t}^{\circ}{\widehat{V}}_{t \mid  t - 1}^{\circ}{{\Lambda }_{s_t}^{\circ\prime}} + \widehat{{\sigma}}^{2}{I}_{pq}}\right) {}^{-1}\left( {{y}_{t} - {\widehat{y}}_{t \mid  t - 1}^{\circ}}\right) ,
\end{align}
where the definition of $ \widehat{y}_{t \mid t- 1}^{\circ}$ and $ \widehat{V}_{t \mid t- 1}^{\circ}$ are the same as that of ${y}_{t \mid t- 1}$ and $ {V}_{t \mid t- 1}$  in equation (\ref{Th1_1}), with the only requirement being to replace the parameters with $ \{{R}^{\circ},{C}^{\circ},{\widehat{\sigma }^{2}},{\widehat{\theta }}_{2},P\}$ .

Similar to equation (\ref{Th1_2}), we have
\begin{align}
  &  \sum\limits_{t} {\left( {y}_{t} - {\widehat{y}}_{t \mid t-1}^{\circ}\right) }^{\prime }{\left( {\Lambda }_{s_t}^{\circ}\widehat{{V}}_{t \mid t- 1}^{\circ}{{\Lambda }_{{st}}^{\circ\prime}} + {\widehat{\sigma}}^{2}{I}_{pq}\right) }^{-1}\left( {{y}_{t} - {\widehat{y}}_{t \mid t- 1}^{\circ}}\right) \notag\\
  =&  {\widehat{\sigma}}^{-2}\mathop{\sum }\limits_{{t = 1}}^{n}\parallel {M}_{{\widetilde{\Lambda }}_{s_t}^{\circ}}( {{y}_{t} - {\widehat{y}}_{t \mid t- 1}^{\circ}}) {\parallel }_{F}^{2} + \mathop{\sum }\limits_{{t = 1}}^{n}{\left( {y}_{t} - {\widehat{y}}_{t \mid t- 1}^{\circ}\right) }^{\prime }{\widetilde{\Lambda }}_{s_t}^{\circ}{\left( {\widehat{\sigma}^{2}}I_{r} +{{{{\widetilde{\Lambda }}_{s_t}^{\circ\prime}}}}{\widetilde{\Lambda }}_{s_t}^{\circ}\right) }^{-1} \notag\\
  & {\left( {\widetilde{\Lambda }}_{s_t}^{\circ\prime}{\widetilde{\Lambda }}_{s_t}^{\circ}\right) }^{-1}{\widetilde{\Lambda}}_{s_t}^{\circ\prime}\left( {{y}_{t} - {\widehat{y}}_{t \mid  t - 1}^{\circ}}\right),
\end{align}
where \({{\widetilde{\Lambda}}_{s_t}^{\circ}} = {{\Lambda}_{s_t}^{\circ}}{\left( {\widehat{V}}_{t \mid t-1}^{\circ}\right) }^{\frac{1}{2}}\).

Step(3) : Consider \(\ell\left( {\widehat{R},\widehat{C},{\widehat{\sigma }}^{2},\widehat{\theta }_{2},P}\right)  - \ell\left( {{R}^{ \circ  },{C}^{ \circ  },{\widehat{\sigma }}^{2},{\widehat{\theta }}_{2},P}\right)\) .

Since \(\ell\left( {\widehat{R},\widehat{C},{\widehat{\sigma }}^{2},{\widehat{\theta }}_{2},P}\right)  - \ell\left( {{R}^{ \circ  },{C}^{ \circ  },{\widehat{\sigma }}^{2},{\widehat{\theta }}_{2},P}\right)  \geq  0\) , then
\begin{align}\label{Th1_4}
    & \frac{1}{2} \left\lbrack
    {\widehat{\sigma}}^{-2}
    \sum_{t = 1}^{n}
    \left\lVert
        {M}_{\widetilde{\Lambda}_{m_t}} \left( {y}_{t} - {\widehat{y}}_{t \mid t-1} \right)
    \right\rVert_{F}^{2}
    -
    {\widehat{\sigma}}^{-2}
    \sum_{t = 1}^{n}
    \left\lVert
        {M}_{\widetilde{\Lambda}_{s_t}^{\circ}} \left( {y}_{t} - {y}_{t \mid t-1}^{\circ} \right)
    \right\rVert_{F}^{2}
\right\rbrack \notag\\
\leq&    - n\log \mathop{\min }\limits_{{j,k}}{p}_{jk} - \frac{1}{2}\mathop{\sum }\limits_{{t = 1}}^{n}\log \frac{\left| {\widetilde{\Lambda }}_{m_t}{\widetilde{\Lambda}_{m_t}}^{\prime} + \widehat{\sigma}^{2}{I}_{pq}\right| }{\left| {\widetilde{\Lambda}}_{s_t}^{\circ}{{\widetilde{\Lambda }}_{s_t}^{\circ \prime}} + {\widehat{\sigma}}^{2}{I}_{pq}\right| } \notag\\
& - \frac{1}{2}\mathop{\sum }\limits_{{t = 1}}^{n}{\left( {y}_{t} - {\widehat{y}}_{t \mid  t - 1}\right) }^{\prime }{\widetilde{\Lambda }}_{m_t}{\left( {\widehat{\sigma }}^{2}{I}_{r} + {\widetilde{\Lambda }}_{m_t}^{\prime }{\widetilde{\Lambda }}_{m_t}\right) }^{-1}\left( {{\widetilde{\Lambda }}_{m_t}^{\prime }{\widetilde{\Lambda }}_{m_t}}\right) {}^{-1}{\widetilde{\Lambda }}_{m_t}^{\prime }\left( {{y}_{t} - {\widehat{y}}_{t \mid  t - 1}}\right)\notag\\
& + \frac{1}{2}\mathop{\sum }\limits_{{t = 1}}^{n}{\left( {y}_{t} - {\widehat{y}}_{t \mid  t - 1}^{\circ}\right) }^{\prime }{\widetilde{\Lambda}}_{{s_t}}^{\circ}{\left( {\widehat{\sigma }}^{2}{I_{r}} + {\widetilde{\Lambda }}_{{s_t}}^{\circ\prime }{\widetilde{\Lambda}}_{{s_t}}^{\circ}\right) }^{-1}\left( {{\widetilde{\Lambda }}_{{s_t}}^{\circ\prime }{\widetilde{\Lambda}}_{{S_t}}^{\circ}}\right)^{-1} {\widetilde{\Lambda}}_{{s_t}}^{\circ\prime}\left( {{y}_{t} - {\widehat{y}}_{t \mid  t - 1}^{\circ}}\right)
\end{align}

(3.1) The first term on the right hand side is $O(n)$ since \(\mathop{\min }\limits_{{j,k}}{p}_{jk} > 0\).

(3.2) consider the second term on the right hand side of (\ref{Th1_4}).
Since $\left| {{\widetilde{\Lambda}}_{m_t}{\widetilde{\Lambda}}_{m_t}^{\prime} + \widehat{\sigma}^{2}{I}_{pq}}\right|  = {\left( \widehat{\sigma}^{2}\right) }^{pq}\left| {{I}_{pq} + \frac{1}{\widehat{\sigma}^{2}}{\widetilde{\Lambda}}_{m_t }{\widetilde{\Lambda}}_{m_t }^{\prime}}\right|  = {\left( \widehat{\sigma }^{2}\right) }^{pq}\left| {{I}_{r} + \frac{1}{\widehat{\sigma }^{2}}{\widetilde{\Lambda}}_{m_t}^{\prime}{\widetilde{\Lambda}}_{m_t}}\right| $
and $\left| {{\widetilde{\Lambda }}_{s_t}^{\circ}{{\widetilde{\Lambda }}_{s_t}^{\circ\prime}}+ {\widehat{\sigma }}^{2}{I}_{pq}}\right|  = {\left( {\widehat{\sigma}}^{2}\right) }^{pq}\left| {{I}_{r} + \frac{1}{{\widehat{\sigma}}^{2}}{{\widetilde{\Lambda }}_{s_t}^{\circ\prime}}{\widehat{\Lambda }}_{s_t}^{\circ}}\right|$, where \(r = {k}_{1}{k}_{2}\) , then
\begin{align}
    & - \frac{1}{2}\mathop{\sum }\limits_{{t = 1}}^{n}\log \frac{\left| {\widetilde{\Lambda }}_{{m_t}}{\widetilde{\Lambda }}_{{m_t}}^{\prime}
    + {\widehat{\sigma}}^{2}{I}_{pq}\right| }{\left| {\widetilde{\Lambda }}_{{s_t}}^{\circ}{{\widetilde{\Lambda }}_{{s_t}}^{\circ\prime}} + {\widehat{\sigma }}^{2}{I}_{pq}\right| }\notag\\
    =&  - \frac{1}{2}\mathop{\sum }\limits_{{t = 1}}^{n}\log \left| {{\widetilde{\Lambda }}_{{m_t}}{\widetilde{\Lambda }}_{{m_t}}^{\prime} + {\widehat{\sigma}}^{2}{I}_{pq}}\right|
    \notag+ \frac{1}{2}\mathop{\sum }\limits_{{t = 1}}^{n}\log \left| {{\widetilde{\Lambda }}_{{s_t}}^{\circ}{{\widetilde{\Lambda }}_{{s_t}}^{\circ\prime}} + {\widehat{\sigma}}^{2}{I}_{pq}}\right| \notag\\
    = & - \frac{1}{2}\mathop{\sum }\limits_{{t = 1}}^{n}\left\lbrack  {{pq}{\ln}{\widehat{\sigma}}^{2} + \log \left| {{I}_{r} + \frac{1}{\widehat{\sigma }^{2}}{\widetilde{\Lambda }}_{m_t}^{\prime}{\widetilde{\Lambda }}_{m_t}}\right| }\right\rbrack   + \frac{1}{2}\mathop{\sum }\limits_{{t = 1}}^{n}\left\lbrack  {{pq}{\ln}{\widehat{\sigma}}^{2}+ \log \left| {{I}_{r} + \frac{1}{\widehat{\sigma }^{2}}{{\widetilde{\Lambda }}_{s_t}^{\circ\prime}}{\widetilde{\Lambda }}_{s_t}^{\circ}}\right| }\right\rbrack \notag\\
    \leq&    \frac{1}{2}\mathop{\sum }\limits_{{t= 1}}^{n}\log \left| {{I}_{r} + \frac{1}{\widehat{\sigma }^{2}}{{\widetilde{\Lambda}}_{s_t}^{\circ\prime} }\widetilde{\Lambda}_{s_t}}\right|  = {O}_{p}\left( {n\log \left( pq\right) }\right).
\end{align}

(3.3) The third term on the right hand side of (\ref{Th1_4}) is negative, thus inequality (\ref{Th1_4}) Still holds when this term is throw away.

(3.4) Consider the fourth term on the right hand side of (\ref{Th1_4}), that is,
\begin{align}
 & \frac{1}{2}\mathop{\sum }\limits_{{t = 1}}^{n}{\left( {y}_{t} - {y}^{\circ}_{t \mid t-1}\right) }^{\prime }{\widetilde{\Lambda }}_{s_t}^{\circ}{\left( {\widehat{\sigma}}^{2}{I}_{r} + {{\widetilde{\Lambda }}_{s_t}^{\circ\prime}}{\widetilde{\Lambda}}_{s_t}^{\circ}\right) }^{-1}{\left( {{\widetilde{\Lambda }}_{s_t}^{\circ\prime}}{\widetilde{\Lambda }}_{s_t}^{\circ}\right) }^{-1}{{\widetilde{\Lambda }}_{s_t}^{\circ\prime}}\left( {{y}_{t} - {y}_{t \mid t- 1}^{\circ}}\right)   \notag\\
= & \frac{1}{2}\mathop{\sum }\limits_{{t = 1}}^{n}{\left( {y}_{t} - {y}_{t \mid t - 1}^{ \circ  }\right) }^{\prime }{\widetilde{\Lambda }}_{{s_t}}^{ \circ  }{\left( {\widetilde{\Lambda }}_{{s_t}}^{\circ \prime }{\widetilde{\Lambda }}_{{s_t}}^{ \circ}\right) }^{-\frac{1}{2}}{\left( {\widetilde{\Lambda }}_{{s_t}}^{\circ \prime }{\widetilde{\Lambda }}_{{s_t}}^{ \circ  }\right) }^{\frac{1}{2}}{\left( {\widehat{\sigma}}^{2}{I}_{r} + {\widetilde{\Lambda }}_{s_t}^{\circ \prime }{\widetilde{\Lambda }}_{{s_t}}^{ \circ  }\right) }^{-1}{\left( {\widetilde{\Lambda }}_{{s_t}}^{\circ \prime }{\widetilde{\Lambda }}_{{s_t}}^{ \circ  }\right) }^{-\frac{1}{2}} \notag \\
&{\left( {\widetilde{\Lambda }}_{{s_t}}^{\circ \prime }{\widetilde{\Lambda }}_{{s_t}}^{ \circ  }\right) }^{-\frac{1}{2}}{\widetilde{\Lambda }}_{{s_t}}^{\circ \prime }\left( {{y }_{t} - {y}_{t \mid t- 1}^{ \circ  }}\right).
\end{align}
Since $ {\left( {\widetilde{\Lambda }}_{{s_t}}^{\circ\prime}{\widetilde{\Lambda }}_{{s_t}}^{\circ}\right) }^{\frac{1}{2}}{\left( {\widehat{\sigma}}^{2}{I_{r}} + {\widetilde{\Lambda }}_{{s_t}}^{\circ\prime}{\widetilde{\Lambda }}_{{s_t}}^{\circ}\right) }^{-1}{\left( {\widetilde{\Lambda }}_{{s_t}}^{\circ\prime}{\widetilde{\Lambda }}_{{s_t}}^{\circ}\right) }^{-\frac{1}{2}} = {\left( {\widehat{\sigma}}^{2}{I_{r}}+ {\widetilde{\Lambda }}_{{s_t}}^{\circ\prime}{\widetilde{\Lambda }}_{{s_t}}^{\circ}\right) }^{-1}$ and
\begin{align}
    \begin{Vmatrix}{{\left( {\widetilde{\Lambda }}_{{s_t}}^{\circ \prime }{\widetilde{\Lambda }}_{{s_t}}^{ \circ  }\right) }^{-\frac{1}{2}}{\widetilde{\Lambda }}_{s_t}^{\circ\prime}\left( {{y}_{t} - {\widehat{y}}_{t \mid t-1}^{\circ}}\right) }\end{Vmatrix}  & = \sqrt{{\left( {y}_{t} - {\widehat{y}}_{t \mid  t - 1}^{\circ}\right) }^{\prime}{\widetilde{\Lambda }}_{s_t}^{\circ}{\left( {\widetilde{\Lambda }}_{{s_t}}^{\circ \prime }{\widetilde{\Lambda }}_{{s_t}}^{ \circ  }\right) }^{-1}{\widetilde{\Lambda }}_{s_t}^{\circ\prime}\left( {{y}_{t} - {\widehat{y}}_{t \mid  t - 1}^{\circ}}\right) } \notag\\
    &= \sqrt{{\left( {y}_{t} - {\widehat{y}}_{t \mid  t - 1}^{ \circ  }\right) }^{\prime }{P}_{{\widetilde{\Lambda}}_{s_t}^{ \circ  }}\left( {{y}_{t} - {\widehat{y}}_{t \mid  t - 1}^{ \circ  }}\right) } \notag\\
    &
    = \sqrt{{\left( {y}_{t} - {\widehat{y}}_{t \mid  t - 1}^{ \circ  }\right) }^{\prime }{P}_{{\widetilde{\Lambda}}_{s_t}^{ \circ  }}^{2}\left( {{y}_{t} - {\widehat{y}}_{t \mid  t - 1}^{ \circ  }}\right) } \notag\\
    & =\begin{Vmatrix}{{{P}_{\widetilde{\Lambda }_{s_t}^{\circ}}}\left( {{y}_{t} - {\widehat{y}}_{t \mid  t - 1}^{\circ}}\right) }\end{Vmatrix} \leq  \begin{Vmatrix}{{y}_{t} - {\widehat{y}}_{t \mid  t - 1}^{\circ}}\end{Vmatrix},
\end{align}
then the fourth term can be bounded by \(\dfrac{1}{2}\mathop{\sum }\limits_{{t= 1}}^{n}{\begin{Vmatrix}{y}_{t} - {\widehat{y}}_{t \mid t- 1}^{\circ}\end{Vmatrix}}^{2} \cdot  {\begin{Vmatrix}{\left( {\widehat{\sigma }}^{2}{I}_{r} +  {\widetilde{\Lambda }}_{{s_t}}^{\circ\prime}{\widetilde{\Lambda }}_{{s_t}}^{\circ}\right) }^{-1}\end{Vmatrix}}\) . By Assumption B(1),
\begin{align}
    \begin{Vmatrix}{\left( {\widehat{\sigma}}^{2}{I_{r}}+ {\widetilde{\Lambda }}_{{s_t}}^{\circ\prime}{\widetilde{\Lambda }}_{{s_t}}^{\circ}\right) }^{-1}\end{Vmatrix} & \leq  \sup\limits_{j} \begin{Vmatrix}{\left( {\widehat{\sigma}}^{2}{I_{r}}+ {\widetilde{\Lambda }}_{{j}}^{\circ\prime}{\widetilde{\Lambda }}_{{j}}^{\circ}\right) }^{-1}\end{Vmatrix} \leq  \sup\limits_{j} \begin{Vmatrix}{\left( {\widetilde{\Lambda }}_{{j}}^{\circ\prime}{\widetilde{\Lambda }}_{{j}}^{\circ}\right) }^{-1}\end{Vmatrix} \notag\\
    & = \sup\limits_{j} \begin{Vmatrix}{\left( {\widehat{V}}_{t \mid t- 1}^{\circ\frac{1}{2}}{\Lambda }_{j}^{\circ\prime}{\Lambda }_{j}^{\circ} {\widehat{V}}_{t \mid t- 1}^{\circ\frac{1}{2}}\right) }^{-1}\end{Vmatrix} = {O}_{p}\left( \frac{1}{{pq}}\right) .
\end{align}

By Assumption B(1), C(1) and D(1), \(\mathop{\sum }\limits_{{t= 1}}^{n}\parallel {y}_{t} - {\widehat{y}}_{t \mid t- 1}^{\circ}{\parallel }^{2} = {O}_{p}\left( {pqn}\right)\) . Thus the fourth term is \({O}_{p}\left( n\right)\).

(3.5) Now consider the left hand side of expression (\ref{Th1_4}). Since \({y}_{t} = {\Lambda }_{s_t}^{\circ}{f}_{t} + {e}_{t}\) and
$ {M}_{\widetilde{\Lambda}_{s_t}^{0}} ={M}_{{\Lambda}_{s_t}^{\circ}},{M}_{{\Lambda}_{s_t}^{0}}{\Lambda }_{s_t}^{\circ} {f}_{t}= 0$, it is easy to verify that
\begin{align} \label{Th1_5}
    & \frac{1}{2}{\widehat{\sigma }}^{-2}\left\lbrack  {\mathop{\sum }\limits_{{t = 1}}^{n}{\begin{Vmatrix}{M}_{\widetilde{\Lambda }_{m_t}}\left( {y}_{t} - \widehat{y}_{t \mid  t - 1}\right) \end{Vmatrix}}^{2} - \mathop{\sum }\limits_{{t = 1}}^{n}{\begin{Vmatrix}{M}_{\widetilde{\Lambda }_{s_t}^{\circ}}\left( {y}_{t} - y^{\circ}_{t \mid  t - 1}\right) \end{Vmatrix}}^{2}}\right\rbrack \notag\\
    =& \frac{1}{2}{\widehat{\sigma }}^{-2}\left\lbrack  {\mathop{\sum }\limits_{{t = 1}}^{n}\parallel {M}_{\widetilde{\Lambda }_{m_t}}\left( {{\Lambda }_{s_t}^{\circ}{f}_{t} - \widehat{\Lambda }_{s_t}^{  }\widehat{f}_{t \mid t- 1}}\right)  + {M}_{\widetilde{\Lambda }_{m_t}}{e}_{t}{\parallel }^{2} - \mathop{\sum }\limits_{{t = 1}}^{n}\parallel {M}_{{\Lambda }_{s_t}^{ \circ  }}{e}_{t}{\parallel }^{2}}\right\rbrack \notag\\
    \geq&  \frac{1}{2}{\widehat{\sigma }}^{-2}\left\lbrack  {\mathop{\sum }\limits_{{t = 1}}^{n}\parallel {M}_{{\widetilde{\Lambda}}_{m_t}}\left( {{\Lambda }_{{s_t}}^{ \circ  } {f}_{t}- {\widehat{\Lambda}}_{{s_t}}{\widehat{f}}_{t \mid t-1}}\right)  + {M}_{{\widetilde{\Lambda }}_{m_t}}{e}_{t}{\parallel }^{2}}\right\rbrack \notag\\
    \geq & \frac{1}{2}{\widehat{\sigma }}^{-2}\left\lbrack  {\mathop{\sum }\limits_{{t = 1}}^{n}\parallel {M}_{\widetilde{\Lambda }_{m_t}}\left( {{\Lambda }_{{s_t}}^{ \circ  } {f}_{t}- {\widehat{\Lambda}}_{{s_t}}{\widehat{f}}_{t \mid t-1}}\right)  + {M}_{{\widetilde{\Lambda }}_{m_t}}{e}_{t}{\parallel }^{2}{I}_{\left( {m_t} = {s_t}\right) }}\right\rbrack \notag\\
    = &\frac{1}{2}{\widehat{\sigma }}^{-2}\left\lbrack  {\mathop{\sum }\limits_{{t = 1}}^{n}\parallel {M}_{\widetilde{\Lambda }_{m_t}}{\Lambda }_{s_t}^{ \circ  }{f}_{t} + {M}_{\widetilde{\Lambda }_{m_t}}{e}_{t}{\parallel }^{2}{I}_{\left( m_t = {s_t}\right) }}\right\rbrack \notag\\
    =& \frac{1}{2}{\widehat{\sigma }}^{-2}\left\lbrack  {\mathop{\sum }\limits_{{t = 1}}^{n}\parallel {M}_{\widetilde{\Lambda }_{m_t}}{\Lambda }_{s_t}^{ \circ  }{f}_{t}{\parallel }^{2}} {I}_{\left( {m_t} = {s_t}\right) } + \mathop{\sum }\limits_{{t = 1}}^{n}\parallel{M}_{\widetilde{\Lambda }_{m_t}}{e}_{t}{\parallel }^{2}{I}_{\left( {m_t} = {s_t}\right) }\right.\notag\\
    & \left. + 2\mathop{\sum }\limits_{t}{e }_{t}^{\prime}{M}_{\widetilde{\Lambda }_{m_t}}{\Lambda }_{s_t}^{ \circ  }{f}_{t}{I}_{\left( {m_t} = {s_t}\right) }\right\rbrack.
\end{align}

For the second term of (\ref{Th1_5}), we have \(\mathop{\sum }\limits_{t}\left| \right| {M}_{\widetilde{\Lambda }_{m_t}}{e}_{t}{\left| \right| }^{2}  = \mathop{\sum }\limits_{t}{e }_{t}^{\prime}{M}_{\widetilde{\Lambda }_{m_t}}{e}_{t} = \mathop{\sum }\limits_{t}{e }_{t}^{\prime}{e}_{t}-\mathop{\sum }\limits_{t}e_t^{\prime}{P}_{\widetilde{\Lambda }_{m_t}} e_t\),
and
\begin{align} \label{Th1_6}
\sum_{t=1}^n {\begin{Vmatrix}{P}_{\widetilde{\Lambda }_{m_t}}{e}_{t}\end{Vmatrix}}_{F}^{2}   & =    \mathop{\sum }\limits_{{t = 1}}^{n}{\begin{Vmatrix}{P}_{\widehat{\Lambda }_{m_t}}{e}_{t}\end{Vmatrix}}_{F}^{2} \leq  \mathop{\sum }\limits_{{j = 1}}^{M}\mathop{\sum }\limits_{{t = 1}}^{n}{\begin{Vmatrix}{P}_{\widehat{\Lambda }_{j}}{e}_{t}\end{Vmatrix}}_{F}^{2} \notag\\
 &= \mathop{\sum }\limits_{{j = 1}}^{M}\mathop{\sum }\limits_{{t= 1}}^{n}{{e}_{t}}^{\prime}{\widehat{\Lambda }_{j}}\left({\widehat{\Lambda }_{j}^{\prime}}{\widehat{\Lambda }_{j}}\right )^{-1}{\widehat{\Lambda }_{j}}{e}_{t} \notag\\
 & = \sum_j\mathop{\sum }\limits_{t}{T}_{r}\left\lbrack  {{e }_{t}^{\prime }{\widehat{\Lambda }}_{j}{\left( {\widehat{\Lambda }}_{j}^{\prime }{\widehat{\Lambda }}_{j}\right) }^{-\frac{1}{2}}{\left( {\widehat{\Lambda }}_{j}^{\prime }{\widehat{\Lambda }}_{j}\right) }^{-\frac{1}{2}}{\widehat{\Lambda }}_{j}^{\prime }{e}_{t}}\right\rbrack \notag\\
 & =\sum_j {T}_{r}{\left\lbrack{\left( {\widehat{\Lambda }}_{j}^{\prime }{\widehat{\Lambda }}_{j}\right) }^{-\frac{1}{2}} {\widehat{\Lambda }}_{j}^{\prime }\mathop{\sum }\limits_{t}{e}_{t}{e}_{t}^{\prime}{\widehat{\Lambda }}_{j}{\left( {\widehat{\Lambda }}_{j}^{\prime }{\widehat{\Lambda }}_{j}\right) }^{-\frac{1}{2}} \right\rbrack  } \notag\\
 &= \sum_j{T}_{r}\left\lbrack {\mathop{\sum }\limits_{t}{e}_{t}{e}_{t}^{\prime}\widehat{\Lambda}_j {\left( {\widehat{\Lambda }}_{j}^{\prime }{\widehat{\Lambda }}_{j}\right) }^{-1} {\widehat{\Lambda }}_{j}^{\prime }}\right\rbrack   =\sum_j {T}_{r}\left\lbrack {\mathop{\sum }\limits_{t}{e}_{t}{e}_{t}^{\prime}P_{\widehat{\Lambda}_{j}}}\right\rbrack   \notag\\
 & \leq  {\rho }_{\max }\left(  {\mathop{\sum }\limits_{t}{e}_{t}{e}_{t}^{\prime}}\right)  \cdot  \sum_j\operatorname{Tr}\left( P_{\widehat{\Lambda }_{j}}\right)  = Mr {\rho }_{\max }\left( {\mathop{\sum }\limits_{t}{e}_{t}{e}_{t}^{\prime}}\right) \notag\\
 & = Mr\begin{Vmatrix}{{E}^{\prime }E}\end{Vmatrix} = {O}_{p}\left( {{p}^{\frac{1}{2}}{q}^{\frac{1}{2}}n + {pq}{n}^{\frac{1}{2}}}\right).
\end{align}
 So the second term of (\ref{Th1_5}) can be bounded by \({O}_{p}\left( {{p}^{\frac{1}{2}}{q}^{\frac{1}{2}}n + {pq}{n}^{\frac{1}{2}}}\right)\) . The third term of (\ref{Th1_5}) equals \(2\mathop{\sum }\limits_{t}{e }_{t}^{\prime}{\Lambda }_{s_t}^{ \circ  }{f}_{t}{I}_{\left( {m_t} = {s_t}\right) } - {2\mathop{\sum }\limits_{t}{e }_{t}^{\prime}{P}_{\widetilde{\Lambda }_{m_t}}{\Lambda }_{s_t}^{ \circ  }{f}_{t}{I}_{\left( {m_t} = {s_t}\right) }}.\)  By Assumption E, we have
\begin{align}
   {\begin{Vmatrix}\mathop{\sum }\limits_{{t= 1}}^{n}{e}_{t}^{\prime}{\Lambda }_{s_t}^{\circ}{f}_{t}\end{Vmatrix}}^{2}
   &= {\begin{Vmatrix}\mathop{\sum }\limits_{t}{Vec}^{\prime}\left( {F}_{t}\right) \left( {C}_{{s_t}}^{\circ\prime} \otimes  {R}_{s_t}^{\circ\prime}\right) {Vec}\left( {E}_{t}\right) \end{Vmatrix}}^{2} = {\begin{Vmatrix}\mathop{\sum }\limits_{t}{Vec}^{\prime}\left( {F}_{t}\right)  {Vec}\left( {R}_{s_t}^{\circ\prime}{E}_{t}{C}_{{s_t}}^{\circ} \right) \end{Vmatrix}}^{2} \notag\\
   & \lesssim  \mathop{\sum }\limits_{{s = 1}}^{{k}_{1}}\mathop{\sum }\limits_{{h = 1}}^{{k}_{2}}{\begin{Vmatrix}\dfrac{1}{\sqrt{n}}\mathop{\sum }\limits_{t}{f}_{t,{sh}}\dfrac{{R}_{{s_t},\cdot s}^{\circ\prime}}{\sqrt{p}}{E}_{t}\dfrac{{C}_{{s_t}, \cdot h}^{\circ}}{\sqrt{q}}\end{Vmatrix}}^{2} \times  {npq} \notag\\
   &=  {O}_{p}\left( {n{pq}}\right).
\end{align}

By expression (\ref{Th1_6}), Assumption A(1) and B(1), we have
$\parallel  \mathop{\sum }\limits_{t}{e}_{t}^{\prime}{P}_{\widetilde{\Lambda }_{m_t}}{\Lambda }_{s_t}^{\circ}{f}_{t}\parallel  \leq  {\left(\mathop{\sum }\limits_{t}\parallel {e}_{t}^{\prime}{P}_{\widetilde{\Lambda }_{m_t}}{\parallel }^{2}\right) }^{\frac{1}{2}}{\left(   \mathop{\sum }\limits_{t}\parallel {f}_{t}{\parallel }^{2}\right) }^{\frac{1}{2}} \cdot \sup\limits_{j} \parallel {\Lambda }_{j}^{\circ}\parallel  = {O}_{p}\left( {{p}^{\frac{3}{4}}{q}^{\frac{3}{4}}n + {pq}{n}^{\frac{3}{4}}}\right) .$  Thus the third term of expression (\ref{Th1_5}) is \({O}_{p}\left( {{p}^{\frac{3}{4}}{q}^{\frac{3}{4}}n + {pq}{n}^{\frac{3}{4}}}\right)\) .

(3.6) Move from expression (\ref{Th1_5}) to the right hand side of equation (\ref{Th1_4}), and
take the results of \(\left( {3.1}\right)  - \left( {3.5}\right)\) together, we have
\begin{align}
   0 \leq  \frac{1}{2}\widehat{\sigma }^{-2} \mathop{\sum}\limits_{t=1}^{n}{\begin{Vmatrix}{M}_{{\widetilde{\Lambda }}_{m_t}}{\Lambda }_{s_t}^{ \circ  }{f}_{t}\end{Vmatrix}}^{2}{I}_{\left( {m_t} = {s_t}\right) }
   & \leq  O\left( n\right)  + {O}_{p}\left( {n\log \left( pq\right) }\right)  + O\left( n\right) \notag\\
   & + {O}_{p}\left( {{p}^{\frac{1}{2}}{q}^{\frac{1}{2}}{n} + {pq}{n}^{\frac{1}{2}}}\right) \notag\\
   & + {O}_{p}\left( {{p}^{\frac{3}{4}}{q}^{\frac{3}{4}}n + {pq}{n}^{\frac{3}{4}}}\right).
\end{align}
Thus$
\mathop{\sum}\limits_{t=1}^{n}{\begin{Vmatrix}{M}_{{\widetilde{\Lambda }}_{m_t}}{\Lambda }_{s_t}^{ \circ  }{f}_{t}\end{Vmatrix}}^{2}{I}_{\left( {m_t} = {s_t}\right) }={O}_{p}\left( {{p}^{\frac{3}{4}}{q}^{\frac{3}{4}}n + {pq}{n}^{\frac{3}{4}}}\right).
$

In the summation, there are about \({\pi}_{1}^{\circ}n\) terms with \({\Lambda }_{s_t}^{\circ} = {\Lambda }_{1}^{\circ}\) , since \({\pi}_{1}^{\circ}\) is the unconditional probability of \({s_t} = 1\) (Rigorously speaking, there are \(\mathop{\sum }\limits_{{t = 1}}^{n}{I}_{(s_t=1)}\) terms, but\(\frac{1}{n} \mathop{\sum }\limits_{{t = 1}}^{n}{I}_{(s_t=1)}\xrightarrow{P} {\pi}_{1 }^{\circ}\), as \( n\rightarrow \infty\) ). For each \(t\) with \({s}_{t} = 1,{\Lambda }_{1}^{\circ}{f}_{t}\) are projected on one of \({\widetilde{\Lambda }}_{j}\) , \(j = 1,\cdots M\) , thus there exists one certain \({\widetilde{\Lambda }}_{j}\) such that \({\Lambda }_{1}^{\circ}{f}_{t}\) is projected on \(\widetilde{\Lambda }_{j}\) at least \({\pi }_{1}^{ \circ  }n  / M\) times. Defined this \({\widetilde{\Lambda }}_{j}\) as \({\widetilde{\Lambda }}_{1}\), then \(\mathop{\sum }\limits_{{t = 1}}^{n}{I}_{(m_t = 1)}{I}_{(s_t=1)} \geq  \frac{n{\pi }_{1}^{ \circ  }}{M}\) . Thus by Assumption C(2.4),

\[
{\rho }_{\min }\left( {\frac{1}{\mathop{\sum }\limits_{{t = 1}}^{n}{I}_{(m_t = 1)}{I}_{({s_t} = 1)}}\mathop{\sum }\limits_{{t = 1}}^{n}{f}_{t}{f}_{t}^{\prime }{I}_{(m _t =1)}{I}_{({s_t} = 1)}}\right)  \geq  {c}_{0},
\]
for some \({c}_{0}>0 \)  \( w.p.a.1.\) Then we have
\begin{align}
    {O}_{p}\left( {{p}^{\frac{3}{4}}{q}^{\frac{3}{4}}n + {pq}{n}^{\frac{3}{4}}}\right) & = \mathop{\sum}\limits_{t=1}^{n}{\begin{Vmatrix}{M}_{{\widetilde{\Lambda }}_{m_t}}{\Lambda }_{s_t}^{ \circ  }{f}_{t}\end{Vmatrix}}^{2}{I}_{\left( {m_t} = {s_t}\right) }\geq  \mathop{\sum}\limits_{t=1}^{n}{\begin{Vmatrix}{M}_{{\widetilde{\Lambda }}_{1}}{\Lambda }_{s_t}^{ \circ  }{f}_{t}\end{Vmatrix}}^{2}I_{\left( {{m_t} = 1}\right)} I_{\left({{s_t} = 1}\right)} \notag\\
    & = \operatorname{tr}\left( {{\Lambda }_{1}^{\circ\prime}{M}_{\widetilde{\Lambda }_{1}}{\Lambda }_{1}^{\circ}\mathop{\sum }\limits_{{t = 1}}^{n}{f}_{t}{f}_{t}^{\prime}{I}_{\left( {m_t} = 1\right) }{I}_{\left( {s_t} = 1\right) }}\right) \notag\\
    &\geq \operatorname{tr}\left( {\Lambda }_{1}^{\circ}{}^{\prime}{M}_{\widetilde{\Lambda }_{1}}{\Lambda }_{1}^{\circ}\right) {\rho }_{\min }
\left(
    \mathop{\sum}\limits_{t}
    {f}_{t} {f}_{t }^{\prime}
    {I}_{\left( {m}_{t} = 1 \right)}
    {I}_{\left( s_t = 1 \right)}
\right) \notag\\
&\geq \operatorname{tr}\left( {{\Lambda }_{1}^{\circ\prime}{M}_{\widehat{\Lambda }_{1}}{\Lambda }_{1}^{\circ}}\right) \frac{n{\pi }_{1}^{o}}{M}{c}_{0}\;\text{ w.p.a.1 }
\end{align}
Thus \(\frac{1}{{pq}}{\begin{Vmatrix}{M}_{{\widehat{\Lambda }}_{1}}{\Lambda }_{1}^{ \circ  }\end{Vmatrix}}_{F}^{2} = \frac{1}{pq}\operatorname{tr}\left( {{\Lambda }_{1}^{\circ\prime}{M}_{{\widehat{\Lambda }}_{1}}{\Lambda }_{1}^{ \circ  }}\right)  = {O}_{p}\left( {{p}^{-\frac{1}{4}}{q}^{-\frac{1}{4}} + {n}^{-\frac{1}{4}}}\right)  \stackrel{\triangle}{=}
{O}_{p}\left( \frac{1}{\sqrt{{\delta}_{p q n}}}\right),\)  where
\({\delta}_{pqn} = \min \{ \sqrt{pq},\sqrt{n}\}\) . Similarly, for \(k = 2,\cdots ,M\) , we also have \(\frac{1}{pq}{\begin{Vmatrix}{M}_{\widehat{\Lambda}_{k}}{\Lambda }_{k}^{0}\end{Vmatrix}}_{F}^{2} = {O}_{p}\left( \frac{1}{\sqrt {{\delta}_{pqn} }}\right)\).
Under Assumption B(1), we can further have
\begin{align}
    O_{p}\left( \frac{1}{\sqrt {{\delta}_{pqn} }}\right) & = \frac{1}{{pq}}{\begin{Vmatrix}{M}_{{\widehat{\Lambda}}_{k}}{\Lambda }_{k}^{\circ}\end{Vmatrix}}_{F}^{2} = \frac{1}{{pq}}\operatorname{tr}\left\lbrack  {{\Lambda }_{k}^{\circ\prime}{M}_{{\widehat{\Lambda}}_{k}}{\Lambda }_{k}^{0}}\right\rbrack \notag\\
    & = \frac{1}{{pq}}\operatorname{tr}\left\lbrack  {{ \Lambda_{k}^{\circ\prime}\Lambda_k^{\circ}} - {\Lambda }_{k}^{\circ\prime}{\widehat{\Lambda}}_{k}\left( {{\widehat{\Lambda}}_{k}^{\prime}{\widehat{\Lambda}}_{k}}\right)^{-1}  \widehat{\Lambda }_{k}^{\prime}{\Lambda}_{k}^{\circ}}\right\rbrack \notag\\
    &= \operatorname{tr}\left\lbrack  {D_k^{(2)}\otimes D_k^{(1)} - \frac{{C}_{k}^{\circ\prime}\widetilde{C}_{k}\widetilde{C}_{k}^{\prime}{C}_{k}^{\circ}}{q} \otimes  \frac{{R}_{k}^{\circ\prime}\widetilde{R}_{k}\widetilde{R}_{k}^{\prime}{R}_{k}^{\circ}}{p}}\right\rbrack ,
\end{align}
where $\widetilde{R} = \widehat{R}\widehat{D}_1^{-1/2}$ and $\widetilde{C} = \widehat{C}\widehat{D}_2^{-1/2}$, which means
\begin{align}
    & \frac{{C}_{k}^{\circ\prime }{\widetilde{C}}_{k}}{q} \cdot  \frac{\widetilde{C}_{k}^{\prime }{C}_{k}^{\circ}}{q} \otimes  \frac{{R}_{k}^{\circ\prime }\widetilde{R}_{k}}{p} \cdot  \frac{\widetilde{R}_{k}^{\prime}{R}_{k}^{\circ }}{p} = D_k^{(2)}\otimes D_k^{(1)} + {O}_{p}\left( \frac{1}{\sqrt{\delta_{pqn}}}\right),  \notag\\
 & \frac{\widetilde{C}_{k}^{\circ\prime }{\widetilde{C}}_{k}}{q} \cdot  \frac{\widetilde{C}_{k}^{\prime }{\widetilde{C}_{k}^{\circ}}}{q} \otimes  \frac{{\widetilde{R}}_{k}^{\circ\prime }\widetilde{R}_{k}}{p} \cdot  \frac{\widetilde{R}_{k}^{\prime}{\widetilde{R}}_{k}^{\circ }}{p} = I_r + {O}_{p}\left( \frac{1}{\sqrt{\delta_{pqn}}}\right).
\end{align}
Likewise, by examining the likelihood function of \(\left( {\text{Vec}\left( {y}_{t}^{\prime }\right) }\right)\) , we can also get
\[
\frac{{\widetilde{R}}_{k}^{\circ\prime}\widetilde{R}_{k}}{p} \cdot  \frac{\widetilde{R}_{k}^{\prime  }{\widetilde{R}}_{k}^{\circ}}{p} \\\otimes\frac{{\widetilde{C}}_{k}^{\circ\prime}\widetilde{{C}}_{k}}{q}\frac{\widetilde{C}_{k}^{\prime}{\widetilde{C}}_{k}^{\circ}}{q} = {I}_{r} + {O}_{p}\left( \frac{1}{\sqrt{\delta_{pqn}}}\right) .
\]
\end{proof}

\setcounter{equation}{0}
 \setcounter{subsection}{0}
 \renewcommand{\theequation}{C.\arabic{equation}}
 \renewcommand{\thesubsection}{C.\arabic{subsection}}
\section*{C ~~ Details for Theorem 2}

\begin{proof}
     Step(1): We first show \(\left| {{\widehat{w}}_{t | t}^{\left( k\right) } - {I}_{\left( {s}_{t} = k\right) }}\right|  = {o}_{p}\left( \frac{1}{{N}^{\eta}}\right)\) , where \(N = {pq}\) . Define $\widehat{\Theta}=\{\widehat{\theta }_{1},{\widehat{\theta }}_{2}\}$.

When \({s}_{t} = k\) , since
\begin{align}
{\widehat{w}}_{t | t}^{\left( k\right) }
&= {P}_{r}\left( {{s}_{t} = k \mid  {y}_{t},{\mathcal{Y}}_{t - 1};\widehat{\Theta}}\right)
 \notag\\
&=\frac{{P}_{r}\left( {{s}_{t} = k,{y}_{t} \mid  {\mathcal{Y}}_{t - 1};\widehat{\Theta}}\right)}{p \left( {y}_{t} \mid  {\mathcal{Y}}_{t - 1}; \widehat{\Theta}\right)}
 \notag\\
&= \frac{p\left( {{y}_{t} \mid  {s}_{t} = k,{\mathcal{Y}}_{t-1};\widehat{\Theta}}\right) {P}_{r}\left( {s}_{t} = k \mid  {\mathcal{Y}}_{t-1};\widehat{\Theta} \right)}{\mathop{\sum }\limits_{{j = 1}}^{M}p\left( {{y}_{t} \mid  {s}_{t} = j, {\mathcal{Y}}_{t-1};\widehat{\Theta}}\right) {P}_{r}\left( {{s}_{t} = j \mid  {\mathcal{Y}}_{t - 1};\widehat{\Theta}}\right)}
 \notag\\
&= \frac{{\widehat{w}}_{t | t - 1}^{\left( k\right) }p\left( {{y}_{t} \mid  {s}_{t} = k,  {\mathcal{Y}}_{t - 1};\widehat{\Theta}}\right) }{\mathop{\sum }\limits_{{j = 1}}^{M}{\widehat{w}}_{t | t - 1}^{\left( j\right) }p\left( {{y}_{t} \mid  {s}_{t} = j, {\mathcal{Y}}_{t - 1};\widehat{\Theta}}\right) }, \notag
\end{align}
we have
\begin{align}
\left| {{\widehat{w}}_{t | t}^{\left( k\right) } - {I}_{\left( {s}_{t} = k\right) }}\right|
&= \frac{\mathop{\sum }\limits_{{j \neq  k}}{\widehat{w}}_{t |t - 1}^{\left( j\right) }p\left( {{y}_{t} \mid  {s}_{t} = j,\mathcal{Y}_{t-1};\widehat{\Theta}}\right) }{\mathop{\sum }\limits_{{j = 1}}^{M}{\widehat{w}}_{t | t - 1}^{\left( j\right) }p\left( {{y}_{t} \mid  {s}_{t} = j,\mathcal{Y}_{t-1};\widehat{\Theta}}\right) }
 \notag \\
&\leq\frac{\mathop{\sum }\limits_{{j \neq  k}}{\widehat{w}}_{t | t - 1}^{\left( j\right) }p\left( {y}_{t} \mid  {s}_{t} = j,\mathcal{Y}_{t-1},\widehat{\Theta} \right) }{{\widehat{w}}_{t \mid  t - 1}^{\left( k\right) }p\left( {{y}_{t} \mid  {s}_{t} = k,\mathcal{Y}_{t-1};\widehat{\Theta}}\right) }
\notag\\
&= \mathop{\sum }\limits_{{j \neq  k}}\frac{{\widehat{w}}_{t | t - 1}^{\left( j\right) }}{{\widehat{w}}_{t | t - 1}^{\left( k \right) }}\frac{p\left( {{y}_{t} \mid  {s}_{t} = j,\mathcal{Y}_{t-1};\widehat{\Theta}}\right) }{p\left( {{y}_{t} \mid  {s}_{t} = k,\mathcal{Y}_{t-1};\widehat{\Theta}}\right) }
\notag\\
&= \mathop{\sum }\limits_{{j \neq  k}}\frac{{\widehat{w}}_{t | t - 1}^{\left( j \right)}}{{\widehat{w}}_{t | t - 1}^{\left( k \right)}}\exp \left\{  \log p\left( {{y}_{t}|{s}_{t} = j,\mathcal{Y}_{t-1};\widehat{\Theta}}\right)  \right.
\notag \\
&- \left.
\log p\left( {{y}_{t}|{s}_{t} = k,\mathcal{Y}_{t-1};\widehat{\Theta}}\right) \right\}. \notag
\end{align}
When \({s}_{t} = h \neq k\) , Sine \(\mathop{\sum }\limits_{{j = 1}}^{M}{\widehat{w}}_{t | t}^{\left( j\right) } = 1\) , we have \({\widehat{w}}_{t | t}^{\left( k\right) } - {I}_{({s}_{t} = k)} = {\widehat{w}}_{t | t}^{\left( k\right) } = 1 - \mathop{\sum }\limits_{{j \neq k}}{\widehat{w}}_{t | t}^{\left( j\right) }\)  \(\leq  1 - {\widehat{w}}_{t | t}^{\left( h\right) }\) , thus it suffices to show \(\left| {{\widehat{w}}_{t | t}^{\left( k\right) } - {I}_{\left( {s}_{t} = k\right) }}\right|  = {o}_{p}\left( \frac{1}{{N}^{\eta}}\right)\) when \({s}_{t} = k\) .
Since
\begin{align}
{\widehat{w}}_{t | t-1}^{\left( k\right) }
&= {P}_{r}\left( {{s}_{t} = k \mid  \mathcal{Y}_{t-1};\widehat{\Theta},P}\right)  = \mathop{\sum }\limits_{{i = 1}}^{M}{P}_{r}\left( {s}_{t} = k,{s}_{t - 1} = i \mid \mathcal{Y}_{t-1};\widehat{\Theta},P\right)
\notag\\
&= \mathop{\sum }\limits_{{i = 1}}^{M}{P}_{r}\left( {{s}_{t} = k \mid  {s}_{t - 1} = i;P}\right) {P}_{r}\left( {{s}_{t - 1} = i \mid  \mathcal{Y}_{t-1};\widehat{\Theta},P}\right)
\notag\\
&= \mathop{\sum }\limits_{{i = 1}}^{M}{p}_{ik}{\widehat{w}}_{t-1 | t-1}^{\left( i\right) } \geq  \mathop{\min }\limits_{i}{p}_{ik} > 0,\notag
\end{align}
for all $k$ , it suffices to show
\[\mathop{\sup }\limits_{t} \exp \left\{ \log p\left( {{y}_{t} \mid {s}_{t} = j,\mathcal{Y}_{t-1};\widehat{\Theta}}\right)  - \log p\left( {{y}_{t} \mid {s}_{t} = k,\mathcal{Y}_{t-1};\widehat{\Theta}}\right)\right\} = {o}_{p}\left( \frac{1}{{N}^{\eta}}\right)\]
for any \( k \neq  j \) , i.e., it suffices to show for any fixed \({c}_{0} > 0\)

\[
{P}_{r}\left( \sup_{t} \left\lbrack {\log p\left( {{y}_{t} \mid  {s}_{t} = j,\mathcal{Y}_{t-1};\widehat{\Theta}}\right)  - \log p\left( {{y}_{t} \mid  {s}_{t} = k,\mathcal{Y}_{t-1};\widehat{\Theta}}\right) }\right\rbrack \geq  { \log \frac{{c}_{0}}{{N}^{\eta}} } \right) \to 0
\]

or,
\begin{align}\label{Th2_0}
   & {P}_{r}\left\{ { \mathop{\min }\limits_{t} \left\lbrack  {\log p\left( {{y}_{t} \mid  {s}_{t} = k,\mathcal{Y}_{t-1};\widehat{\Theta}}\right)  - \log p\left( {{y}_{t} \mid  {s}_{t} = j,\mathcal{Y}_{t-1} ;\widehat{\Theta}}\right) }\right\rbrack  } \right. \notag \\
  & \left. \leq  \eta\log N - \log {c}_{0} \right\} \to  0.
\end{align}
Since ${y}_{t \mid  t - 1}^{(k)} \triangleq  E\left\lbrack  {{y}_{t} \mid s_t=k, \mathcal{Y}_{t - 1};\Theta}\right\rbrack = \Lambda_k f_{t|t-1}^{(k)}$ and ${\Sigma }_{t \mid  t - 1}^{(k)} \triangleq  \operatorname{Cov}\left\lbrack  {{y}_{t} \mid  s_t=k,\mathcal{Y}_{t - 1};\Theta}\right\rbrack = \sigma^2I_{pq}+ \Lambda_k V_{t|t-1}^{(k)}\Lambda_k^{\prime}$, where  ${f}_{t \mid  t - 1}^{(k)}\triangleq  E\left\lbrack  {{f}_{t} \mid s_t=k, \mathcal{Y}_{t - 1};\Theta}\right\rbrack$ and \({V}_{t\mid {t} - 1}^{(k)}\triangleq\operatorname{Cov}\left\lbrack  {{f}_{t} \mid  s_t=k,\mathcal{Y}_{t - 1};\Theta}\right\rbrack\) can be given by the kalman filter recursions, then we can get
\begin{align}
    \log p\left( {{y}_{t} \mid  {s}_{t} = k,\mathcal{Y}_{t-1};\widehat{\Theta}}\right)  &=  - \frac{pq}{2}\log \left( 2\pi\right)  - \frac{1}{2}\log  \left|\widehat{\Sigma}_{t \mid  t - 1}^{\left( k\right) }\right| \notag\\
    &  - \frac{1}{2}\left( {y}_{t} - \widehat{y}_{t \mid  t - 1}^{\left( k\right) }\right) ^{\prime} \widehat{\Sigma}_{t \mid  t - 1}^{(k)-1} \left( {y}_{t} - \widehat{y}_{t \mid  t - 1}^{\left( k\right) }\right) ,\notag
\end{align}
where the definitions of $ \widehat{y}_{t \mid t- 1}^{(k)}$ and $ \widehat{\Sigma}_{t \mid t- 1}^{(k)}$ are the same as that of ${y}_{t \mid t- 1}^{(k)}$ and $ {\Sigma}_{t \mid t- 1}^{(k)},$  with the only requirement being to replace the parameters with $\widehat{\Theta}$ .

Similarly to equation (\ref{Th1_2_1}) in the proof of Theorem 1, define ${\widetilde{\Lambda }}_{k} = {\widehat{\Lambda }}_{k}{\left( {\widehat{V}}_{t|t - 1}^{\left( k\right) }\right) }^{\frac{1}{2}}$. Then we have
\begin{align}
  & {\left( {y}_{t} - {\widehat{y}}_{t | t-1}^{\left( k\right) }\right) }^{\prime }{\left( {\widehat{\Sigma}}_{t | t-1}^{\left( k\right) }\right) }^{-1}\left( {{y}_{t} - {\widehat{y}}_{t | t-1}^{\left( k\right) }}\right) \notag\\
  =&  {\widehat{\sigma }}^{-2}{\begin{Vmatrix}{M}_{{\widetilde{\Lambda }}_{k}}\left( {y}_{t} - {\widehat{y}}_{t \mid  t - 1}^{\left( k\right) }\right) \end{Vmatrix}}_{F}^{2} \notag\\
   +& {\left( {y}_{t} - {\widehat{y}}_{t \mid  t - 1}^{\left( k\right) }\right) }^{\prime }{\widetilde{\Lambda }}_{k}{\left( {\widehat{\sigma }}^{2}{I}_{r} + {\widetilde{\Lambda }}_{k}^{\prime }{\widetilde{\Lambda }}_{k}\right) }^{-1}{\left( {\widetilde{\Lambda }}_{k}^{\prime }{\widetilde{\Lambda }}_{k}\right) }^{-1}{\widetilde{\Lambda }}_{k}^{\prime}\left( {{y}_{t} - {\widehat{y}}_{t \mid  t - 1}}\right).
\end{align}
It further implies that
\begin{align}
  & \log p\left( {{y}_{t} \mid  {s}_{t} = k,\mathcal{Y}_{t-1};\widehat{\Theta }}\right)  \notag\\
  &=  - \frac{pq}{2}\log \left( {2\pi}\right)  - \frac{1}{2}\log \left| {{\widehat{\sigma }}^{2}{I}_{pq} + {\widetilde{\Lambda }}_{k}{\widetilde{\Lambda }}_{k}^{\prime}}\right| - \frac{1}{2}{\widehat{\sigma }}^{-2}{\begin{Vmatrix}{M}_{{\widetilde{\Lambda }}_{k}}\left( {y}_{t} - {\widehat{y}}_{t | t-1}^{\left( k\right) }\right) \end{Vmatrix}}_{F}^{2} \notag\\
   &- \frac{1}{2}\left( {{y}_{t} - {\widehat{y }}_{t | t-1}^{\left( k\right) }}\right)^{\prime} {\widetilde{\Lambda }}_{k}{\left( {\widehat{\sigma }}^{2}{I}_{r} + {\widetilde{\Lambda }}_{k}^{\prime}{\widetilde{\Lambda }}_{k}\right) }^{-1}{\left( {\widetilde{\Lambda }}_{k}^{\prime}{\widetilde{\Lambda }}_{k}\right) }^{-1}{\widetilde{\Lambda }}_{k}^{\prime}\left( {{y}_{t} - {\widehat{y}}_{t | t-1}^{\left( k\right) }}\right) ,
\end{align}
\begin{align}
   & \log p\left( {{y}_{t} \mid  {s}_{t} = j,\mathcal{Y}_{t-1};\widehat{\Theta }}\right)  \notag\\
   &=   - \frac{pq}{2}\log \left( {2\pi}\right)  - \frac{1}{2}\log \left| {{\widehat{\sigma }}^{2}{I}_{pq} + {\widetilde{\Lambda }}_{j}{\widetilde{\Lambda }}_{j}^{\prime}}\right|  - \frac{1}{2}{\widehat{\sigma }}^{-2}{\begin{Vmatrix}{M}_{{\widetilde{\Lambda }}_{j}}\left( {y}_{t} - {\widehat{y}}_{t | t-1}^{\left( j\right) }\right) \end{Vmatrix}}_{F}^{2} \notag\\
   &- \frac{1}{2}\left( {{y}_{t} - {\widehat{y }}_{t | t-1}^{\left( j\right) }}\right)^{\prime} {\widetilde{\Lambda }}_{j}\left(\widehat{\sigma }^{2}{I}_{r} + {\widetilde{\Lambda }}_{j}^{\prime}{\widetilde{\Lambda }}_{j}\right)^{-1}{\left( {\widetilde{\Lambda }}_{j}^{\prime}{\widetilde{\Lambda }}_{j}\right) }^{-1}{\widetilde{\Lambda }}_{j}^{\prime}\left( {{y}_{t} - {\widehat{y}}_{t | t-1}^{\left( j\right) }}\right).
\end{align}
Further by
\[
- \frac{1}{2}\log \left| {{\widehat{\sigma }}^{2}{I}_{pq} + {\widetilde{\Lambda }}_{k}{\widetilde{\Lambda }}_{k}^{\prime }}\right|  + \frac{1}{2}\log \left| {{\widehat{\sigma }}^{2}{I}_{pq} + {\widetilde{\Lambda }}_{j}{\widetilde{\Lambda }}_{j}^{\prime }}\right|  =  - \frac{1}{2}\log \frac{ \left| {I}_{r} + \frac{1}{{\widehat{\sigma }}^{2}}{\widetilde{\Lambda }}_{k}^{\prime}{\widetilde{\Lambda }}_{k}\right|}{\left|{{I}_{r} + \frac{1}{{\widehat{\sigma }}^{2}}{\widetilde{\Lambda }}_{j}^{\prime }{\widetilde{\Lambda }}_{j}}\right|},
\] then
\begin{align}
   &\log p\left( {{y}_{t} \mid  {s}_ {t}=k,\mathcal{Y}_{t-1};\widehat{\Theta }}\right)  - \log p\left( {{y}_{t} \mid  {s}_{t} = j ,\mathcal{Y}_{t-1};\widehat{\Theta }}\right)  \notag\\
   &=  - \frac{1}{2}\log \left| {{I}_{r} + \frac{1}{\widehat{\sigma }^{2}}{\widetilde{\Lambda }}_{k}^{\prime }{\widetilde{\Lambda }}_{k}}\right|  + \frac{1}{2}\log \left| {{I}_{r} + \frac{1}{\widehat{\sigma }^{2}}{\widetilde{\Lambda }}_{j}^{\prime }\widetilde{\Lambda }_{j}}\right| \notag\\
   &+ \frac{1}{2}{\widehat{\sigma }}^{-2}{\begin{Vmatrix} {M}_{{\widetilde{\Lambda }}_{j}}\left( {{y}_{t} - {\widehat{y}}_{t \mid  t - 1}^{\left( j\right) }}\right) \end{Vmatrix}}_{F}^{2} - \frac{1}{2}{\widehat{\sigma }}^{-2}{\begin{Vmatrix}{M}_{{\widetilde{\Lambda }}_{j}^{0}}\left( {{y}_{t} - {y}_{t \mid  t - 1}^{0\left( j\right) }}\right)\end{Vmatrix}}_{F}^{2} \notag\\
   &+ \frac{1}{2}{\widehat{\sigma }}^{-2}{\begin{Vmatrix}{M}_{\widetilde{\Lambda }_{k}^{0}}\left( {y}_{t} - {y}_{t |t-1}^{0\left( k\right) }\right) \end{Vmatrix}}_{F}^{2} - \frac{1}{2}{\widehat{\sigma }}^{-2}{\begin{Vmatrix}{M}_{{\widetilde{\Lambda }}_{k}}\left( {y}_{t} - {\widehat{y}}_{t |t- 1}^{\left(k\right) }\right) \end{Vmatrix}}_{F}^{2} \notag\\
   &+ \frac{1}{2}{\widehat{\sigma }}^{-2}{\begin{Vmatrix} {M}_{{\widetilde{\Lambda }}_{j}^{0}}\left( {{y}_{t} - {y}_{t \mid  t - 1}^{0\left( j\right) }}\right) \end{Vmatrix}}_{F}^{2} - \frac{1}{2}{\widehat{\sigma }}^{-2}{\begin{Vmatrix}{M}_{{\widetilde{\Lambda }}_{k}^{0}}\left( {{y}_{t} - {y}_{t \mid  t - 1}^{0\left( k\right) }}\right)\end{Vmatrix}}_{F}^{2} \notag\\
   & - \frac{1}{2}\left( {{y}_{t} - {\widehat{y}}_{t \mid  t - 1}^{\left( k\right) }}\right)^{\prime} {\widetilde{\Lambda }}_{k}{\left( {\widehat{\sigma }}^{2}{I}_{r} + {\widetilde{\Lambda }}_{k}^{\prime}{\widetilde{\Lambda }}_{k}\right) }^{-1}{\left( {\widetilde{\Lambda }}_{k}^{\prime}{\widetilde{\Lambda }}_{k}\right) }^{-1}{\widetilde{\Lambda }}_{k}^{\prime}\left( {{y}_{t} - {\widehat{y}}_{t \mid  t - 1}^{\left( k\right) }}\right) \notag \\
   &+ \frac{1}{2}\left( {{y}_{t} - {\widehat{y}}_{t \mid  t - 1}^{\left( j\right) }}\right)^{\prime} {\widetilde{\Lambda }}_{j}{\left( {\widehat{\sigma }}^{2}{I}_{r} + {\widetilde{\Lambda }}_{j}^{\prime}{\widetilde{\Lambda }}_{j}\right) }^{-1}{\left( {\widetilde{\Lambda }}_{j}^{\prime}{\widetilde{\Lambda }}_{j}\right) }^{-1}{\widetilde{\Lambda }}_{j}^{\prime}\left( {{y}_{t} - {\widehat{y}}_{t \mid  t - 1}^{\left( j\right) }}\right)
\end{align}
Since
\begin{align}\label{Th2_1}
  &  {\begin{Vmatrix} {M}_{\widetilde{\Lambda}_{j }^{0}}\left( {{y}_{t} - {y}_{t |t-1}^{0\left( j\right) }}\right) \end{Vmatrix}}_{F}^{2} - {\begin{Vmatrix} {M}_{\widetilde{\Lambda }_{k}^{0}}\left( {{y}_{t} - {y}_{t|t-1}^{0\left( k\right) }}\right) \end{Vmatrix}}_{F}^{2} \notag\\
  =& {\begin{Vmatrix}{M}_{{\Lambda }_{j}^{0}}{y}_{t} - {M}_{{\Lambda }_{j}^{0}}{\Lambda }_{j}^{0}{f}_{t |t - 1}^{0\left( j\right) }\end{Vmatrix}}_{F}^{2} - {\begin{Vmatrix}{M}_{{\Lambda }_{k}^{0}}\left( {\Lambda }_{k}^{0}{f}_{t} + {e}_{t} - {\Lambda }_{k}^{0}{f}_{t |t - 1}^{0\left( k\right) }\right) \end{Vmatrix}}_{F}^{2} \notag\\
  =& {\begin{Vmatrix}{M}_{{\Lambda}_{j}^{0}}{y}_{t}\end{Vmatrix}}_{F}^{2} - {\begin{Vmatrix}{M}_{{\Lambda}_{k}^{0}}{e}_{t}\end{Vmatrix}}_{F}^{2} \notag\\
  = &{y}_{t}^{\prime }{M}_{{\Lambda }_{j}^{0}}{y}_{t} - {e}_{t}^{\prime }{M}_{{\Lambda }_{k}^{0}}{e}_{t} \notag\\
  =& {\left( {\Lambda }_{k}^{0}{f}_{t} + {e}_{t}\right) }^{\prime }{M}_{{\Lambda }_{j}^{0}}\left( {{\Lambda }_{k}^{0}{f}_{t} + {e}_{t}}\right)  - {e}_{t}^{\prime }{M}_{{\Lambda }_{k}^{0}}{e}_{t}\notag\\
  =& {f}_{t}^{\prime }{\Lambda }_{k}^{0\prime }{M}_{{\Lambda }_{j}^{0}}{\Lambda }_{k}^{0}{f}_{t} + 2{e}_{t}^{\prime}{M}_{{\Lambda }_{j}^{0}}{\Lambda }_{k}^{0}{f}_{t} + {e}_{t}^{\prime}{M}_{{\Lambda }_{j}^{0}}{e}_{t} - {e}_{t}^{\prime}{M}_{{\Lambda }_{k}^{0}}{e}_{t} \notag\\
  = &{f}_{t}^{\prime }{\Lambda }_{k}^{0\prime }{M}_{{\Lambda }_{j}^{0}}{\Lambda }_{k}^{0}{f}_{t} + 2{e}_{t}^{\prime }{M}_{{\Lambda }_{j}^{0}}{\Lambda }_{k}^{0}{f}_{t} + {e}_{t}^{\prime }{P}_{{\Lambda }_{k}^{0}}{e}_{t} - {e}_{t}^{\prime }{P}_{{\Lambda }_{j}^{0}}{e}_{t} \notag\\
  \geq&   - {e}_{t}^{\prime }{P}_{{\Lambda }_{j}^{0}}{e}_{t} + 2{e}_{t}^{\prime }{M}_{{\Lambda }_{j}^{0}}{\Lambda }_{k}^{0}{f}_{t} + {f}_{t}^{\prime }{\Lambda }_{k}^{0\prime}{M}_{{\Lambda }_{j}^{0}}{\Lambda }_{k}^{0}{f}_{t},
\end{align}
then
\begin{align}
  & \log p\left({y}_{t} \mid  {s}_{t} = k,\mathcal{Y}_{t-1};\widehat{\Theta }\right)  - \log p\left( {{y}_{t} \mid  {s}_{t} = j,{\mathcal{Y}}_{t-1};\widehat{\Theta }}\right)  \notag\\
  \geq &  - \frac{1}{2}\log \left| {{I}_{r} + \frac{1}{\widehat{\sigma }^{2}}{\widetilde{\Lambda }}_{k}^{\prime}{\widetilde{\Lambda }}_{k}}\right| \notag\\
  - &\frac{1}{2}{\left( {y}_{t} - {\widehat{y }}_{t|t-1}^{\left( k\right) }\right) }^{\prime }\widehat{\Lambda }_{k}{\left( {\widehat{\sigma}}^{2}{I}_{r} + {\widetilde{\Lambda }}_{k}^{\prime }\widetilde{\Lambda }_{k}\right) }^{-1}{\left( {\widetilde{\Lambda }}_{k}^{\prime }\widetilde{\Lambda }_{k}\right) }^{-1}{\widetilde{\Lambda }}_{k}^{\prime }\left( {{y}_{t} - {\widehat{y}}_{t|t-1}^{\left( k\right) }}\right) \notag\\
  +& \frac{1}{2}{\widehat{\sigma }}^{-2}{\begin{Vmatrix} {M}_{\widetilde{\Lambda}_{j }}\left( {y}_{t} - \widehat{y}_{t|t-1}^{\left( j\right) }\right) \end{Vmatrix}}_{F}^{2} - \frac{1}{2}{\widehat{\sigma }}^{-2}{\begin{Vmatrix} {M}_{\widetilde{\Lambda}_{j }^{0}}\left( {y}_{t} - {y}_{t|t-1}^{0\left( j\right) }\right) \end{Vmatrix}}_{F}^{2} \notag\\
  +& \frac{1}{2}{\widehat{\sigma}}^{-2}{\begin{Vmatrix} {M}_{\widetilde{\Lambda}_{k}^{0}}\left( {y}_{t} - {y}_{t|t-1}\right) \end{Vmatrix}}_{F}^{2} - \frac{1}{2}{\widehat{\sigma}}^{-2}{\begin{Vmatrix} {M}_{\widetilde{\Lambda}_{k }}\left( {y}_{t} - \widehat{y}_{t|t-1}^{\left(k\right)}\right) \end{Vmatrix}}_{F}^{2} \notag\\
  -&\frac{1}{2}\widehat{\sigma}^{-2}{e}_{t}^{\prime}{P}_{{\Lambda }_{j}^{0}}{e}_{t} + \widehat{\sigma }^{-2}{e}_{t}^{\prime}{M}_{{\Lambda }_{j}^{0}}{\Lambda }_{k}^{0}{f}_{t} + \frac{1}{2}\widehat{\sigma}^{-2}{f}_{t}{}^{\prime }{\Lambda }_{k}^{0\prime }{M}_{{\Lambda }_{j}^{0}}{\Lambda }_{k}^{0}{f}_{t},
\end{align}
where the inequility follows from (\ref{Th2_1}) and throwing away \(\frac{1}{2}\log \left| {{I}_{r} + \frac{1}{\widehat{\sigma}^{2}}{\widetilde{\Lambda}_{j}}^{\prime }{\widetilde{\Lambda}_{j}}}\right|\) and \(\frac{1}{2}{\left( {y}_{t} - {\widehat{y}}_{t|t-1}^{\left( j\right) }\right) }^{\prime }{\widetilde{\Lambda}}_{j}\left( {\widehat{\sigma}}^{2}{I}_{r} + {\widetilde{\Lambda }}_{j}^{\prime }\widetilde{\Lambda }_{j}\right)^{-1}{\left( {\widetilde{\Lambda }}_{j}^{\prime }\widetilde{\Lambda }_{j}\right) }^{-1}{\widetilde{\Lambda }}_{j}^{\prime }\left( {{y}_{t} - {\widehat{y}}_{t|t-1}^{\left( j\right) }}\right)\). It follows that
\begin{align}\label{Th2_2}
   & \mathop{\min }\limits_{t}\left\lbrack  {\log p\left( {{y}_{t} \mid  {\mathcal{Y}}_{t-1},{s}_{t}=k;\widehat{\Theta}}\right)  - \log p\left( {{y}_{t} \mid  {\mathcal{Y}}_{t-1},{s}_{t} = j;\widehat{\Theta}}\right) }\right\rbrack \notag\\
    \geq & - \left( {{A}_{1} + {A}_{2} + {A}_{3} + {A}_{4} + {A}_{5} + {A}_{6}}\right) + \frac{1}{2}\widehat{\sigma }^{-2}\mathop{\min }\limits_{t}{f}_{t}^{\prime}{{\Lambda}_{k}^{0}}^{\prime}{M}_{{\Lambda}_{j}^{0}}{\Lambda}_{k}^{0}{f}_{t},
\end{align}
where
\begin{align}
    {A}_{1} & =\frac{1}{2}\log \left| {{I}_{r} + \frac{1}{\widehat{\sigma }^{2}}{\widetilde{\Lambda }}_{k}^{\prime}{\widetilde{\Lambda }}_{k}}\right|, \notag\\
{A}_{2}&=\frac{1}{2}\mathop{\sup }\limits_{t}{\left( {y}_{t} - {\widehat{y}}_{t|t - 1}^{\left( k\right) }\right) }^{\prime }\widetilde{\Lambda }_{k}{\left( {\widehat{\sigma }}^{2 }{I}_{r} + {\widetilde{\Lambda}}_{k}^{\prime}{\widetilde{\Lambda }}_{k}\right) }^{-1}\left( {\widetilde{\Lambda }}_{k}^{\prime}{\widetilde{\Lambda }}_{k}\right)^{-1}{\widetilde{\Lambda }_{k}}^{\prime}\left( {{y}_{t} - {\widehat{y}}_{t|t - 1}^{\left( k\right) }}\right),\notag\\
{A}_{3}&=\frac{1}{2}\widehat{\sigma }^{-2}\mathop{\sup }\limits_{t}\left| {\begin{Vmatrix}{M}_{\widetilde{\Lambda}_{j}}\left( {{y}_{t} - {\widehat{y}}_{t|t-1}^{\left( j\right) }}\right) \end{Vmatrix}}_{F}^{2} - {\begin{Vmatrix}{M}_{{\widetilde{\Lambda }}_{j}^{0}}\left( {{y}_{t} - {y}_{t|t-1}^{0\left( j\right) }}\right) \end{Vmatrix}}_{F}^{2}\right|,\notag\\
{A}_{4}&=\frac{1}{2}\widehat{\sigma }^{-2}\mathop{\sup }\limits_{t}\left| {\begin{Vmatrix}{M}_{\widetilde{\Lambda}_{k}^{0}}\left( {{y}_{t} - {y}_{t|t-1}^{0\left( k\right) }}\right) \end{Vmatrix}}_{F}^{2} - {\begin{Vmatrix}{M}_{{\widetilde{\Lambda }}_{k}}\left( {{y}_{t} - \widehat{y}_{t|t-1}^{\left( k\right) }}\right) \end{Vmatrix}}_{F}^{2}\right|,\notag\\
{A}_{5}&=\frac{1}{2}\widehat{\sigma }^{-2}\mathop{\sup }\limits_{t}{e}_{t}^{\prime}{P}_{{\Lambda}_{j}^{0}}{e}_{t},\notag\\
{A}_{6}&=\frac{1}{2}\mathop{\sup }\limits_{t}\left| {{e}_{t}^{\prime}{M}_{{\Lambda}_{j}^{0}}{\Lambda}_{k}^{0}{f}_{t}}\right|.\notag
\end{align}
Thus, for expression (\ref{Th2_0}), it suffices to show
\[
{P}_{r}\left( \frac{1}{2}{\widehat{\sigma }}^{-2}\mathop{\min }\limits_{t}{f}_{t}^{\prime}{\Lambda }_{k}^{0\prime}{M}_{{\Lambda}_{j}^{0}}{\Lambda }_{k}^{0}{f}_{t} \leq  {A}_{1} + {A}_{2} + {A}_{3} + {A}_{4} + {A}_{5} + {A}_{6} + \eta \log N\right) \rightarrow  0.
\]
By Assumption B(2), \(\mathop{\min }\limits_{t}{f}_{t}^{\prime}{\Lambda }_{k}^{0\prime}{M}_{{\Lambda}_{j}^{0}}{\Lambda }_{k}^{0}{f}_{t} \geq  N{c}_{0}\) for some \({c}_{0} > 0\). Thus, it suffices to show that \({A}_{1},\ldots ,{A}_{6}\) are all \({o}_{p}\left( N\right)\).

(1) Consider \({A}_{1}\).
\begin{align}
{A}_{1}  & = \frac{1}{2}\log \left| {{I}_{r} + \frac{1}{\widehat{\sigma }^{2}}{\widetilde{\Lambda }}_{k}^{\prime}{\widetilde{\Lambda }}_{k}}\right| \notag\\
&= \frac{1}{2}\log \left| {{I}_{r} + \frac{1}{\widehat{\sigma }^{2}}{\left( {\widehat{V}}_{t|t-1}^{\left( k\right) }\right) }^{\frac{1}{2}}{\widehat{\Lambda }}_{k}^{\prime }{\widehat{\Lambda }}_{k}{\left( {\widehat{V}}_{t|t-1}^{\left( k\right) }\right) }^{\frac{1}{2}}}\right|\notag\\
& = {O}_{p}\left( {\log \left( {pq}\right) }\right)  = {o}_{p}\left( {N}\right) .
\end{align}

(2) Consider \({A}_{2}\).
\begin{align}
{A}_{2} &= \frac{1}{2}\mathop{\sup }\limits_{t}{\left( {y}_{t} - {\widehat{y}}_{t|t-1}^{\left( k\right) }\right) }^{\prime }\widetilde{\Lambda}_{k}{\left( {\widehat{\sigma }}^{2}{I}_{r} + {\widetilde{\Lambda}}_{k}^{\prime}{\widetilde{\Lambda }}_{k}\right) }^{-1}\left( {\widetilde{\Lambda }}_{k}^{\prime}{\widetilde{\Lambda }}_{k}\right)^{-1}{\widetilde{\Lambda }_{k}}^{\prime}\left( {{y}_{t} - {\widehat{y}}_{t|t-1}^{\left( k\right) }}\right)\notag\\
&= \frac{1}{2}\mathop{\sup }\limits_{t}{\left( {y}_{t} - {\widehat{y}}_{t|t-1}^{\left( k\right) }\right) }^{\prime }\widetilde{\Lambda}_{k}{\left( {\widetilde{\Lambda }}_{k}^{\prime}{\widetilde{\Lambda }}_{k}\right) }^{-\frac{1}{2}}{\left( {\widetilde{\Lambda }}_{k}^{\prime}{\widetilde{\Lambda }}_{k}\right) }^{\frac{1}{2}}{\left( {\widehat{\sigma }}^{2}{I}_{r} + {\widetilde{\Lambda}}_{k}^{\prime}{\widetilde{\Lambda }}_{k}\right) }^{-1}{\left( {\widetilde{\Lambda }}_{k}^{\prime}{\widetilde{\Lambda }}_{k}\right) }^{-\frac{1}{2}}\notag\\
&~~~~{\left( {\widetilde{\Lambda }}_{k}^{\prime}{\widetilde{\Lambda }}_{k}\right) }^{-\frac{1}{2}}{\widetilde{\Lambda }}_{k}^{\prime}{\left( {y}_{t} - {\widehat{y}}_{t|t-1}^{\left( k\right) }\right) }\notag\\
&\leq  \frac{1}{2}\mathop{\sup }\limits_{t}{\begin{Vmatrix}{\left( {\widetilde{\Lambda }}_{k}^{\prime}{\widetilde{\Lambda }}_{k}\right) }^{-\frac{1}{2}}{\widetilde{\Lambda }}_{k}^{\prime}\left( {{y}_{t} - {\widehat{y}}_{t|t-1}^{\left( k\right) }}\right) \end{Vmatrix}}^{2}{\begin{Vmatrix}{\left( {\widetilde{\Lambda }}_{k}^{\prime}{\widetilde{\Lambda }}_{k}\right) }^{\frac{1}{2}}{\left( {\widehat{\sigma }}^{2}{I}_{r} + {\widetilde{\Lambda}}_{k}^{\prime}{\widetilde{\Lambda }}_{k}\right)}^{-1}{\left( {\widetilde{\Lambda }}_{k}^{\prime}{\widetilde{\Lambda }}_{k}\right) }^{-\frac{1}{2}}\end{Vmatrix}}\notag\\
&= \frac{1}{2}\mathop{\sup }\limits_{t}{\begin{Vmatrix}{P}_{{\widetilde{\Lambda }}_{k}}\left( {{y}_{t} - {\widehat{y}}_{t|t- 1}^{\left( k\right) }}\right) \end{Vmatrix}}^{2}\begin{Vmatrix}{ ({\widehat{\sigma }}^{2}{I}_{r}+ \widetilde{\Lambda }_{k}^{\prime}\widetilde{\Lambda }_{k})^{-1}}\end{Vmatrix} \notag\\
& \leq  \frac{1}{2}\mathop{\sup }\limits_{t}{\begin{Vmatrix} {{y}_{t} - {\widehat{y}}_{t|t- 1}^{\left( k\right) }}\end{Vmatrix}}_{F}^{2} \begin{Vmatrix}{ ({\widehat{\sigma }}^{2}{I}_{r}+ \widetilde{\Lambda }_{k}^{\prime}\widetilde{\Lambda }_{k})^{-1}}\end{Vmatrix} \notag.
\end{align}
By Assumption A(2) and B(1), \(\mathop{\sup }_{t} {\begin{Vmatrix}{\Lambda }_{k}^{0}{f}_{t}\end{Vmatrix}}^{\alpha } \leq  {\begin{Vmatrix}{\Lambda }_{k}^{0}\end{Vmatrix}}^{\alpha }
 \mathop{\sum }_{{t = 1}}^{n}{\begin{Vmatrix}{f}_{t}\end{Vmatrix}}^{\alpha }\)\(  = {O}_{p}\left( {p}^{\frac{\alpha }{2}}{q}^{\frac{\alpha}{2}}{ n}\right).\)
By Holder inequality, \({\begin{Vmatrix} {e}_{t}\end{Vmatrix}}^{2} = \mathop{\sum }_{i = 1}^{p}\mathop{\sum }_{{j = 1}}^{q}{e}_{t,{ij}}^{2} \leq {\left( \mathop{\sum }_{{i = 1}}^{p}\mathop{\sum }_{{j = 1}}^{q}{\left( {e}_{t,ij}^{2}\right) }^{\frac{\alpha}{2}}\right)}^{\frac{2}{\alpha}} {N }^{1-\frac{2}{\alpha}},\)  thus \(\mathop{\sup }_{t} {\begin{Vmatrix}{e}_{t}\end{Vmatrix}}^{\alpha} \leq  {N}^{\frac{\alpha }{2} - 1}\mathop{\sup }_{t}\left( {\mathop{\sum }_{{i = 1}}^{p}\mathop{\sum }_{{j = 1}}^{q}{e}_{t,ij}^{\alpha}}\right)  \leq  {N}^{\frac{\alpha }{2} - 1}\mathop{\sum }_{{t = 1}}^{n}\left( {\mathop{\sum }_{{i = 1}}^{p}\mathop{\sum }_{{j = 1}}^{q}{e}_{t,{ij}}^{\alpha }}\right)  = {O}_{p}\left( {{N}^{\frac{\alpha }{2}}n}\right)  = {O}_{p}\left( {{p}^{\frac{\alpha }{2}}{q}^{\frac{\alpha }{2}}n}\right)\) by Assumption C(1). It follows that \(\mathop{\sup }_{t}{\begin{Vmatrix}{y}_{t}\end{Vmatrix}}  \leq  \mathop{\sup }_{t}{\begin{Vmatrix} {\Lambda}_{k}^{0}{f}_{t}\end{Vmatrix}}  + \mathop{\sup }_{t}{\begin{Vmatrix} {e}_{t}\end{Vmatrix}}  = {O}_{p}\left({{p}^{\frac{1}{2}}{q}^{\frac{1}{2}}{n}^{\frac{1}{\alpha}}}\right).\)
Thus,
\begin{align}
{A}_{2}& \leq  {\frac{1}{2}} \left\lbrack  {\mathop{\sup }\limits_{t} {\begin{Vmatrix} {y}_{t}\end{Vmatrix}}^{2} + \mathop{\sup }\limits_{t} {\begin{Vmatrix} \widehat{y}_{t|t-1}^{(k)}\end{Vmatrix}}^{2}}\right\rbrack  \begin{Vmatrix}{ ({\widehat{\sigma }}^{2}{I}_{r}+ \widetilde{\Lambda }_{k}^{\prime}\widetilde{\Lambda }_{k}})^{-1}\end{Vmatrix} \notag\\
&= {O}_{p}\left( {{pq}{n}^{\frac{2}{\alpha}}}\right)  \cdot  {O}_{p}\left( {{p}^{-1}{q}^{-1}}\right)  = {O}_{p}\left( {n}^{\frac{2}{\alpha}}\right) \notag\\
&= {o}_{p}\left( N\right)   ~~\text{ when }\frac{{n}^{\frac{2}{\alpha }}}{N} \rightarrow  0\text{ and } \alpha > 2 . \notag
\end{align}

(3)Consider \({A}_{3}\).
\begin{align}
{A}_{3}& =\frac{1}{2}\widehat{\sigma }^{-2}\mathop{\sup }\limits_{t}\left| {\begin{Vmatrix}{M}_{\widetilde{\Lambda}_{j}}\left( {{y}_{t} - {\widehat{y}}_{t|t-1}^{\left( j\right) }}\right) \end{Vmatrix}}^{2} - {\begin{Vmatrix}{M}_{{\Lambda }_{j}^{0}}\left( {{y}_{t} - {y}_{t|t-1}^{0\left( j\right) }}\right) \end{Vmatrix}}^{2}\right| \notag\\
&= \frac{1}{2}{\widehat{\sigma }}^{-2}\mathop{\sup }\limits_{t}\left| {\begin{Vmatrix}{M}_{\widehat{\Lambda}_{j}}{y}_{t} \end{Vmatrix}}^{2} - {\begin{Vmatrix}{M}_{{\Lambda }_{j}^{0}}{y}_{t} \end{Vmatrix}}^{2}\right|\notag\\
&= \frac{1}{2}{\widehat{\sigma }}^{-2}\mathop{\sup }\limits_{t}\left| {{y}_{t}{}^{\prime }{P}_{\widehat{\Lambda}_{j}}{y}_{t} - {y}_{t}{}^{\prime }{P}_{{\Lambda}_{j}^{0}}{y}_{t}}\right|
\leq  \frac{1}{2}{\widehat{\sigma }}^{-2}\begin{Vmatrix}{{P}_{\widehat{\Lambda}_{j}} - {P}_{{\Lambda}_{j}^{0}}}\end{Vmatrix} \cdot  \mathop{\sup }\limits_{t}{\begin{Vmatrix}{y}_{t}\end{Vmatrix}}^{2}\notag\\
&= {O}_{p}\left( pq{n}^{\frac{2}{\alpha }}\right) \begin{Vmatrix}{{P}_{\widehat{\Lambda}_{j}} - {P}_{{\Lambda}_{j}^{0}}}\end{Vmatrix} \notag.
\end{align}
Since,
\begin{align}
{\begin{Vmatrix}{{P}_{\widehat{\Lambda}_{j}} - {P}_{{\Lambda}_{j}^{0}}}\end{Vmatrix}}^{2}& \leq  {\begin{Vmatrix}{{P}_{\widehat{\Lambda}_{j}} - {P}_{{\Lambda}_{j}^{0}}}\end{Vmatrix}}_{F}^{2} = {tr}\left\lbrack  {\left( {P}_{\widehat{\Lambda}_{j}} - {P}_{{\Lambda}_{j}^{0}}\right) }^{2}\right\rbrack \notag\\
&= 2{tr}\left\lbrack  {{I}_{r} - {P}_{\widehat{\Lambda}_{j}}{P}_{{\Lambda}_{j}^{0}}}\right\rbrack\notag\\
&= 2\left( {{tr}\left( {I}_{r}\right) - {tr}\left\lbrack  {P}_{{\Lambda}_{j}^{0}}\right\rbrack + {tr}\left\lbrack  {{M}_{\widehat{\Lambda}_{j}}{P}_{{\Lambda}_{j}^{0}}}\right\rbrack  }\right)\notag\\
&= {2tr}\left\lbrack  {{M}_{\widehat{\Lambda }_{j}}{P}_{{\Lambda }_{j}^{0}}}\right\rbrack   = {2tr}\left\lbrack  {{M}_{\widehat{\Lambda }_{j}}{\Lambda }_{j}^{0}{\left( {\Lambda }_{j}^{0\prime }{\Lambda }_{j}^{0}\right) }^{-1}{\Lambda }_{j}^{0\prime}}\right\rbrack \notag\\
&= 2{tr}\left\lbrack  {{\left( {\Lambda }_{j}^{0\prime }{\Lambda }_{j}^{0}\right) }^{-\frac{1}{2}}{\Lambda }_{j}^{0\prime }{M}_{\widehat{\Lambda }_{j}}{\Lambda }_{j}^{0}{\left( {\Lambda }_{j}^{0\prime }{\Lambda }_{j}^{0}\right) }^{-\frac{1}{2}}}\right\rbrack\notag\\
&= 2{\begin{Vmatrix}{M}_{\widehat{\Lambda }_{j}}{\Lambda }_{j}^{0}{\left( {\Lambda }_{j}^{0\prime }{\Lambda }_{j}^{0}\right) }^{-\frac{1}{2}}\end{Vmatrix}}_{F}^{2}\notag\\
&\leq  \frac{2}{pq}{\begin{Vmatrix}{M}_{\widehat{\Lambda }_{j}}{\Lambda }_{j}^{0}\end{Vmatrix}}_{F}^{2} = {O}_{p} \left( \frac{1}{\sqrt{{\delta}_{pqn}}}\right),\notag
\end{align}
then \({A}_{3} = {O}_{p}\left( {N \cdot  {n}^{\frac{2}{\alpha }}/{\delta}_{pqn}^{\frac{1}{4}}}\right)  = {o}_{p}\left( N\right)\) when \(\frac{{n}^{\frac{16}{\alpha }}}{N} \rightarrow  0\) and \(\alpha  > {16}\) . Similar to term \({A}_{3}\) , Term \({A}_{4}\) is also \({o}_{p}\left( N\right)\) .

(4) Consider \({A}_{5}\).
\begin{align}
{A}_{5} &= \frac{1}{2}\widehat{\sigma }^{-2}\mathop{\sup }\limits_{t}{e}_{t}^{\prime}{P}_{{\Lambda}_{j}^{0}}{e}_{t} \notag\\
&= \frac{1}{2}\widehat{\sigma }^{-2}\mathop{\sup }\limits_{t}{e}_{t}^{\prime}{\Lambda}_{j}^{0}{\left( {\Lambda }_{j}^{0\prime}{\Lambda }_{j}^{0 }\right) }^{-1}{\Lambda }_{j}^{0\prime}{e}_{t}\notag\\
&= \frac{1}{2}\widehat{\sigma }^{-2}\mathop{\sup }\limits_{t}\frac{{\begin{Vmatrix} {e}_{t}^{\prime}{\widetilde{\Lambda}}_{j}^{0}\end{Vmatrix}}^{2}}{N}\notag,
\end{align}
where $\widetilde{\Lambda}_{j}^{0}= \widetilde{C}_j^{(0)}\otimes  \widetilde{R}_j^{(0)}$ with $\widetilde{C}_j^{(0)} = C_j^0 (D_j^{(2)})^{-1/2}$ and $\widetilde{R}_j^{(0)} = R_j^0 (D_j^{(1)})^{-1/2}$.
By Assumption D(1), \(\mathop{\sup }\limits_{t}{\begin{Vmatrix}\frac{{e}_{t}^{\prime}{\widetilde{\Lambda}}_{j}^{0}}{\sqrt{N}}\end{Vmatrix}}^{\beta } \leq  \mathop{\sum }\limits_{{t = 1}}^{n}{\begin{Vmatrix}\frac{{e}_{t}^{\prime}{\widetilde{\Lambda}}_{j}^{0}}{\sqrt{N}}\end{Vmatrix}}^{\beta } =  {O}_{p}\left( n\right).\) Thus \({A}_{5} = {O}_{p}\left( {n}^{\frac{2}{\beta }}\right)  = {o}_{p}\left( N\right)\) when \(\frac{{n}^{\frac{\alpha }{\beta }}}{N} \rightarrow  0.\)

(5) Consider \({A}_{6}\).
\begin{align}
{A}_{6} &= \widehat{\sigma}^{-2}\mathop{\sup }\limits_{t}\left| {{e}_{t}^{\prime}{M}_{{\Lambda}_{j}^{(0)}}{\Lambda}_{k}^{0}  {f}_{t}}\right| \notag\\
&= \widehat{\sigma}^{-2}\mathop{\sup }\limits_{t}\left| {{e}_{t}^{\prime}{\Lambda }_{k}^{0}{f}_{t}} + {e}_{t}^{\prime }{\Lambda }_{j}^{0}{\left( {\Lambda }_{j}^{0\prime}{\Lambda }_{j}^{0}\right) }^{-1}{\Lambda }_{j}^{0\prime}{\Lambda }_{k}^{0}{f}_{t}\right|\notag\\
&\leq  \widehat{\sigma }^{-2}\mathop{\sup }\limits_{t}\left| {{e}_{t}^{\prime }{\Lambda }_{k}^{0}{f}_{t}}\right| + \frac{\widehat{\sigma }^{-2}}{pq}\mathop{\sup }\limits_{t}\left| {{e}_{t}^{\prime }{\widetilde{\Lambda} }_{j}^{0}{\widetilde{\Lambda} }_{j}^{0\prime}{\Lambda }_{k}^{0}{f}_{t}}\right| \notag.
\end{align}
By Assumption C(1), \(\mathop{\sup }_{t}{\begin{Vmatrix}{f}_{t}\end{Vmatrix}}^{\alpha } \leq  \mathop{\sum }_{{t = 1}}^{n}{\begin{Vmatrix} {f}_{t}\end{Vmatrix}}^{\alpha } = {O}_{p}\left( n\right)\). Thus,
\begin{align}
\mathop{\sup }\limits_{t}\left| {{e}_{t}^{\prime }{\Lambda }_{k}^{0}{f}_{t}}\right|
&\leq  \mathop{\sup }\limits_{t} {\begin{Vmatrix} {e}_{t}^{\prime }{\Lambda }_{k}^{0}\end{Vmatrix}}  \cdot \mathop{\sup }\limits_{t} {\begin{Vmatrix} {f}_{t}\end{Vmatrix}} \notag\\
&= {O}_{p}\left( {N}^{\frac{1}{2}}{n}^{\frac{1}{\beta}}\right)  \times  {O}_{p}\left( {n}^{\frac{1}{\alpha}}\right)\notag\\
&= {O}_{p}\left( {p}^{\frac{1}{2}}{q}^{\frac{1}{2}}{n }^{\frac{1}{\alpha} + \frac{1}{\beta }}\right)  = {o}_{p}\left( N\right) \text{when}\frac{{n}^{\frac{2}{\alpha } + \frac{2}{\beta} }}{N} \rightarrow  0.\notag
\end{align}
It follows that \({A}_{6} = {o}_{p}\left( N\right)\) when \(\frac{{n}^{\frac{2}{\alpha } + \frac{2}{\beta} }}{N} \rightarrow  0.\)

Step (2): We next prove \({\widehat{w}}_{t|n}^{\left(j\right)} = {o}_{p}\left( \frac{1}{{N}^{\eta}}\right)\) for \(j \neq  k\) when the true state is \({s}_{t} = k\)

Let \({P}_{\cdot\ell}\) denotes the \(\ell\) -th column of P. It easy to show that
\begin{align}
{w}_{t|n}^{\left(j\right)} &= {P}_{r}\left( {{s}_{t} = j \mid  \mathcal{Y}_{n}}\right)  = \mathop{\sum }\limits_{{\ell  = 1}}^{M}{P}_{r}\left( {{s}_{t} = j \mid  {s}_{t + 1} = \ell ,{\mathcal{Y}}_{n}}\right) {P}_{r}\left( {{s}_{t + 1} = \ell \mid  {\mathcal{Y}}_{n}}\right)\notag\\
&= \mathop{\sum }\limits_{{\ell  = 1}}^{M}{P}_{r}\left( {{s}_{t} = j \mid  {s}_{t + 1} = \ell ,{\mathcal{Y}}_{t}}\right) {w}_{t + 1|n}^{\left( \ell\right) }\notag\\
&= \mathop{\sum }\limits_{{\ell = 1}}^{M}\frac{{P}_{r}\left( {{s}_{t + 1} = \ell \mid  {s}_{t}=j,{\mathcal{Y}}_{t}}\right)  }{  {P}_{r}\left( {{s}_{t + 1} = \ell \mid  {\mathcal{Y}}_{t}}\right) } {P}_{r}\left( {{s}_{t}=j \mid  {\mathcal{Y}}_{t}}\right){w}_{t + 1|n}^{\left( \ell\right) }\notag\\
&= \mathop{\sum }\limits_{{\ell = 1}}^{M}\frac{{p}_{j\ell}{w}_{t | t}^{\left( j\right) }}{{w}_{t+1 \mid  t}^{\left(\ell\right) }}{w}_{t + 1 \mid  n}^{\left( \ell\right) }\notag\\
&= {w}_{t | t}^{\left( j\right) }{P}_{\cdot\ell}^{\prime }\left( {{w}_{t + 1|n} \ominus {w}_{t + 1|t}}\right)\notag,
\end{align}
where $w_{t+1|n}=(w_{t+1|n}^{(1)},\cdots, w_{t+1|n}^{(M)})^{\prime}$ and $ \ominus$ denotes element-wise division. Note that the second equality is due to that \({y}_{t + 1}\) is independent with \({s}_{t}\) when given \({s}_{t+1}\) and \({y}_{t}\). Furthermore, since $
{w}_{t + 1 \mid  t}^{\left( \ell\right) } = {P}_{r}\left( {{s}_{t + 1} = \ell \mid  {\mathcal{Y}}_{t}}\right)  = \mathop{\sum }\limits_{{i = 1}}^{M}{P}_{r}\left( {{s}_{t + 1} = \ell \mid  {s}_{t} = i,{\mathcal{Y}}_{t}}\right) {P}_{r}\left( {{s}_{t} = i \mid  {\mathcal{Y}}_{t}}\right)
\geq  {p}_{k\ell}{w}_{t|t}^{\left( k\right) },
$
then
\begin{align}
{\widehat{w}}_{t|n}^{\left( j\right) }
&= {\widehat{w}}_{t|t}^{\left( j\right) } \mathop{\sum }\limits_{{\ell = 1}}^{M}{\widehat{w}}_{t+1|n}^{\left( \ell\right) }\frac{{p}_{j\ell}}{{\widehat{w}}_{t + 1|t}^{\left( \ell\right) }} \notag\\
&\leq  {\widehat{w}}_{t | t}^{\left( j\right) }\mathop{\sum }\limits_{{\ell = 1}}^{M}{\widehat{w}}_{t + 1|n}^{\left( \ell\right) }\frac{{p}_{j\ell}}{{p}_{k\ell}{\widehat{w}}_{t|t}^{\left( k\right) }} \notag\\
&\leq  {\widehat{w}}_{t|t}^{(j)}\mathop{\max }\limits_{\ell}\frac{{p}_{j\ell}}{{p}_{k\ell}}\frac{1}{{\widehat{w}}_{t|t}^{(k)}} = {o}_{p}\left( \frac{1}{{N}^{\eta}}\right)  \notag,
\end{align}
where the last equality follows from step(1) and  \(\mathop{\min }\limits_{\ell}{p}_{k\ell} > 0\). We have finished the proof of the first result of Theorem2.

Consider the second result of Theorem2. Similar to expression (\ref{Th2_2}), it suffices to show

\[
{P}_{r}\left\lbrack  {\frac{1}{2}{\widehat{\sigma }}^{-2}{f}_{t}{}^{\prime }{{\Lambda }_{k}^{0}}^{\prime}{M}_{{\Lambda}_{j}^{0}}{\Lambda }_{k}^{0}{f}_{t} \leq  {A}_{1}^{\prime } + {A}_{2}^{\prime } + {A}_{3}^{\prime } + {A}_{4}^{\prime } + {A}_{5}^{\prime } + {A}_{6}^{\prime } + \eta\log N}\right\rbrack   \rightarrow  0,
\]
where \({A}_{1}^{\prime },\cdots ,{A}_{6}^{\prime }\) equals \({A}_{1},\cdots ,{A}_{6}\) without taking suppermum with respect to $t$.
Given the calculation of terms \({A}_{1},\ldots ,{A}_{6}\) , it is not difficult to see that without taking suppermum, \({A}_{1}^{\prime },\ldots ,{A}_{6}^{\prime }\) become \({O}_{p}\left( {\log N}\right) ,{O}_{p}\left( 1\right) ,{O}_{p}\left( N/{\delta}_{pqn}^{\frac{1}{4}}\right) ,{O}_{p}\left( N/{{\delta}_{pqn}^{\frac{1}{4}}}\right),\)  \({O}_{p}\left( 1\right)\) and \({O}_{p}\left( {N}^{\frac{1}{2}}\right)\) respectively. Since \({f}_{t}^{\prime }{\Lambda }_{k}^{0\prime }{M}_{{\Lambda }_{j}^{0}}{\Lambda }_{k}^{0\prime } {f}_{t} \geq  {N{c}_{0}}\) for some \({c}_{0}>0\) ,  \({A}_{1}^{\prime },\ldots ,{A}_{6}^{\prime }\) are all dominated by this term. Then we finish the proof of result(2) in Theorem2.

\end{proof}
\setcounter{equation}{0}
 \setcounter{subsection}{0}
 \setcounter{lemma}{0}
 \renewcommand{\theequation}{D.\arabic{equation}}
 \renewcommand{\thesubsection}{D.\arabic{subsection}}
  \renewcommand{\thelemma}{D.\arabic{lemma}}
\section*{D ~~ Details for Theorem 3}

Before presenting the proof of Theorem 3, we first provide the following three lemmas (Lemma \ref{lemaC1}- Lemma \ref{lemaC3}) and their proofs.
\begin{lemma}\label{lemaC1}
Let $\widehat{\theta}$ denote the quasi-maximum likelihood estimator (QMLE). Define
\(
\widehat{F}_{t|n}^{(k)} = E\left[F_t \mid s_t = k, \mathcal{Y}_n; \widehat{\theta}\right],
\)
$ \widehat{D}_k^{(1)}=\widehat{R}_k^{\prime}\widehat{R}_k/p$ and $\widehat{D}_k^{(2)}=\widehat{C}_k^{\prime}\widehat{C}_k/q$.
Under Assumptions (A)-(E) and asymptotic conditions
\(
{n^{16/\alpha}}/{(pq)} \to 0 \) and \(   {n^{2/\alpha + 2/\beta}}/{(pq)} \to 0 \), as $ p, q, n \to \infty,
$ we have
    $$\widehat{F}_{t|n}^{(k)}=\frac{1}{pq}(\widehat{D}_k^{(1)})^{-1}\widehat{R}_{k}^{\prime}Y_t\widehat{C}_k (\widehat{D}_k^{(2)})^{-1} +O_p(p^{-1}q^{-1}).$$
\end{lemma}

\begin{lemma}\label{lemaC2}
    Let $\widehat{R}_k$ and $\widehat{C}_k$ for $k\in[M]$ denote the QMLEs of $R_k^0$ and $C_k^0$.  Define the normalized estimators $\widetilde{R}_k= \widehat{R}_k(\widehat{R}_k^\prime\widehat{R}_k)^{-1/2}$ and $\widetilde{C}_k= \widehat{C}_k(\widehat{C}_k^\prime\widehat{C}_k)^{-1/2}$, with $\widetilde{R}_k^0$ and  $\widetilde{C}_k^0$ defined analogously. Under the Assumptions $A-E$ and asymptotic conditions
\(
{n^{16/\alpha}}/{(pq)} \to 0 \) and \(   {n^{2/\alpha + 2/\beta}}/{(pq)} \to 0 \), for any  $k_1\times k_1$ matrix $\widetilde{H}_{1k}$ and  $k_2\times k_2$ matrix $\widetilde{H}_{2k}$ satisfying $\|\widetilde{H}_{jk}\|=O_p(1)$ for $j=1,2$, the following hold:
\begin{align}
    (1) &~~~~ \left\| \sum_{t=1}^nI_{(s_t=k)}F_tC_k^0\widetilde{C}_k\widetilde{C}_k^\prime E_t^\prime\right\|_F^2=O_p(pq^3n)\left[1+\left\| \widetilde{C}_k-\widetilde{C}_k^0\widetilde{H}_{2k}\right\|_F^2 \right];\notag\\
(2) &~~~~ \left\| \sum_{t=1}^nI_{(s_t=k)}E_t\widetilde{C}_k\widetilde{C}_k^\prime E_t^\prime\right\|_F^2=O_p(p^2q^4n+pq^3n^2);\notag\\
(3) &~~~~ \left\| \sum_{t=1}^nI_{(s_t=k)}E_t\widetilde{C}_k\widetilde{C}_k^\prime E_t^\prime\widetilde{R}_k\right\|_F^2=O_p(p^2q^3n+pq^2n^2)\left[1+ \left\| \widetilde{C}_k-\widetilde{C}_k^0\widetilde{H}_{2k}\right\|_F^2\right] \notag\\
&~~~~~~~~~~~~~~~~~~~~~~~~~~~~~~~~~~~~+ O_p(p^2q^4n+pq^3n^2)\left\| \widetilde{R}_k-\widetilde{R}_k^0\widetilde{H}_{1k}\right\|_F^2;\notag\\
(4) &~~~~ \left\| \sum_{t=1}^n(\widehat{w}_{t|n}-I_{(s_t=k)})Y_t\widetilde{C}_k\widetilde{C}_k^\prime Y_t^\prime\right\|_F={o}_{p} \left( \frac{1}{N^{\eta}} \right)   \cdot  {O}_{p}\left( {p{q}^{2}n}\right),\notag
\end{align}
where $\widehat{w}_{t|n}$ and $\eta>0$ are given in Theorem \ref{Th2}.
\end{lemma}

\begin{lemma}\label{lemaC3}
Let $\widehat{\theta}$ denote the QMLE. Define the conditional expectation
\(
\widetilde{P}_{t|n}^{(k)} = E\left[F_t D_k^{(2)} F_t^\prime \mid s_t = k, \mathcal{Y}_n; \widehat{\theta}\right]
\)
and the matrix
\[
\bar{\Delta}_k = \frac{1}{n} \sum_{t=1}^n \widehat{w}_{t|n}^{(k)} (\widehat{D}_k^{(1)})^{1/2} \widetilde{P}_{t|n}^{(k)} (\widehat{D}_k^{(2)})^{1/2} \quad \text{for } k \in [M].
\]
Under Assumptions (A)-(E) and the asymptotic conditions
\(
\frac{n^{16/\alpha}}{pq} \to 0 \) and \(   \frac{n^{2/\alpha + 2/\beta}}{pq} \to 0 \) as $ p, q, n \to \infty,
$
the matrix $\bar{\Delta}_k$ converges almost surely to a positive definite matrix.
\end{lemma}

\textbf{Proof of Lemma \ref{lemaC1}}
\begin{proof}
\textbf{Step (1)}: We first show
$${\widehat{f} }_{ t|t }^{\left( k\right) } = \frac{1}{pq}\operatorname{vec}\left( \widehat{D}_k^{(1)})^{-1}{{\widehat{R}}_{k}^{\prime }{}^{t}{Y}_{t}{\widehat{C}}_{k}}\widehat{D}_k^{(2)})^{-1}\right)  + {O}_{p}\left( {p}^{ - 1}{q}^{ - 1}\right).$$

From the filtering process and Lemma \ref{lema1} of the main text, we have
\begin{align}
{\widehat{f} }_{t|t}^{\left( k\right) }& = \mathop{\sum }\limits_{{i = 1}}^{M}\frac{{\widehat{w}}_{t - 1,t \mid  t}^{\left( i,k\right) }}{{\widehat{w}}_{t \mid  t}^{\left( k\right) }}{\widehat{f}}_{t \mid  t}^{\left( i,k\right) } \notag\\
&
= \mathop{\sum }\limits_{{i = 1}}^{M}\frac{{\widehat{w}}_{t - 1 , t \mid t}^{\left( i,k\right) }}{{\widehat{w}}_{t \mid t}^{\left( k\right) }}\left\lbrack  {{\widehat{f}}_{t \mid t-1}^{\left( i,k\right) } + \frac{1}{{\widehat{\sigma }}^{2}}{\widehat{V}}_{t \mid t}^{\left( i,k\right) }\left( {\widehat{\Lambda}_k^{\prime}y_t  -\widehat{\Lambda}_k^{\prime}\widehat{\Lambda}_k {\widehat{f}}_{t \mid t-1}^{\left( i,k\right) }}\right) }\right\rbrack \notag\\
&= \mathop{\sum }\limits_{{i = 1}}^{M}\frac{{\widehat{w}}_{t - 1,t \mid t}^{\left( i,k\right) }}{{\widehat{w}}_{t \mid t}^{\left( k\right) }}\left\lbrack  {\left( {{I}_{r} - \frac{{1}}{{\widehat{\sigma} }^{2}}{\widehat{V}}_{t \mid t}^{\left( i,k\right) }\widehat{\Lambda}_k^{\prime}\widehat{\Lambda}_k}\right) {\widehat{f}}_{t \mid t-1}^{\left( i,k\right) } + \frac{{1}}{{\widehat{\sigma} }^{2}}{\widehat{V}}_{t \mid t}^{\left( i,k\right) } \widehat{\Lambda}_k^{\prime}y_t  }\right\rbrack  . \notag
\end{align}
From Lemma 1 of the main text, we have
\begin{align}\label{lemac1_0}
  \frac{{1}}{{\widehat{\sigma} }^{2}}{\widehat{V}}_{t \mid t}^{\left( i,k\right)}  \widehat{\Lambda}_k^{\prime}\widehat{\Lambda}_k &= \frac{1}{\widehat{\sigma}^2}\left[ I_r+ \frac{{1}}{{\widehat{\sigma} }^{2}}{\widehat{V}}_{t \mid t-1}^{\left( i,k\right)}  \widehat{\Lambda}_k^{\prime}\widehat{\Lambda}_k \right]^{-1}{\widehat{V}}_{t \mid t-1}^{\left( i,k\right)}  \widehat{\Lambda}_k^{\prime}\widehat{\Lambda}_k \notag\\
  &=\left[\widehat{\sigma}^2({\widehat{V}}_{t \mid t-1}^{\left( i,k\right)}  \widehat{\Lambda}_k^{\prime}\widehat{\Lambda}_k )^{-1}+I_r \right ]^{-1} \notag\\
  &= I_r-\widehat{\sigma}^2[{\widehat{V}}_{t \mid t-1}^{\left( i,k\right)}  \widehat{\Lambda}_k^{\prime}\widehat{\Lambda}_k+\widehat{\sigma}^2I_r ]^{-1} =I_r+O_p(\frac{1}{pq}),
\end{align}
which further implies that
$$ {{I}_{r} - \frac{{1}}{{\widehat{\sigma} }^{2}}{\widehat{V}}_{t \mid t}^{\left( i,k\right) }\widehat{\Lambda}_k^{\prime}\widehat{\Lambda}_k}=O_p(p^{-1}q^{-1})$$
and
\begin{align*}
    \frac{{1}}{{\widehat{\sigma} }^{2}}{\widehat{V}}_{t \mid t}^{\left( i,k\right) } \widehat{\Lambda}_k^{\prime}y_t  &= \frac{{1}}{{\sigma }^{2}}{\widehat{V}}_{t \mid t}^{\left( i,k\right) } \widehat{\Lambda}_k^{\prime}\widehat{\Lambda}_k (\widehat{\Lambda}_k^{\prime}\widehat{\Lambda}_k)^{-1} \widehat{\Lambda}_k^{\prime}y_t \\
  &  =(\widehat{\Lambda}_k^{\prime}\widehat{\Lambda}_k)^{-1} \widehat{\Lambda}_k^{\prime}y_t+ {O}_{p}( {p}^{ - 1}{q}^{ - 1}) \\
  &=  \frac{1}{pq}\operatorname{Vec}\left( (\widehat{D}_k^{(1)})^{-1}{{\widehat{R}}_{k}^{\prime }{}{Y}_{t}{\widehat{C}}_{k}}(\widehat{D}_k^{(2)})^{-1}\right)  + {O}_{p}\left( {p}^{ - 1}{q}^{ - 1}\right).
\end{align*}
It follows that
\begin{align}\label{lemac1_1}
{\widehat{f}}_{t \mid t}^{\left( k\right) } &= \mathop{\sum }\limits_{{i = 1}}^{M}\frac{{\widehat{w}}_{t - 1,t \mid t}^{\left( i,k\right) }}{{\widehat{w}}_{t \mid t }^{\left( k\right) }}\left\lbrack  \frac{1}{pq}\operatorname{Vec}\left( (\widehat{D}_k^{(1)})^{-1}{{\widehat{R}}_{k}^{\prime }{Y}_{t}{\widehat{C}}_{k}}(\widehat{D}_k^{(2)})^{-1}\right)  + {O}_{p}\left( {p}^{ - 1}{q}^{ - 1}\right)\right\rbrack \notag\\
&= \frac{1}{pq}\operatorname{Vec}\left( (\widehat{D}_k^{(1)})^{-1}{{\widehat{R}}_{k}^{\prime }{Y}_{t}{\widehat{C}}_{k}}(\widehat{D}_k^{(2)})^{-1}\right)  + {O}_{p}\left( {p}^{ - 1}{q}^{ - 1}\right).
\end{align}

\textbf{Step (2)}: We next show
$$\widehat{f}_{t \mid n}^{(k)} =\frac{1}{pq}\operatorname{Vec}\left( (\widehat{D}_k^{(1)})^{-1}{{\widehat{R}}_{k}^{\prime }{Y}_{t}{\widehat{C}}_{k}}(\widehat{D}_k^{(2)})^{-1}\right)  + {O}_{p}\left( {p}^{ - 1}{q}^{ - 1}\right).$$

We can first show that
\begin{align}
  {\widehat{f}}_{t \mid n}^{\left( k,\ell \right) }&  = {\widehat{f}}_{t \mid t}^{\left( k\right) } + {\widehat{G}}_{t}^{\left( k,\ell\right) }\left( {{\widehat{f}}_{t + 1 \mid n}^{\left(\ell\right) } - {\widehat{f}}_{t + 1 \mid t}^{\left( k,\ell\right) }}\right)  \notag\\
  &= \frac{1}{pq}\operatorname{Vec}\left( (\widehat{D}_k^{(1)})^{-1}{{\widehat{R}}_{k}^{\prime }{Y}_{t}{\widehat{C}}_{k}}(\widehat{D}_k^{(2)})^{-1}\right)  + {O}_{p}\left( {p}^{ - 1}{q}^{ - 1}\right)\notag
\end{align}
by the following three facts:
\begin{description}

  \item[(1)] By equation (\ref{lemac1_1}),
  $${\widehat{f}}_{t \mid t}^{\left( k\right) } = \frac{1}{pq}\operatorname{vec}\left( (\widehat{D}_k^{(1)})^{-1}{{\widehat{R}}_{k}^{\prime }{Y}_{t}{\widehat{C}}_{k}}(\widehat{D}_k^{(2)})^{-1}\right)  + {O}_{p}\left( {p}^{ - 1}{q}^{ - 1}\right).$$

 \item[(2)] By Step (4) in Section 3.2.1 of the main text,  and equations (\ref{lemac1_0}) and (\ref{lemac1_1}),
\begin{align}\label{lemac1_2}
    {\widehat{V}}_{t \mid t}^{\left( k\right) }& = \sum_{i = 1}^{M} \frac{{\widehat{w}}_{t-1,t|t}^{\left( i,k\right) }}{{\widehat{w}}_{t | t}^{\left( k\right) }} \left[ {\widehat{V}}_{t | t}^{\left( i,k\right) } + \left(\widehat{f}_{t\vert t}^{(k)} - \widehat{f}_{t\vert t}^{(i,k)}\right) \left(\widehat{f}_{t\vert t}^{(k)} - \widehat{f}_{t\vert t}^{(i,k)}\right)^{\prime} \right] \notag\\
    &= \widehat{\sigma}^2(\widehat{\Lambda}_k^{\prime}\widehat{\Lambda}_k)^{-1} + O_p\left( {p}^{-2} q^{-2} \right).
\end{align}

 \item[(3)] By Step (2) in Section 3.2.2 of the main text, and result of  equation $  {\widehat{V}}_{t \mid t}^{\left( k\right) },$
\begin{align}\label{lemac1_21}
{\widehat{G} }_{t}^{\left( {k,\ell}\right)}  & = {\widehat{V} }_{t | t}^{\left( k\right) }(\widehat{\Phi}_k^{(2)}\otimes\widehat{\Phi}_k^{(1)}) {\left\lbrack  {V}_{t + 1 | t}^{\left(  k,\ell \right) }\right\rbrack  }^{-1} = {O}_{p}\left( {{p}^{-1}{q}^{-1}}\right).
\end{align}
\end{description}
Finally, we have the results of Lemma \ref{lemaC1}.
\end{proof}

\newpage
\textbf{Proof of Lemma \ref{lemaC2}}

\begin{proof}
In the proof of Lemma \ref{lemaC2}, we assume that \({k}_{1} = {k}_{2} = 1\). The same order results can be generalized to the general case of \(\left( {{k}_{1},{k}_{2}}\right)\). In the following, define $\widetilde{C}_k=\widehat{C}_k(\widehat{D}_k^{(2)})^{-1/2} $ and $\widetilde{R}_k=\widehat{R}_k(\widehat{D}_k^{(1)})^{-1/2} $.

\textbf{(1) Consider result (1) of Lemma \ref{lemaC2}.}
\begin{align}
& \left\Vert \sum_{t} {I}_{\left( {S}_{t} = k\right)} {F}_{t} {\widetilde{C}}_{k}^{0\prime} {\widetilde{C}}_{k}{\widetilde{C} }_{k}^{\prime} {E}_{t}^{\prime} \right\Vert_{F}^{2} = O_p\left( q^{2} \right) \left\|\sum_{t} I_{(s_t = k)} F_t\widetilde{C}_k^{\prime}E_t^{\prime}  \right\|_F^2 \notag\\
&
  \lesssim O_p\left( {q}^{2}\right) \left(\left\|\sum_{t} I_{(s_t = k)} F_t {\widetilde{H} _{2k}^{\prime}} \widetilde{C}_k^{0\prime} E_t^{\prime}\right\|_F^2 + \left\|\sum_{t} I_{(s_t = k)} F_t (\widetilde{C}_k - \widetilde{C}_k^{0}\widetilde{H}_{2k})^{\prime}E_t^{\prime} \right\|_F^2\right) \notag\\
  &= O_p\left( q^{2} \right) \left( I_1 + I_2 \right) \notag,
\end{align}
where
\begin{align}
 {I}_{1} &= \left\|\sum\limits_{t} I_{(s_t = k)} F_t \widetilde{C}_k^{0\prime} E_t^{\prime}\right\|_F^2  O_p\left( 1 \right)
 = \left\|\sum\limits_{t} I_{(s_t = k)} E_t \widetilde{C}_k^0 F_t^{\prime}\right\|_F^2  O_p\left( 1 \right) \notag\\
 &
 \approx \sum\limits_{i=1}^{p} \left\| \sum\limits_{t} I_{(s_t = k)} e_{t,i\cdot}^{\prime} \widetilde{C}_k^0 F_t^{\prime} \right\|_F^2
=\sum\limits_{i=1}^{p} \left\| \sum\limits_{t} {I_{(s_t = k)} 1_{p}^{(i)}}^{\prime} E_t \widetilde{C}_k^0 F_t^{\prime} \right\|_F^2 \notag\\
&= {O}_{p}\left( {pqn}\right) ,\left( {\text{ By Assumption }E}\right)
\end{align}
and
\begin{align}
{I}_{2} &= \left\|\sum\limits_{t} I_{(s_t = k)} F_t (\widetilde{C}_k - \widetilde{C}_k^{0}\widetilde{H}_{2k})^{\prime}E_t^{\prime} \right\|_F^2 \notag\\
&= \left\|\sum\limits_{t} I_{(s_t = k)} (\widetilde{C}_k - \widetilde{C}_k^{0}\widetilde{H}_{2k})^{\prime}E_t^{\prime}F_t \right\|_F^2\notag\\
&\leq \left\| \widetilde{C}_{k} - \widetilde{C}_{k}^{0} \widetilde{H}_{2k} \right\|_F^2 \sum\limits_{i = 1}^{p} \sum\limits_{j = 1}^{q} \left\| \sum_{t} I_{(s_t = k)} 1_{p}^{(i)\prime} E_t 1_{q}^{(j)} F_t \right\|_F^2  \text{(By Assumption E)}\notag\\
&= {O}_{p}\left( {pqn}\right)  \cdot  {O}_{p}\left( {\left| \right| \widetilde{C}_{k} - {\widetilde{C}}_{k}^{ 0 }{H}_{2k}{\left| \right| }_{F}^{2}}\right) .
\end{align}
Then we have  $$\left\Vert \sum\limits_{t} I_{(S_{t} = k)} F_{t} \widetilde{C}_{k}^{0\prime} \widetilde{C}_{k} \widetilde{C}_{k}^{\prime} E_{t}^{\prime} \right\Vert_{F}^{2} = O_{p}(pq^{3}n) \left[ 1 + O_{p}\left( \left\| \widetilde{C} _{k} - \widetilde{C}_{k}^{0} \widetilde{H}_{2k} \right\|_F^2 \right) \right].$$

\textbf{(2) Consider the result (2) of Lemma \ref{lemaC2}.}

Under Assumptions D(2.1) and D(3.1),
\begin{align}
\left\Vert \sum_{t} I_{(s_{t} = k)} E_{t} \widetilde{C}_{k} \widetilde{C}_{k}^{\prime} E_{t}^{\prime} \right\Vert_{F}^{2}
&= \sum_{i = 1}^{p} \sum_{i_{1} = 1}^{p} \left\Vert \sum_{t} I_{(s_{1} = k)} e_{t,i\cdot }^{\prime} \widetilde{C}_{k} \widetilde{C}_{k}^{\prime} e_{t,i_{1\cdot }} \right\Vert_{F}^{2}
\notag\\
&= \sum_{i = 1}^{p} \sum_{i_{1} = 1}^{p} \left\Vert \sum_{t} \widetilde{C}_k^{\prime}  I_{(s_{1} = k)} e_{t,i\cdot }  {e_{t,i_{1\cdot }}^{\prime}} \widetilde{C}_k\right\Vert_{F}^{2}
\notag\\
&\leq \left\Vert \widetilde{C}_{k} \right\Vert_{F}^{4} \cdot \sum_{i = 1}^{p} \sum_{i_{1} = 1}^{p} \left\Vert \sum_{t} I_{(s_{1} = k)} e_{t,i\cdot  } e_{t,i_{1\cdot }}^{\prime} \right\Vert_{F}^{2}
\notag\\
&= {q}^{2}\mathop{\sum }\limits_{{i = 1}}^{p}\mathop{\sum }\limits_{{{i}_{1} = 1}}^{p}\sum_{j = 1}^{q} \sum_{j_{1} = 1}^{q} \left\Vert \sum_{t} I_{(s_{1} = k)} e_{t,ij} e_{t,i_{1}j_{1}} \right\Vert_{F}^{2}
\notag\\
&\leq  {q}^{2}\mathop{\sum }\limits_{{i ,i_1}}\mathop{\sum }\limits_{{j,j_1}}\left\Vert \sum_{t} I_{(s_{1} = k)} \left[e_{t,ij} e_{t,i_{1}j_{1}} - E(e_{t,ij}e_{i_{1}j_{1}})\right]\right\Vert_{F}^{2}
\notag\\
&+ {q}^{2}\mathop{\sum }\limits_{{{i},{i}_{1}}}\mathop{\sum }\limits_{{{j},{j}_{1}}}\left\Vert \sum_{t} I_{(s_{1} = k)} E(e_{t,ij}e_{t,i_{1}j_{1}})\right\Vert_{F}^{2}
\notag\\
 & \lesssim  {O}_{p}\left( {{p}^{2}{q}^{4}n}\right)  + {q}^{2}\mathop{\sum }\limits_{{i,i_1}}\mathop{\sum }\limits_{{j,j_1}}\mathop{\sum }\limits_{{t,d}}|E(e_{t,ij},e_{t,i_1j_1})|
\notag\\
&= {O}_{p}\left( {{p}^{2}{q}^{4}n}\right)  + {O}_{p}\left( {p{q}^{3}{n}^{2}}\right) . \notag
\end{align}

\textbf{(3) Consider the result (3) of Lemma \ref{lemaC2}.}
\begin{align}\label{lemac2_0}
 & \left\Vert \sum\limits_{t} I_{(s_t = k)} E_{t} \widetilde{C}_{k} \widetilde{C}_{k}^{\prime} E_{t}^{\prime} \widetilde{R}_k \right\Vert_{F}^{2}
  \notag\\
 & \lesssim\left\Vert \sum\limits_t I_{(s_t = k)} E_{t} \widetilde{C}_{k} \widetilde{C}_{k}^{\prime} E_{t}^{\prime}\widetilde{R}_k^0 \widetilde{H}_{1k} \right\Vert_{F}^{2} +\left\Vert \sum\limits_{t} I_{(s_t = k)} E_{t} \widetilde{C}_{k} \widetilde{C}_{k}^{\prime} E_{t}^{\prime}(\widetilde{R}_k - \widetilde{R}_k^{0}\widetilde{H}_{1k}  ) \right\Vert_{F}^{2}
  \notag\\
&  \lesssim  {O}_{p}\left( 1\right) \left\Vert \sum\limits_t I_{(s_t = k)} E_{t} \widetilde{C}_{k} \widetilde{C}_{k}^{\prime} E_{t}^{\prime} \widetilde{R}_k^0 \right\Vert_{F}^{2} + \left\Vert \sum\limits_t I_{(s_t = k)} E_{t} \widetilde{C}_{k} \widetilde{C}_{k}^{\prime} E_{t}^{\prime}  \right\Vert_{F}^{2} \left\Vert \widetilde{R}_{k}- \widetilde{R}_{k}^{0} \widetilde{H}_{1k} \right\Vert_{F}^{2}
  \notag\\
&= O_{p}(1) \left\Vert \sum_{t} I_{(s_t = k)} E_{t} \widetilde{C}_{k} \widetilde{C}_{k}^{\prime} E_{t}^{\prime} \widetilde{R}_{k}^{0} \right\Vert_{F}^{2} + O_{p}\left( p^{2} q^{4} n + p q^{3} n^{2} \right) \left( \left\Vert \widetilde{R}_{k} - \widetilde{R}_{k}^{0} \widetilde{H}_{1k} \right\Vert_{F}^{2} \right),
\end{align}
where the last equality holds due to the result (2) of Lemma \ref{lemaC2}.
Then we just need to compute the order of \(\left \Vert \sum\limits_{t} I_{(s_t = k)} E_{t} \widetilde{C}_{k} \widetilde{C}_{k}^{\prime} E_t^{\prime} \widetilde{R}_{k}^{0} \right\Vert_{F}^{2} \), which can be bounded by
$$\left\Vert \sum\limits_{t} I_{(s_t = k)} E_{t} \widetilde{C}_{k} {\widetilde{R}_{k}^{0\prime}} E_{t} \right\Vert_{F}^{2} \left\Vert \widetilde{C}_{k} \right\Vert_{F}^{2} = q\left\Vert \sum\limits_{t} I_{(s_t = k)} E_{t} \widetilde{C}_{k} {\widetilde{R}_{k}^{0\prime}} E_{t} \right\Vert_{F}^{2}. $$
Since
\begin{align*}
   q\left\Vert \sum\limits_{t} I_{(s_t = k)} E_{t} \widetilde{C}_{k} {\widetilde{R}_{k}^{0\prime}} E_{t} \right\Vert_{F}^{2} &\lesssim q\left\Vert \sum\limits_{t} I_{(s_t = k)} E_{t} \widetilde{C}_{k}^{0} \widetilde{H}_{2k} {\widetilde{R}_{k}^{0\prime}} E_{t} \right\Vert_{F}^{2} \\
   &+ q\left\Vert \sum\limits_{t} I_{(s_t = k)} E_{t} (\widetilde{C}_{k} -C_k^0 \widetilde{H}_{2k}) {\widetilde{R}_{k}^{0\prime}} E_{t} \right\Vert_{F}^{2} \\
   &\equiv I_3 + I_4,
\end{align*}
we only need to compute the orders of each term in the above inequality.

Consider \({I}_{3}\). Under Assumptions D(3.2) and D(2.2), we have
\begin{align}
  {I}_{3}& =  q\left\Vert \sum\limits_{t} I_{(s_t = k)} E_{t} \widetilde{C}_{k}^{0} \widetilde{H}_{2k} {\widetilde{R}_{k}^{0\prime}} E_{t} \right\Vert_{F}^{2} =  O_p(q)\left\Vert \sum\limits_{t} I_{(s_t = k)} E_{t} \widetilde{C}_{k}^{0} {\widetilde{R}_{k}^{0\prime}} E_{t} \right\Vert_{F}^{2} \notag\\
 & = q\mathop{\sum }\limits_{{i = 1}}^{p}\mathop{\sum }\limits_{{j = 1}}^{q} \left\Vert \sum\limits_{t} I_{(s_t = k)} e_{t,i\cdot } \widetilde{C}_{k}^{0} {\widetilde{R}_{k}^{0\prime}} e_{t,\cdot j} \right\Vert_{F}^{2}
  \notag\\
&  = q\mathop{\sum }\limits_{{i = 1}}^{p}\mathop{\sum }\limits_{{j = 1}}^{q} \left\Vert \sum\limits_{t} \sum\limits_{i_1} \sum\limits_{j_1}   I_{(s_t = k)} \widetilde{r}_{k,i_1 }^{0} \widetilde{c}_{k,j_1}^{0} e_{t,ij_1} e_{t,i_1j} \right\Vert_{F}^{2}
  \notag\\
&  \lesssim  q\mathop{\sum }\limits_{{i = 1}}^{p}\mathop{\sum }\limits_{{j = 1}}^{q} \left\Vert \sum\limits_{t} \sum\limits_{i_1} \sum\limits_{j_1}   I_{(s_t = k)} \widetilde{r}_{k,i_1 }^{0} \widetilde{c}_{k,j_1}^{0} [e_{t,ij_1} e_{t,i_1j} - E(e_{t,ij_1}  e_{t,i_1j})]\right\Vert_{F}^{2}
  \notag\\
&  + q\mathop{\sum }\limits_{{i = 1}}^{p}\mathop{\sum }\limits_{{j = 1}}^{q} \left\Vert \sum\limits_{t} \sum\limits_{i_1} \sum\limits_{j_1}   I_{(s_t = k)} \widetilde{r}_{k,i_1 }^{0} \widetilde{c}_{k,j_1}^{0} E(e_{t,ij_1} e_{t,i_1j})\right\Vert_{F}^{2}
  \notag\\
&= {O}_{p}\left( {n{p}^{2}{q}^{3} + {n}^{2}p{q}^{2}}\right) .
\end{align}

Consider \({I}_{4}\) , where
\begin{align}\label{lemac2_2_1}
 {I}_{4} &= q\left\Vert \sum\limits_{t} I_{(s_t = k)} E_{t} (\widetilde{C}_{k} - \widetilde{C}_k^{0} \widetilde{H}_{2k} ){\widetilde{R}_{k}^{0\prime}} E_{t} \right\Vert_{F}^{2}
 \notag\\
 &= q \sum\limits_{i = 1}^{p} \sum\limits_{j = 1}^{q}  \left\Vert \sum\limits_{t} I_{(s_t = k)} e_{t,i\cdot } ^{\prime}(\widetilde{C}_{k} - \widetilde{C}_k^{0} \widetilde{H}_{2k} ){\widetilde{R}_{k}^{0\prime}} e_{t,\cdot j} \right\Vert_{F}^{2}  \notag\\
 &\leqslant  q   \left\Vert  \widetilde{C}_{k} - \widetilde{C}_k^{0} \widetilde{H}_{2k} \right\Vert_{F}^{2} \sum\limits_{i = 1}^{p} \sum\limits_{j = 1}^{q}  \left\Vert \sum\limits_{t} I_{(s_t = k)}e_{t,i\cdot } e_{t,\cdot j}^{\prime}R_k^{0}\right\Vert_{F}^{2}
  \notag\\
&  =  q   \left\Vert (\widetilde{C}_{k} - \widetilde{C}_k^{0} \widetilde{H}_{2k}  \right\Vert_{F}^{2} \sum\limits_{i = 1}^{p} \sum\limits_{j ,j_1}^{q}  \left\Vert \sum\limits_{t} \sum_{i_1}I_{(s_t = k)}e_{t,ij_1} e_{t,i_1j}\widetilde{r}_{k,i_1}^{0}\right\Vert_{F}^{2}
   \notag\\
&   \lesssim  q   \left\Vert  \widetilde{C}_{k} - \widetilde{C}_k^{0} \widetilde{H}_{2k}  \right\Vert_{F}^{2} \sum\limits_{i = 1}^{p} \sum\limits_{j,j_1}^{q}  \left\Vert \sum\limits_{t} \sum\limits_{i_1}  I_{(s_t = k)}\widetilde{r}_{k,i_1}^{0}[e_{t,ij_1 } e_{t,i_1j}-E(e_{t,ij_1}e_{t,i_1j})] \right\Vert_{F}^{2}
    \notag\\
 &   +   q   \left\Vert  \widetilde{C}_{k} - \widetilde{C}_k^{0} \widetilde{H}_{2k}  \right\Vert_{F}^{2} \sum\limits_{i = 1}^{p} \sum\limits_{j,j_1}^{q}  \left\Vert \sum\limits_{t} \sum\limits_{i_1}  I_{(s_t = k)}\widetilde{r}_{k,i_1}^{0}E(e_{t,ij_1}e_{t,i_1j})] \right\Vert_{F}^{2}
     \notag\\
&     = {O}_{p}\left( {n{p}^{2}{q}^{3}}\right) \left\Vert  \widetilde{C}_{k} - \widetilde{C}_k^{0} \widetilde{H}_{2k}  \right\Vert_{F}^{2} +   q\left\Vert  \widetilde{C}_{k} - \widetilde{C}_k^{0} \widetilde{H}_{2k}  \right\Vert_{F}^{2}
 \notag\\
&    \sum\limits_{i } \sum\limits_{j,j_1}\sum\limits_{i_1,i_2} \sum\limits_{t,d}  I_{(s_t = k)}I_{(s_{d} = k)}\widetilde{r}_{k,i_1}^{0}\widetilde{r}_{k,i_2}^{0}E(e_{t,ij_1}e_{t,i_1j}) E(e_{d,ij_1}e_{d,i_2j})
 \notag\\
& =  {O}_{p}\left( {n {p}^{2}{q}^{3}}\right) {\begin{Vmatrix}\widetilde{{C}}_{k} - \widetilde{C}_{k}^{ 0  }\widetilde{H}_{2k}\end{Vmatrix}}_{F}^{2} +q   \left\Vert  \widetilde{C}_{k} - \widetilde{C}_k^{0} \widetilde{H}_{2k}  \right\Vert_{F}^{2}
  \notag\\
&    \sum\limits_{i } \sum\limits_{j,j_1} \sum\limits_{i_1} \sum\limits_{t}  I_{(s_t = k)}\widetilde{r}_{k,i_1}^{0}E(e_{t,ij_1},e_{t,i_1j}) \left[\sum\limits_{i_2}\sum\limits_{d} \widetilde{r}_{k,i_2} E(e_{d,ij_1},e_{d,i_2j})\right].
\end{align}
By Assumption D (2.2), we have
\(\sum\limits_{i_2} \sum\limits_{d} \widetilde{r}_{k,i_2}^0 |E(e_{d,ij_1} e_{d,i_2j})| \leq nc_0,\) and \(\sum\limits_{i_1=1}^{p} \sum\limits_{j_1=1}^{p} \widetilde{r}_{k,i_1}^0 |E(e_{t,ij_1} e_{t,i_1j})| \leq c_0,\) which implies that the second term on the right-hand side of (\ref{lemac2_2_1}) can be bounded by \(O_p(n^2 pq^2) \left\| \widetilde{C}_{k} - \widetilde{C}_{k}^{0} \widetilde{H}_{2k} \right\|_F^2\).
Then we have
$$I_{4}=O_p(np^2q^3 + n^2 pq^2)\left\| \widetilde{C}_{k} - \widetilde{C}_{k}^{0} \widetilde{H}_{2k} \right\|_F^2. $$

Combined results of \({I}_{3} \) and \( {I}_{4}\), we have
\begin{align}
   \left\lVert \sum\limits_{t} I(s_t = k) E_t \widetilde{C}_k \widetilde{C}_k^{\prime} E_t^{\prime} \widetilde{R}_k \right\rVert_{F}^{2}
   &= O_p\left( n p^{2} q^{3} + n^{2} p q^{2} \right) \left[ 1 + O_p\left\| \widetilde{C}_k - \widetilde{C}_{k}^{0} \widetilde{H}_{2k} \right\|_F^2 \right] \notag\\
&+  {{O}_{p}\left( {{p}^{2}{q}^{4}n + p{q}^{3}{n}^{2}}\right) \left\| \widetilde{R}_k - \widetilde{R}_{k}^{0} \widetilde{H}_{1    k} \right\|_F^2.} \notag
\end{align}

\textbf{(4) Consider result (4) of Lemma \ref{lemaC2}.}
\begin{align}
    \left\Vert \sum\limits_{t}(\widehat{w}_{t|n}^{(k)} - I_{(s_t = k)}) Y_{t} \widetilde{C}_k \widetilde{C}_k^{\prime} Y_{t}^{\prime} \right\Vert_{F}
    &\leq \sum\limits_{t} \left| \widehat{w}_{t | n}^{(k)} - I{(s_{t} = k)} \right| \left \Vert Y_{t} \widetilde{C}_k \widehat{C}_k^{\prime} Y_{t}^{\prime} \right\Vert_{F}\notag\\
&\lesssim k_2\sum\limits_{t} \left| \widehat{w}_{t | n}^{(k)} - I{(s_{t} = k)} \right| \left \Vert Y_{t} \widetilde{C}_k \widetilde{C}_k^{\prime} Y_{t}^{\prime} \right\Vert
\notag\\
&\leqslant k_2 \sup_{t} \left| \widehat{w}_{t | n}^{(k)} - I(s_{t} = k) \right| \sum_{t} \left\| Y_{t} \widetilde{C}_k \widetilde{C}_k^{\prime} Y_{t}^{\prime} \right\|
\notag\\
&\leqslant k_2 \sup_{t} \left| \widehat{w}_{t | n}^{(k)} - I(s_{t} = k) \right| \sum_{t} \left\| \widetilde{C}_k^{\prime} Y_{t}^{\prime} Y_{t} \widetilde{C}_k \right\|_F
\notag\\
&\lesssim {O}_{p} \left( \frac{1}{N^{\eta}} \right) \cdot \left\| \widetilde{C}_{k} \right\|_{F}^{2} \sum_{t}  \left\| {Y}_t \right\|_{F}^{2}
\notag\\
&= {o}_{p} \left( \frac{1}{N^{\eta}} \right)   \cdot  {O}_{p}\left( {p{q}^{2}n}\right) \notag.
\end{align}
Note that the second inequality holds due to that for a matrix, \(\parallel A{\parallel }_{F} \leq  \operatorname{rank}\left( A\right) \|A\|\) and the last inequality holds due to Assumptions B(1), C(1) and D(1).
\end{proof}

\textbf{Proof of Lemma \ref{lemaC3}}
\begin{proof}
Consider the estimation equation of \(\widehat{R}_k \) given in the main text, that is
$$\widehat{R}_k(\sum_{t=1}^n\widehat{w}_{t|n}^{(k)}\widehat{P}_{t|n}^{c(k)})=\sum_{t=1}^n \widehat{w}_{t|n}^{(k)}Y_t\widehat{C}_k\widehat{F}_{t|n}^{(k)\prime},$$
where  $\widehat{P}_{t|n}^{c(k)}=E[F_tC_k^\prime C_kF_t^\prime|s_t=k,\mathcal{Y}_n;\widehat{\theta}]=qE[F_t D_k^{(2)}F_t^{\prime}|s_t=k,\mathcal{Y}_n;\widehat{\theta} ] \triangleq q \widetilde{P}_{t|n}^{(k)}$. Define $\widehat{\Delta}_k=\frac{1}{n}\sum_{t=1}^n\widehat{w}_{t|n}^{(k)}\widetilde{P}_{t|n}^{(k)}$. From the asymptotic representation of \( \widehat{f}_{t|n}^{(k)} \) given in Lemma \ref{lemaC1}, we have the following.
\begin{align}\label{lemac3_1}
  \widehat{R}_k \widehat{{\Delta}} _k &
  =    \frac{1}{p{q}^{2}n }\mathop{\sum }\limits_{{t = 1}}^{n}\widehat{w}_{t|n}^{(k)}{Y}_{t}\widehat{C}_k (\widehat{D}_k^{(2)})^{-1}\widehat{{C}}_k^{\prime}{Y}_{t}^{\prime}\widehat{ R}_{k}(\widehat{D}_k^{(1)})^{-1}
  \notag\\
  &+
  \underbrace{\frac{1}{nq}\mathop{\sum }\limits_{{t = 1}}^{n}\widehat{w}_{t|n}^{(k)}{Y}_{t}\widehat{{C}}_k\Delta_{ f_{t|n}^{(k)}}\left( {p ^{-1} q ^{- 1}}\right)}_{:= B},\notag\\
  &=  \frac{1}{p{q}^{2}n }\mathop{\sum }\limits_{{t = 1}}^{n}I_{s_t=k}{Y}_{t}\widehat{C}_k (\widehat{D}_k^{(2)})^{-1}\widehat{{C}}_k^{\prime}{Y}_{t}^{\prime}\widehat{ R}_{k}(\widehat{D}_k^{(1)})^{-1}\notag\\
  &
  + \underbrace{\frac{1}{p{q}^{2}n }\mathop{\sum }\limits_{{t = 1}}^{n}(\widehat{w}_{t|n}^{(k)}-I_{s_t=k}){Y}_{t}\widehat{C}_k (\widehat{D}_k^{(2)})^{-1}\widehat{{C}}_k^{\prime}{Y}_{t}^{\prime}\widehat{ R}_{k}(\widehat{D}_k^{(1)})^{-1}}_{:=D}+B,
\end{align}
where $ \Delta_{ f_{t|n}^{(k)}}\left( {p ^{-1} q ^{- 1}}\right)$ denotes a $k_1 \times k_1$ matrix with each element bounded by $O_p(p^{-1}q^{-1})$. In the following, define $\widetilde{R}_k =\widehat{R}_k(\widehat{D}_k^{(1)})^{-1/2}$ and $\widetilde{C}_k =\widehat{C}_k(\widehat{D}_k^{(2)})^{-1/2}$. Then $\widetilde{R}_k^{\prime}\widetilde{R}_k=pI_{k_1}$ and $\widetilde{C}_k^{\prime}\widetilde{C}_k=qI_{k_2}$. Define$\bar{\Delta}_k=(\widehat{D}_k^{(1)})^{1/2}\widehat{\Delta}_k(\widehat{D}_k^{(1)})^{1/2}$. Then equation (\ref{lemac3_1}) can be rewritten as
\begin{align}\label{lemac3_2}
    \widetilde{R}_k\bar{\Delta}_k= \frac{1}{npq^2}\sum_{t=1}^nI_{(s_t=k)} Y_t \widetilde{C}_k\widetilde{C}_k^{\prime}Y_t^{\prime}\widetilde{R}_k+(D+B)(\widehat{D}_k^{(1)})^{1/2}.
\end{align}
Substituting $Y_t = R_k^0 F_t C_k^{0\prime}+E_t$ when $s_t=k$ into the first term on the right hand side of above equation, we have
\begin{align}\label{lemac3_3_0}
     \widetilde{R}_k\bar{\Delta}_k& =
     \frac{1}{{p}{q}^{2} n}\mathop{\sum }\limits_{{t = 1}}^{n}{I}_{\left( s_t = k\right) }{R}_{k}^{0}{F}_{t}{{C}_{k}^{0}}^{\prime}\widetilde {C}_{k}\widetilde{ C}_{k}^{\prime}{C}_{k}^{0}{F_t}^{\prime}{{R}_{k}^{0}}^{\prime}\widetilde{R}_k \notag\\
     &+\underbrace{\frac{1}{{p}{q}^{2} n}\mathop{\sum }\limits_{{t = 1}}^{n}{I}_{\left( s_t = k\right) }{R}_{k}^{0}{F}_{t}{{C}_{k}^{0}}^{\prime}\widetilde {C}_{k}\widetilde {C}_{k}^{\prime}{E_t}^{\prime}\widetilde{R}_k }_{:=A_1} \notag\\
     &+\underbrace{\frac{1}{{p}{q}^{2} n}\mathop{\sum }\limits_{{t = 1}}^{n}{I}_{\left( s_t = k\right) }{E}_{t}\widetilde {C}_{k}\widetilde {C}_{k}^{\prime}{C}_{k}^{0}{F_t}^{\prime}{{R}_{k}^{0}}^{\prime}\widetilde{R}_k}_{:=A_2} \notag\\
     &+\underbrace{ \frac{1}{{p}{q}^{2} n}\mathop{\sum }\limits_{{t = 1}}^{n}{I}_{\left( s_t = k\right) }{E}_{t}\widetilde {C}_{k}\widetilde {C}_{k}^{\prime}{{E}_{t}}^{\prime}\widetilde{R}_k}_{:=A_3}\notag\\
     &+(D+B)(\widehat{D}_k^{(1)})^{1/2}.
\end{align}
In addition, define \(\bar{R}_{k} = \frac{\widetilde{R}_{k}}{\sqrt{p} }\) as the normalized version of $\widehat{R}_k$.
 Then equation (\ref{lemac3_3_0}) becomes
 \begin{align}
     \bar{R}_{k}\bar{\Delta} _ k &= \frac{1}{p{q}^{2}n}\mathop{\sum }\limits_{{t = 1}}^{n} {I}_{(s_t= k)} {R}_{k}^{0}{F}_{t}{{C}_{k}^{0}}^{\prime}\widetilde {C}_{k}\widetilde {C}_{k}^{\prime}{C}_k^0{F}_{t}^{\prime}{R}_{k}^{0}\bar{R}_k
     \notag \\
     &+ \frac{1}{\sqrt{p} }{A}_{1} + \frac{1}{\sqrt{p} }{A}_{2}+ \frac{1}{\sqrt{p} }{A}_{3}+ \frac{1}{\sqrt{p} }(D+B)(\widehat{D}_k^{(1)})^{1/2}.
 \end{align}
By Lemma \ref{lemaC2} (1), we have
\begin{align}
    {\begin{Vmatrix}\frac{1}{\sqrt{p}}{A}_{1}\end{Vmatrix}}_{F}^{2}& = \frac{1}{p}{\begin{Vmatrix}{A}_{1}\end{Vmatrix}}_{F}^{2} = \frac{1}{p^3q^4n^2} \left\Vert \sum\limits_{t} I_{(S_{t} = k)} R_{k}^{0}{F_{t}C_k^0}^{\prime} \widetilde{C}_k \widetilde{C}_k^{\prime} E_{t}^{\prime} \widetilde{R}_{k} \right\Vert_{F}^{2} \notag\\
    &\lesssim  \frac{\left| \right| {R}_{k}^{0}{\left| \right| }_{F}^{2}}{p} \cdot  \frac{\left| \right| {\widetilde{R}}_{k}{\left| \right| }_{F}^{2}}{{p}} \cdot  \frac{1}{pq^4n^2} \left\Vert \sum\limits_{t} I_{(S_{t} = k)} {F_{t}C_k^0}^{\prime} \widetilde{C}_k \widetilde{C}_k^{\prime} E_{t}^{\prime}  \right\Vert_{F}^{2}
    \notag\\
    &= O_{p}\left( \frac{1}{qn}\right) \left[1 + \left\| \widetilde{C}_{k} - \widetilde{C}_{k}^{0} \widetilde{H}_{2k} \right\|_F^2 \right] = o_p(1). \notag
\end{align}
Similarity, \(\left| \right| \frac{1}{\sqrt {p}}{A}_{2}{\left| \right| }_{F}^{2} \asymp  \left| \right| \frac{1}{\sqrt {p}}{A}_{1}{\left| \right| }_{F}^{2}\). By Lemma \ref{lemaC2}(3),
\begin{align}
    \left| \right| \frac{1}{\sqrt {p}}{A}_{3}{\left| \right| }_{F}^{2} &= \frac{1}{{p}^{3}{q}^{4} n^{2}}\left\Vert \sum\limits_{t} I_{(S_{t} = k)} E_{t} \widetilde{C}_k \widetilde{C}_k^{\prime} {E_{t}}^{\prime}\widetilde{ R}_k  \right\Vert_{F}^{2}
    \notag\\
    &=O_p(\frac{1}{pqn}+\frac{1}{p^2q^2})\left[1 + \left\| \widetilde{C}_k - C_{k}^{0} \widetilde{H}_{2k} \right\|_F^2 \right] \notag\\
    &+O_p(\frac{1}{pn}+\frac{1}{p^2q})\left\| \widetilde{R}_k - \widetilde{R}_{k}^{0} \widetilde{H}_{1k} \right\|_F^2=o_p(1).
\end{align}
By Lemma \ref{lemaC2}(4) and Theorem \ref{Th2},
\begin{align}
    \left| \right| \frac{1}{\sqrt{p}}D{\left| \right| }_{F}^{2} &= \frac{1}{{p}^{3}{q}^{4}{n}^{2}}\left\Vert \sum_{t} (\widehat{w}_{t|n}^{(k)} -I_{(s_{t} = k)}) Y_t \widetilde{C}_k \widetilde{C}_k^{\prime} {Y_{t}}^{\prime} \widetilde{R}_k  \right\Vert_{F}^{2} \notag\\
    &=o_p(\frac{1}{p^{2\eta+1}q^{2\eta+1}})=o_p(1).
\end{align}
By the same calculation, we can show that \(\frac{1}{p}\parallel B{\parallel }_{F}^{2}\) can be bounded by \({\begin{Vmatrix}\frac{1}{\sqrt{p} }{A}_{1}\end{Vmatrix}}_{F}^{2},\cdots {\begin{Vmatrix}\frac{1}{\sqrt{p} }D\end{Vmatrix}}_{F}^{2}\).
Then we have
 \begin{align}\label{lemac3_2}
     \bar{R}_{k}\bar{\Delta} _ k &= \frac{1}{p{q}^{2}n}\mathop{\sum }\limits_{{t = 1}}^{n} {I}_{(s_t= k)} {R}_{k}^{0}{F}_{t}{{C}_{k}^{0}}^{\prime}\widetilde {C}_{k}\widetilde {C}_{k}^{\prime}{C}_k^0{F}_{t}^{\prime}{R}_{k}^{0\prime}\bar{R}_k
     +o_p(1).
 \end{align}
Next, we will show that \(\frac{{\widetilde{C}_{k}^{0\prime}}\widetilde{C}_k}{q} \cdot  \frac{{\widetilde{C}_k^{\prime} }\widetilde{C}_k^{0}}{q} = \frac{1}{a}{I}_{k_2}+ o_p\left( 1\right)\) , where \(a > 0\) and is bound away from 0. From Theorem \ref{Th1}, we have the following.
\begin{align}\label{lemac3_2_1}
   \frac{{\widetilde{R}_{k}^{0\prime}}\widetilde{R}_k} {p} \cdot  \frac{{\widetilde{R}_k }^{\prime}\widetilde{R}_k^{0}}{p} \otimes  \frac{{\widetilde{C}_{k}^{0\prime}}\widetilde{C_k}} {q}   \frac{{\widetilde{C_k} }^{\prime}\widetilde{C}_k^{0}}{q}= {I}_{r} + {o}_{p}(1),
\end{align}
which means \({a_s}^{\prime }{a}_{s} \cdot  \frac{{\widetilde{C}_{k}^{0\prime}}\widetilde{C}_k} {q} \cdot  \frac{{\widetilde{C}_k }^{\prime}\widetilde{C}_k^{0}}{q}= {I}_{{k}_{2}} + {o}_p{{\left( 1\right) }}\) for all \(s \in  \left\lbrack  {k}_{1}\right\rbrack\) , where \({a}_{s} = ({\frac{{\widetilde{R}_{k,\cdot s}^{0\prime}}\widetilde{R}_k} {p}})^{\prime}\) . Since \({a_1}^{\prime }{a}_{1} \cdot  \frac{{\widetilde{C}_{k}^{0\prime}}\widetilde{C}_k} {q} \cdot  \frac{{\widetilde{C}_k }^{\prime}\widetilde{C}_k^{0}}{q}= {I}_{{k}_{2}} + {o}_p{{\left( 1\right) }}\) and \({a_2}^{\prime }{a}_{2} \cdot  \frac{{\widetilde{C}_{k}^{0\prime}}\widetilde{C}_k} {q} \cdot  \frac{{\widetilde{C}_k }^{\prime}\widetilde{C}_k^{0}}{q}= {I}_{{k}_{2}} + {o}_{p}(1)\) , then \({{a}_{1}}^{\prime}{a}_{1} - {{a}_{2}}^{\prime}{a}_{2}= {o_p}\left( 1\right),\)  which means \({a}_{s}^{\prime }{a}_{s} = a + o_{p}{\left( 1\right) }\) , for all \(s \in  \left\lbrack  {k}_{1}\right\rbrack\) , and $a$ is bounded away from 0. So equation (\ref{lemac3_2}) can be further simplified as
\begin{align}\label{lemac3_3}
    \bar{R}_k \bar{\Delta }_k = \frac{1}{{a}}\frac{1}{{pn}}\mathop{\sum }\limits_{t=1}^{n}{I}_{(s_{t = k}) }{R}_{k}^{0}{F}_{t}(D_k^{(2)}){F}_{t}^{\prime}{{R}_{k}^{0}}^{\prime}\bar{R}_k +o_p(1).
\end{align}

Let \(V_{k}\) be a \({k}_{1} \times  {k}_{2}\) diagonal matrix consisting of eigenvalues of \(\frac{1}{n\pi_{k}^{0}} \sum\limits_{t} I_{(s_t=k)}(D_k^{(1)})^{1/2}F_t D_k^{(2)}{F_t}^{\prime}(D_k^{(1)})^{1/2}\) in descending order, and \(\tau _k\) be the corresponding vectors. Let \(\bar {R}_k^0 = \frac{1}{\sqrt{p}} {R}_{k}^{0 }(D_k^{(1)})^{-1/2}{\tau}_{k}\) , then \( \bar{R}_k^{0\prime}\bar{R}_k^0 = {I}_{k_1}\). It follows that
\begin{align}\label{lemac3_4}
    \bar{R}_k \bar{\Delta} _k & = \frac{1}{{\alpha}} \cdot \frac{1}{{pn}}\mathop{\sum }\limits_{t=1}^{n}{I}_{(s_{t = k}) }{R}_{k}^{0}{F}_{t} D_k^{(2)}{F}_{t}^{\prime}{{R}_{k}^{0}}^{\prime}\bar{R}_k +o_p(1) \notag\\
    &= \frac{1}{\alpha } \cdot  {\pi }_{k}^{0}\bar{R}_k^0{V}_{k}\bar{R}_k^{0\prime} \bar{R}_k+ {o_p}\left( 1\right).
\end{align}
The left-hand side of equation (\ref{lemac3_4}) equals \({P}_{{\bar{R}_k^0}}\bar{R}_k \bar{\Delta}_k + {M}_{{\bar{R}_k^0}}\bar{R}_k \bar{\Delta}_k  = {\bar{R}_k^0}{\bar{R}_k^{0\prime}}\bar{R}_k \bar{\Delta _k} + {M}_{{\bar{R}_k^0}}\bar{R}_k \bar{\Delta} _k \) , thus we have
\[{\bar{R}}_{k}^{0}\left( \bar{R}_k^{0\prime} {\bar{R}_k\bar{\Delta} _k} - \frac{1}{a}{\pi }_{k}^{0}{V}_{k}{\bar{R}_k^{0\prime}}\bar{R}_k\right)  + {M}_{\bar{R}_k^0}\bar{R}_k\bar{\Delta}_k = {o_p\left( 1\right) }.\]
Since the two terms on the left-hand side are orthogonal to each other, thus both \({\begin{Vmatrix}{\bar{R}_k}^{0}\left( \bar{R}_k^{0\prime} {\bar{R}_k\bar{\Delta} _k} - \frac{1}{a}{\pi }_{k}^{0}{V}_{k}{\bar{R}_k^{0\prime}}\bar{R}_k\right) \end{Vmatrix}}_{F}\) and \({\begin{Vmatrix}{M}_{\bar{R}_k^0}\bar{{R}_{k}}\bar{\Delta} _k\end{Vmatrix}}_{F}\) are \(o_p\left( 1\right)\) .
Let \(A =  \bar{R}_k^{0\prime} {\bar{R}_k\bar{\Delta} _k} - \frac{1}{a}{\pi }_{k}^{0}{V}_{k}{\bar{R}_k^{0\prime}}\bar{R}_k\). Since
$$
 o_p \left( 1\right)  = {\begin{Vmatrix}{\bar{R}}_{k}^{0}A\end{Vmatrix}}_{F} = \sqrt{{tr}\left( {{\bar{R}}_{k}^{0}A{A}^{\prime }{\bar{R}}_{k}^{0}}^{\prime}\right) } = \sqrt{{tr}\left( A{A}^{\prime }{{\bar{R_k^0}}^{\prime}{\bar{R}}_{k}^{0}}\right) } = \sqrt{{tr}\left( {A{A}^{\prime }}\right) } = \parallel A{\parallel }_{F},
$$
 then \(\parallel A{\parallel }_{F} = {\begin{Vmatrix}\left( \bar{R}_k^{0\prime} {\bar{R}_k\bar{\Delta} _k} - \frac{1}{a}{\pi }_{k}^{0}{V}_{k}{\bar{R}_k^{0\prime}}\bar{R}_k\right)\end{Vmatrix}}_{F} = {o}_{p\left( 1\right) }\), which means
 \begin{align}\label{lemac3_5}
 &   \bar{R}_k^{0\prime} {\bar{R}_k\bar{\Delta}_k} = \frac{1}{a}{\pi }_{k}^{0}{V}_{k}{\bar{R}_k^{0\prime}}\bar{R}_k + o_p(1) \notag\\
\Rightarrow &  \bar{R}_k^{\prime}\bar{R}_k^0\bar{R}_k^{0\prime} {\bar{R}_k\bar{\Delta }_k} = \frac{1}{a}{\pi }_{k}^{0}\bar{R}_k^{\prime}\bar{R}_k^0V_k\bar{R}_k^{0\prime} \bar{R}_k + {o_p}\left( 1\right).
 \end{align}
Combining equation (\ref{lemac3_2_1}) and \(\frac{{\widetilde{C}_{k}^{0\prime}}\widetilde{C_k}}{q} \cdot  \frac{{\widetilde{C}_k }^{\prime}\widetilde{C}_k^{0}}{q} = \frac{1}{a}{I}_{k_2}+ o_p\left( 1\right)\)
we have \(\frac{{{\widetilde{R}_k^{0\prime}}}{{\bar{R}_k}}}{p}\frac{{{\bar{R}_k}^{\prime}}{{\widetilde{R}_k^0}}}{p}  = a{I}_{k_1} + {o}_{p}\left( 1\right)\),
which implies that \(\bar{R}_k^{0\prime}\bar{R}_k\bar{R_k}^\prime\bar{R}_k^0=aI_{k_1}+o_p(1)\). Furthermore, since
  \begin{align}
    \begin{Vmatrix}aI_{k_1}- \bar{R}_k^\prime\bar{R}_k^{0}\bar{R}_k^{0\prime}\bar{R}_k
    \end{Vmatrix}_F &\leq \sqrt{k_1}\|aI_{k_1}- \bar{R}_k^\prime\bar{R}_k^{0}\bar{R}_k^{0\prime}\bar{R}_k\|\notag\\
    &\leq \sqrt{k_1}\text{Tr}[aI_{k_1}- \bar{R}_k^\prime\bar{R}_k^{0}\bar{R}_k^{0\prime}\bar{R}_k ] \notag\\
     &=\sqrt{k_1}\text{Tr}[aI_{k_1}- \bar{R}_k^{0\prime}\bar{R}_k\bar{R_k}^\prime\bar{R}_k^0 ] \notag\\
      &=o_p(1),
  \end{align}
then $\bar{R}_k^\prime\bar{R}_k^{0}\bar{R}_k^{0\prime}\bar{R}_k= aI_{k_1}+o_p(1)$. Backing to equation (\ref{lemac3_5}), we have
\begin{align}
&\bar{\Delta}_k = \frac{1}{a^2} \pi_k^0  \bar{R}_k^\prime\bar{R}_k^{0}V_k\bar{R}_k^{0\prime}\bar{R}_k+o_p(1).
\end{align}
  where $\bar{R}_k^\prime\bar{R}_k^{0} $ is an orthogonal matrix of asymptotic orthogonal columns. Then we finished the proof of Lemma \ref{lemaC3}.
\end{proof}

\begin{proof}

\textbf{Proof of Theorem 3}

Proof: From equation (\ref{lemac3_3_0}) in the proof of Lemma \ref{lemaC3}, we have
\begin{align}\label{Th3_1}
  {\widetilde{R}}_{k} &= \widetilde{R}_{k}^{0} \underbrace{ \frac{1}{pq^2n}\mathop{\sum }\limits_{{t = 1}}^{n}{I}_{\left( {s_t = k}\right) }(D_k^{(1)})^{\frac{1}{2}}F_t{{C}_{k}^{0}}^\prime\widetilde{C}_k{\widetilde{C}_k}^{\prime}C_k^0F_t^{\prime}{R_t^0}^\prime \widetilde{R}_k{\bar{\Delta}_k^{-1}}}_{:= \widetilde{H}_{1k}}\notag\\
  &+ \left( {{A}_{1} + {A}_{2} + {A}_{3} + D(\widehat{D}_k^{(1)})^{\frac{1}{2}}+B(\widehat{D}_k^{(1)})^{\frac{1}{2}}}\right)   \bar{\Delta }_{k}^{-1},
\end{align}
which means
\begin{align}\label{Th3_2}
   &  \widetilde{R}_{k} - \widetilde{R}_{k}^{ 0 } \widetilde{H}_{1k} = \left( {{A}_{1} + {A}_{2} + {A}_{3} + D(\widehat{D}_k^{(1)})^{\frac{1}{2}} + B(\widehat{D}_k^{(1)})^{\frac{1}{2}} }\right) \bar{\Delta }_{k}^{-1} \notag\\
   \Rightarrow &
    \frac{1}{p}{\begin{Vmatrix} \widetilde{R}_{k} - \widetilde{R}_{k}^{ 0 } \widetilde{H}_{1_k}\end{Vmatrix}}_{F}^{2} \leq  \left\lbrack  {\frac{1}{p}{\begin{Vmatrix}{A}_{1}\end{Vmatrix}}_{F}^{2} + \frac{1}{p}{\begin{Vmatrix}{A}_{2}\end{Vmatrix}}_{F}^{2} + \frac{1}{p}{\begin{Vmatrix}{A}_{3}\end{Vmatrix}}_{F}^{2} + \frac{1}{p}{\begin{Vmatrix}B\end{Vmatrix}}_{F}^{2} + \frac{1}{p}{\begin{Vmatrix}D\end{Vmatrix}}_{F}^{2}}\right\rbrack   \cdot  O_p\left( 1\right).
\end{align}
By Lemma \ref{lemaC2}(1)-(3) and Theorem \ref{Th2}, we have
\begin{align}
    \frac{1}{p}\parallel {A}_{1}{\parallel }_{F}^{2}  & =O_p\left( \frac{1}{q n} \right) \left[1+  \left\| \widetilde{C}_k - \widetilde{C}_k^0 \widetilde{H}_{2k} \right\|_F^2\right], \notag\\
        \frac{1}{p}\parallel {A}_{2}{\parallel }_{F}^{2}  & =O_p\left( \frac{1}{q n} \right) \left[1+  \left\| \widetilde{C}_k - \widetilde{C}_k^0 \widetilde{H}_{2k} \right\|_F^2\right] ,\notag\\
    \frac{1}{p}\parallel {A}_{3}{\parallel }_{F}^{2}  & =O_p\left( \frac{1}{pq n}+\frac{1}{p^2q^2} \right) \left[1+  \left\| \widetilde{C}_k - \widetilde{C}_k^0 \widetilde{H}_{2k} \right\|_F^2\right]  + o_p\left(\frac{\left\| \widetilde{R}_k - \widetilde{R}_k^0 \widetilde{H}_{1k} \right\| }{p}\right),\notag\\
 \frac{1}{p}\parallel D{\parallel }_{F}^{2} &=o_p(\frac{1}{p^{2\eta+1}q^{2\eta+1}}). \notag
\end{align}
 By the same calculation, we can show that \(\frac{1}{p}\parallel B{\parallel }_{F}^{2}\) can be bounded by \(\frac{1}{p}\parallel {A}_{1}{\parallel }_{F}^{2} + \frac{1}{p}\parallel {A}_{2}{\parallel }_{F}^{2} + \frac{1}{p}\parallel {A}_{3}{\parallel }_{F}^{2} + \frac{1}{p}\parallel D{\parallel }_{F}^{2}\) . Then we have
\begin{align}\label{Th3_3}
\frac{1}{p}{\begin{Vmatrix}\widetilde{R}_{k} - \widetilde{R}_{k}^{0}\widetilde{H}_{1k}\end{Vmatrix}}_{F}^{2} = O_{p}\left( {\frac{1}{qn} + \frac{1}{{p}^{2}{q}^{2}}}\right)\left [1 +  {\begin{Vmatrix}\widetilde{C}_{k} - \widetilde{C}_{k}^{0}\widetilde{H}_{2k}\end{Vmatrix}}_{F}^{2}\right].
\end{align}
By the same way, from the estimation equation of \(\widehat{C}_k\) , we have
\begin{align}\label{Th3_4}
\frac{1}{q}{\begin{Vmatrix}\widetilde{C}_{k} - \widetilde{C}_{k}^{0}\widetilde{H}_{2k}\end{Vmatrix}}_{F}^{2} = O_{p}\left( {\frac{1}{pn} + \frac{1}{{p}^{2}{q}^{2}}}\right) \left[1 +  {\begin{Vmatrix}\widetilde{R}_{k} - \widetilde{R}_{k}^{0}\widetilde{H}_{1k}\end{Vmatrix}}_{F}^{2}\right].
\end{align}
Combined the results of (\ref{Th3_3}) and (\ref{Th3_4}), we finally have
\begin{align}
\frac{1}{p}{\begin{Vmatrix}\widetilde{R}_{k} - \widetilde{R}_{k}^{0}\widetilde{H}_{1k}\end{Vmatrix}}_{F}^{2}&= {O}_{p}\left( {\frac{1}{{q}{n}} + \frac{1}{{p}^{2}{q}^{2}} + \frac{1}{p{n}^{2}}}\right)\triangleq O_p(w_1),\notag\\
\frac{1}{q}{\begin{Vmatrix}\widetilde{C}_{k} - \widetilde{C}_{k}^{0}\widetilde{H}_{2k}\end{Vmatrix}}_{F}^{2}& = {O}_{p}\left( {\frac{1}{pn} + \frac{1}{{p}^{2}{q}^{2}} + \frac{1}{q{n}^{2}}}\right)\triangleq O_p(w_2) . \notag
\end{align}
It remains to show that \(\widetilde{H}_{1k}^{\prime }\widetilde{H}_{1k} = {I}_{k_1} + {o_p }\left( 1\right)\) . Since
\begin{align}\label{Th3_5}
{\begin{Vmatrix}\frac{1}{p}{{\widetilde{R}}_{k}}^{\prime }(\widetilde{R}_k- \widetilde{R}_{k}^{0} \widetilde{H}_{1k}) \end{Vmatrix}}_{F}^{2} & \leq  \frac{\parallel \widetilde{R}_{k}{\parallel }_{F}^{2}}{p} \cdot  \frac{1}{p}{\begin{Vmatrix}\widetilde{R}_{k} - \widetilde{R}_{k}^{0}\widetilde{H}_{1k}\end{Vmatrix}}_{F}^{2}=  {O}_{p}\left( {w}_{1}\right),
\notag\\
{\begin{Vmatrix}\frac{1}{p}{{\widetilde{R}}_{k}}^{\prime }(\widetilde{R}_k- \widetilde{R}_{k}^{0} \widetilde{H}_{1k}) \end{Vmatrix}}_{F}^{2} &\leq  \frac{\parallel \widetilde{R}_{k}{\parallel }_{F}^{2}}{p} \cdot  \frac{1}{p}{\begin{Vmatrix}\widetilde{R}_{k} - \widetilde{R}_{k}^{0}\widetilde{H}_{1k}\end{Vmatrix}}_{F}^{2}   \triangleq  {O}_{p}\left( {w}_{1}\right),
\end{align}
then from the first inequality, it is straightforward to demonstrate that
\begin{align}\label{Th3_6}
&\frac{1}{p}\widetilde{R}_k^{\prime }\left( \widetilde{R}_k - \widetilde{R}_k^{0} \widetilde{H}_{1k}\right)  = {I}_{{k}_{1}} - \frac{\widetilde{R}_k^{\prime }  \widetilde{R}_k^{0}}{P} \widetilde{H}_{1k} = {O}_{p}\left( {w}_{1}^{\frac{1}{2}}\right),
\notag\\
\Rightarrow  & {I}_{{k}_{1}} = \frac{\widetilde{R}_k^{\prime }  \widetilde{R}_k^0} {p} \widetilde{H}_{1k} + {O}_{p}\left( {w}_{1}^{\frac{1}{2}}\right) .
\end{align}
From the second inequality of equation (\ref{Th3_5}), we have
\begin{align}\label{Th3_7}
&\frac{1}{p}{\widetilde{R}_k^{0\prime}}\left( \widetilde{R}_k - \widetilde{R}_k^{0} \widetilde{H}_{1k}\right)  = \frac{\widetilde{R}_k^{0\prime }  \widetilde{R}_k}{p} - \widetilde{H}_{1k} = {O}_{p}\left( {w}_{1}^{\frac{1}{2}}\right),
\notag \\
\Rightarrow&\frac{\widetilde{R}_k^{0\prime }  \widetilde{R}_k}{p} =\widetilde{H}_{1k} +{O}_{p}\left( {w}_{1}^{\frac{1}{2}}\right).
\end{align}
Taking (\ref{Th3_7}) into (\ref{Th3_6}), we have \({I}_{{k}_{1}} = \frac{\widetilde{R}_k^{\prime }  \widetilde{R}_k^{0}}{p}\frac{\widetilde{
R}_k^{0\prime }  {R_k}}{p}  +{O}_{p}\left( {w}_{1}^{\frac{1}{2}}\right) \) and \(\widetilde{H}_{1k}^\prime \widetilde{H}_{1k} = {I}_{k_1} + {O}_{p}\left( {w}_{1}^{\frac{1}{2}}\right)\) . Similarly, we have \({I}_{k_2} = \frac{\widetilde{C}_k^{\prime }  \widetilde{C}_k^{0}}{q}\frac{\widetilde{C}_k^{0\prime }  \widetilde{C}_k^{\prime}}{q}  +{O}_{p}\left( {w}_{1}^{\frac{1}{2}}\right) \) and \(\widetilde{H}_{2k}^\prime \widetilde{H}_{2k} = {I}_{k_2} + {O}_{p}\left( {w}_{1}^{\frac{1}{2}}\right)\) .
\end{proof}

\setcounter{equation}{0}
 \setcounter{subsection}{0}
 \renewcommand{\theequation}{E.\arabic{equation}}
 \renewcommand{\thesubsection}{E.\arabic{subsection}}
\section*{E ~~ Details for Theorem 4}

\begin{proof}
     Let \({A}_{1i},{A}_{2i},{A}_{3i},{D}_{i},{B}_{i}\) denote the $i$-th row of \({A}_{1},{A}_{2},{A}_{3}\) , \(D\) and \(B\) respectively.
From (\ref{Th3_1}), we have
\begin{align}\label{Th4_1}
    {\widetilde{\gamma}}_{k,i}^{\prime }. - \widetilde{\gamma}_{k,i}^{\circ\prime }.{\widetilde{H}}_{1k} = \left( {{A}_{1i} + {A}_{2i} + {A}_{3i} + {D}_{i}(\widehat{D}_k^{(1)})^{\frac{1}{2}} + {B}_{i}}(\widehat{D}_k^{(1)})^{\frac{1}{2}}\right) {\bar{\Delta }}_{k}^{ -1 }.
\end{align}
Next, we will consider the order of each term on the right hand side of (\ref{Th4_1}).

(1) Consider \({A}_{1i}\) \({\bar{\Delta }}_{k}^{ -1 }\) .
By Lemma \ref{lemaC3}, \({\left| \right| {\bar{\Delta }}_{k}^{ -1 }\left| \right|}_{F}^{2} =  {O_p}(1)\).
We only need to consider \(\left| \right| {A}_{1i}{\left| \right| }_{F}^{2}\) , where
\begin{align}\label{Th4_2}
{\begin{Vmatrix}{A}_{1i}\end{Vmatrix}}_{F}^{2} &= \frac{1}{{p}^{2}{q}^{4} {n}^{2}}{\begin{Vmatrix}\mathop{\sum }\limits_{t}{I}_{({s}_{t}=k)}{\gamma}^{\prime\circ }_{k,i}  {F}_{t}{C}_{k}^{\circ\prime}  {\widetilde{{C}}_{k}}{\widetilde{{C}}_{k}^{\prime}}{{E}_{t}^{\prime}}{\widetilde{{R}}_{k}}\end{Vmatrix}}_{F}^{2}
\notag\\
&\lesssim  \frac{\parallel {C}_{k}^{\circ\prime }{\widetilde{C}_{k}}{\parallel }_{F}^{2}}{{q}^{2}} \cdot  \frac{1}{{p}^{2}{q}^{2}{n}^{2}}{\begin{Vmatrix}\mathop{\sum }\limits_{t}I\left( {{s}_{t} = k}\right) {F}_{t}{\widetilde{C}_{k}}^{\prime}{E}_{t}^{\prime}{\widetilde{R}_{k}
} \end{Vmatrix}}_{F}^{2} {\begin{Vmatrix}{\gamma}_{k,i}^{\circ }.\end{Vmatrix}}_{F}^{2}.
\end{align}
Since
\begin{align}
&\frac{1}{{p}^{2}{q}^{2} {n}^{2}} {\begin{Vmatrix}\mathop{\sum }\limits_{t}{I}_{({s}_{t}=k)}{F}_{t}{\widetilde{C}_{k}^{\prime}}{{E}_{t}^{\prime}}{\widetilde{R}_{k}}\end{Vmatrix}}_{F}^{2} \notag\\
\lesssim&  \frac{1}{{p}^{2}{q}^{2}{n}^{2}}{\begin{Vmatrix}\mathop{\sum }\limits_{t}{I}_{({s}_{t}=k)}{F}_{t}\widetilde{H}^{\prime}_{2k} \widetilde{C}_{k}^{\circ\prime}{E}_{t}^{\prime}\widetilde{R}_{k}^{ \circ}\widetilde{H}_{1k}\end{Vmatrix}}_{F}^{2}
\notag\\
+&\frac{1}{{p}^{2}{q}^{2}{n}^{2}}{\begin{Vmatrix}\mathop{\sum }\limits_{t}{I}_{({s}_{t}=k)}{F}_{t}(\widetilde{C}_{k}-
\widetilde{C}_{k}^{\circ}\widetilde{H}_{2k})^{\prime}{E}_{t}^{\prime}\widetilde{R}_{k}^{ \circ}\widetilde{H}_{1k}\end{Vmatrix}}_{F}^{2}
\notag\\
+&\frac{1}{{p}^{2}{q}^{2}{n}^{2}}{\begin{Vmatrix}\mathop{\sum }\limits_{t}{I}_{({s}_{t}=k)}{F}_{t}\widetilde{H}^{\prime}_{2k} \widetilde{C}_{k}^{ \circ\prime}{E}_{t}^{\prime}\left( {\widetilde{R}_{k} - \widetilde{R}_{k}^{\circ}\widetilde{H}_{1k}}\right)\end{Vmatrix}}_{F}^{2}
\notag\\
+&\frac{1}{{p}^{2}{q}^{2}{n}^{2}}{\begin{Vmatrix}\mathop{\sum }\limits_{t}{I}_{({s}_{t}=k)}{F}_{t}(\widetilde{C}_{k}-
\widetilde{C}_{k}^{\circ}\widetilde{H}_{2k})^{\prime}{E}_{t}^{\prime}\left( {\widetilde{R}_{k} - \widetilde{R}_{k}^{\circ}\widetilde{H}_{1k}}\right)\end{Vmatrix}}_{F}^{2}
\notag\\
\triangleq & {I}_{11} + {I}_{12} + {I}_{13} + {I}_{14}.
\end{align}
We will determine the order of \({I}_{11},{I}_{12},{I}_{13}\) and \({I}_{14}\) , respectively. Under Assumption E(1).
$$
{I}_{11} \asymp \frac{1}{npq}
\begin{Vmatrix}\dfrac{1}{\sqrt{n}}\mathop{\sum }\limits_{t}{I}_{({s}_{t}=k)}{F}_{t}\dfrac{\widetilde{C}_{k}^{\circ\prime}}{\sqrt{q}}{E}_{t}^{\prime}\dfrac{\widetilde {R}_{k}^{ \circ}}{\sqrt{p}}\end{Vmatrix}_{F}^{2}.
{O}_{p}\left(1\right)={O}_{p}\left(\frac{1}{pqn}\right).
$$
By Assumption E(1) and Theorem \ref{Th3},
\begin{align}
{I}_{12} &\lesssim \frac{1}{q}{\begin{Vmatrix}\widetilde{C}_{k}-
\widetilde{C}_{k}^{\circ}\widetilde{H}_{2k}\end{Vmatrix}}_{F}^{2}
\frac{1}{{p}^{2}{q}{n}^{2}} {\begin{Vmatrix}\mathop{\sum }\limits_{t}{F}_{t}^{\prime}\widetilde{R}_{k}^{\circ\prime}{E}_{t}\end{Vmatrix}}_{F}^{2}
\notag\\
&= {O}_{p}\left( {w}_{2}\right)  \times  {O}_{p}\left( {\frac{1}{np}}\right)
= {O}_{p}\left( \frac{1}{{n}^{2}{p}^{2}}\right)  + {o}_{p}\left( \frac{1}{{n}{p}{q}}\right).
\end{align}
Similarly, we have
\(
{I}_{13} = {O}_{p}\left( {w}_{1}\right)  \times  {O}_{p}\left( \frac{1}{nq}\right)  = {O}_{p}\left( \frac{1}{{n}^{2}{q}^{2}}\right)  + {o}_{p}\left( \frac{1}{npq}\right)
\)
, and
\(
{I}_{14} = {O}_{p}\left( {{w}_{1} \times  {w}_{2}}\right)  \times  {O}_{p}\left( \frac{1}{n}\right)  = {O}_{p}\left( \frac{1}{{n}^{4}{q}^{2}} + \frac{1}{{n}^{4}{p}^{2}}\right)  + {o}_{p}\left( \frac{1}{npq}\right)
\) .
Combined the results of
\({I}_{11}-{I}_{14}\) and (\ref{Th4_2}), we have \(\left| \right| {A}_{1i}{\left| \right| }_{F}^{2} = {O}_{p}\left( \frac{1}{{n}^{2}{p}^{2}}\right) + {o}_{p}\left( \frac{1}{np}\right)\) .
Note that when we determine the order, we assume \({k}_{1} = {k}_{2} = 1\) . Actually the order results can be easily generalized to the general case just as the proof process of Lemma B4 in \cite{2024Quasi}
.

(2) Consider \({A}_{2i} {\bar{\Delta }}_{k}^{-1}\) .
\begin{align}
A_{2i} \bar{\Delta}_k^{-1}
&= \frac{1}{p q^2 n} \sum_t I_{(s_t=k)} e_{t,i.}' \widetilde{C}_k \widetilde{C}_k' C_k^\circ F_t' R_k^{\circ\prime} \widetilde{R}_k \bar{\Delta}_k^{-1} \notag \\
&= \frac{1}{q n} \sum_t I_{(s_t=k)} e_{t,i.}' \left( \widetilde{C}_k - \widetilde{C}_k^0 \widetilde{H}_{2k} + \widetilde{C}_k^0 \widetilde{H}_{2k} \right)
\frac{\widetilde{C}_k' \widetilde{C}_k^0}{q} \widetilde{F}_t' \frac{\widetilde{R}_k^{\circ\prime} \widetilde{R}_k}{p} \bar{\Delta}_k^{-1} \notag \\
&= \frac{1}{q n} \sum_{t=1}^n I_{(s_t=k)} e_{t,i.}' \left( \widetilde{C}_k - \widetilde{C}_k^0 \widetilde{H}_{2k} \right)
\frac{\widetilde{C}_k' C_k^\circ}{q} \widetilde{F}_t' \frac{\widetilde{R}_k^{\circ\prime} \widetilde{R}_k}{p} \bar{\Delta}_k^{-1} \notag \\
&\quad + \frac{1}{q n} \sum_{t=1}^n I_{(s_t=k)} e_{t,i.}' C_k^\circ \widetilde{H}_{2k}
\frac{\widetilde{C}_k' C_k^\circ}{q} \widetilde{F}_t' \frac{\widetilde{R}_k^{\circ\prime} \widetilde{R}_k}{p} \bar{\Delta}_k^{-1} \notag \\
&\triangleq I_{21} + I_{22},
\end{align}
where $\widetilde{F}_t=(D_k^{(1)})^{1/2}F_t(D_k^{(2)})^{1/2}.$ By Assumption E(1) and Theorem \ref{Th3}, we have
\begin{align}
    {\begin{Vmatrix}{I}_{{21}}\end{Vmatrix}}_{F}^{2} & \asymp  {\begin{Vmatrix}\dfrac{1}{qn}\sum\limits_{t}I_{\left({s}_{t}=k\right)}{e}_{t,i.}^{\prime}\left(\widetilde{C}_{k}-\widetilde{C}_{k}^{\circ}\widetilde{H}_{2k}\right){F}_{t}^{\prime}\end{Vmatrix}}_{F}^{2}
\notag\\
&\lesssim \frac{1}{q} {\begin{Vmatrix}\widetilde{C}_{k}-\widetilde{C}_{k}^{\circ}\widetilde{H}_{2k}\end{Vmatrix}}\frac{1}{qn^{2}}{\left\|\sum\limits_{t}I_{\left({s}_{t}=k\right)}{e}_{t,i.}F_{t}\right\|}_{F}^{2}\notag\\
&= {O}_{p}\left( {w}_{2}\right)   \frac{1}{{q}{n}^{2}}\mathop{\sum }\limits_{{j = 1}}^{q}{\left\|\sum_{t}{I}_{\left({{s}_{t}= k}\right)}{e}_{t,ij}{F}_{t}\right\|}
\notag\\
&= {O}_{p}\left( \frac{{w}_{2}}{n}\right)  = {O}_{p}\left( \frac{1}{{n}^{2}p}\right)  + {o}_{p}\left( \frac{1}{qn}\right).
\end{align}
Similar to equation (\ref{Th3_2}) in the proof of Theorem \ref{Th3},
We can get
\[
\frac{\widetilde{C}_{k}^{\circ\prime}\widetilde{C}_{k}}{q}{\widetilde{H}}_{2k}^{\prime} = \frac{\widetilde{C}_{k}^{\circ\prime}\widetilde{C}_{k}}{q}\frac{\widetilde{C}_{k}^{\prime}\widetilde{C}_{k}^{\circ}}{q} + {o}_{p}\left( 1\right)  = {I}_{{k}_{2}} + {o}_{p}\left( 1\right) .
\]
Then
$${I}_{22} = \frac{1}{{qn}}\mathop{\sum }\limits_{{t = 1}}^{n}{I}_{\left({s}_{t} = k\right) }{e}_{t,i\cdot}^{\prime } \widetilde{C}_{k}^{ \circ  }
\widetilde{F}_t^{\prime}\frac{\widetilde{R}_{k}^{\circ\prime}\widetilde{R}_{k}}{p}  {\bar{\Delta}}_{k}^{-1} + {o}_{p}\left( \frac{1}{\sqrt{qn}}\right),$$
where \({\begin{Vmatrix}\frac{1}{qn}\sum\limits_{t}{I}_{\left ({s}_{t}= k\right)}{e}_{t,i\cdot}^{\prime }\widetilde{C}_{k}^{ \circ  }\widetilde{F}_{t}\end{Vmatrix}}_{F}= {O}_{p}\left( \frac{1}{\sqrt{qn}}\right)\).

Combing the results of \({I}_{21}\) and \({I}_{22}\) , we have
$${A}_{2i}{{\bar{\Delta }}_{k}}^{-1}
= \frac{1}{qn}\mathop{\sum }\limits_{t}
{I}_{\left( {s}_{t}=k\right)}
{e}_{t,i.}^{\prime }\widetilde{C}^{\circ}_{k}\widetilde{F}_{t}^{\prime }\frac{\widetilde{R}_{k}^{\circ\prime}\widetilde{R}_{k}}{p}  {{\bar{\Delta}}_{k}}^{-1} + {O}_{p}\left( \frac{1}{\sqrt{{n^{2}p}}}\right) + {o}_{p}\left( \frac{1}{\sqrt{qn}}\right). $$

(3) Consider \({A}_{3i}\bar{\Delta }_{k}^{-1}\).

From results (3) of Lemma \ref{lemaC2}, and Theorem \ref{Th3}, it is easy to show that
$$\left| \right| {A}_{3i}\bar{\Delta }_{k}^{-1}{\left| \right| }^{2}_{F} = {O}_{p}\left( {\frac{1}{{p}^{2}{q}^{2}} + \frac{1}{{n}^{2}{p}^{2}} + \frac{1}{{n}^{3}p}}\right)  + {o}_{p}\left( \frac{1}{{qn}}\right).$$

(4)Consider the orders of the last two terms.

By Lemma \ref{lemaC1} (4) and Theorem \ref{Th2}, \(\left| \right| {D}_{i}{\bar{\Delta }}_{k}^{-1}{\left| \right| }_{F}^{2} \lesssim  \left| \right| D{\left| \right| }_{F}^{2}\).  \({O}_{p}\left( 1\right)  = {o}_{p}\left( \frac{1}{{p}^{2\eta}{q}^{2\eta + 1}}\right).\)
By some simple calculation,  \(\left|\right|{B}_{i}\bar{\Delta }_{k}^{-1}\left|\right|_{F}^{2}\)can be bounded by  \(\left|\right|{A}_{1i}\bar{\Delta }_{k}^{-1}\left|\right|_{F}^{2},\left|\right|{A}_{3i}\bar{\Delta }_{k}^{-1}\left|\right|_{F}^{2}\) and \(\left|\right|{D}_{i}\bar{\Delta }_{k}^{-1}\left|\right|.\)

Combing the above results (1) to (4), equation (\ref{Th4_1}) equals
\begin{align}\label{Th4_3}
  \widetilde{\gamma}_{k,i\cdot} - \widetilde{H}_{1k}^{\prime} \widetilde{\gamma}_{k,i\cdot}^{ \circ  }&
  = {\left( {\bar{\Delta }}_{k}\right) }^{-1}  \frac{{\widetilde{R}}_{k}^{\prime}\widetilde{R}_{k}^{\circ}}{p}  \frac{1}{qn}\mathop{\sum }\limits_{{t = 1}}^{n}{I}_{\left({s}_{t}=k\right)}(D_k^{(1)})^{1/2}{F}_{t}{C}_{k}^{\circ\prime}{e}_{t,i\cdot}\notag\\
  &
+{O}_{p}\left( {\frac{1}{{pq}} + \frac{1}{{\sqrt{{n}^{2}p}}}}\right)  + {o}_{p}\left( \frac{1}{\sqrt{qn}}\right).
\end{align}
From equation (\ref{lemac3_4}) and Assumption A(1), we have
\begin{align}
    \bar{\Delta}_k= \pi_k^0\frac{\widetilde{R}_{k}^{\prime}\widetilde{R}_{k}^{\circ}}{p}\Sigma_{F_k}^{(1)} \frac{\widetilde{R}_{k}^{\circ\prime}\widetilde{R}_{k}}{p}+o_p(1),
\end{align}
which means $(\bar{\Delta}_k)^{-1} \frac{\widetilde{R}_{k}^{\prime}\widetilde{R}_{k}^{\circ}}{p}= \frac{1}{\pi_k^{\circ}}(\frac{\widetilde{R}_{k}^{\circ\prime}\widetilde{R}_{k}}{p})^{-1}(\Sigma_{F_k}^{(1)})^{-1} +o_p(1)= \frac{1}{\pi_k^{\circ}}\widetilde{H}_{1k}^\prime(\Sigma_{F_k}^{(1)})^{-1} +o_p(1)$,
where the last equation holds due to that
\(\frac{{\widetilde{R}}_{k}^{\circ\prime}{{R}}_{k}}{p} = \widetilde{H}_{1k} + o_{p} {\left( 1\right) }\) by equation (\ref{Th3_7}). Then left multiplying both side of equation (\ref{Th4_3}) by $\widetilde{H}_{1k}$, we have
\begin{align}
  \widetilde{H}_{1k}(\widetilde{\gamma}_{k,i\cdot} - \widetilde{H}_{1k}^{\prime} \widetilde{\gamma}_{k,i\cdot}^{ \circ  })&
  =\frac{1}{\pi_k^0} {\left( \Sigma_{F_k}^{(1)}\right) }^{-1}   \frac{1}{qn}\mathop{\sum }\limits_{{t = 1}}^{n}{I}_{\left({s}_{t}=k\right)}(D_k^{(1)})^{1/2}{F}_{t}{C}_{k}^{\circ\prime}{e}_{t,i\cdot}\notag\\
  &
+{O}_{p}\left( {\frac{1}{{pq}} + \frac{1}{{\sqrt{{n}^{2}p}}}}\right)  + {o}_{p}\left( \frac{1}{\sqrt{qn}}\right).
\end{align}
Finally, when \(qn = {o}_{p}\left( {{p}^{2}{q}^{2},{n}^{2}{p}}\right)\) , By Assumption F(1), we have
\[\sqrt{qn}\widetilde{H}_{1k}\left({ \widetilde{\gamma}_{k,i}. - \widetilde{H}_{1k}^{\prime}{\gamma}^{\circ}_{k,i}.}\right) \overset{L}{ \rightarrow  }{N}_{{k}_{1}}\left( {0},\left( {\pi }_{k}^{\circ}\right)^{-2}\left({\Sigma}_{F_k}^{\left( 1\right) }\right)^{-1}{V}_{1i}\left({\Sigma}_{F_k}^{\left( 1\right) }\right)^{-1}\right) .\]
Similarity, when \({pn} = {o}_{p}\left( {{p}^{2}{q}^{2},{n}^{2}q}\right)\) , by Assumption F(2), we have
\[\sqrt{pn}\widetilde{H}_{2k}\left( {{\widetilde{C}}_{k,j}, - {\widetilde{H}}_{2k}^{\prime }{C}^{\circ}_{k,j}.}\right)  \overset{L}{\rightarrow } {N}_{{k}_{2}}\left( {0},\left( {\pi }_{k}^{\circ}\right)^{-2}\left({\Sigma}_{F_k}^{\left( 2\right) }\right)^{-1}{V}_{2j}\left({\Sigma}_{F_k}^{\left( 2\right) }\right)^{-1}\right) . \]

\end{proof}

\setcounter{equation}{0}
 \setcounter{subsection}{0}
 \renewcommand{\theequation}{F.\arabic{equation}}
 \renewcommand{\thesubsection}{F.\arabic{subsection}}
\section*{F ~~ Details for Theorem 5}
\begin{proof}
From the estimation equation of \( \rho  = \text{vec}\left( P\right) \) ,i.e,
\[\widehat{\rho } = \widehat{\delta}_{\rho} \oslash [1_M\otimes(1_M'\otimes{I_M})]\widehat{\delta}_p,\]
where \( {\widehat{\delta}}_{p} = \mathop{\sum }\limits_{{t = 2}}^{n}\text{Vec}\left( {\widehat\Delta }_{t}^{\rho}\right) \) with \( {\widehat{\Delta }}_{t}^{\rho} = ( {\widehat{w}}_{t - 1,t \mid  n}^{\left( i,j\right) })_{i = 1,j = 1}^{M, M} \), we have
\begin{align}\label{Th5_1}
   \widehat{p}_{ij} = \frac{\mathop{\sum }\limits_{{t = 2}}^{n}{\widehat{w}}_{t - 1,t \mid  n}^{\left( i,j\right) }}{\mathop{\sum }\limits_{{t = 2}}^{n}\mathop{\sum }\limits_{{k = 1}}^{M}{\widehat{w}}_{t - 1,t \mid  n}^{\left( i,k\right) }} = \frac{\frac{1}{n-1} \mathop{\sum }\limits_{{t=2}}^{n}{\widehat{w}}_{t - 1,t \mid  n}^{\left( i,j\right) }}{\frac{1}{n-1}\mathop{\sum }\limits_{{t = 2}}^{n}{\widehat{w}}_{t - 1\mid n}^{\left( i\right) }} = \frac{\frac{1}{n - 1}\mathop{\sum }\limits_{{t = 2}}^{n}{\widehat{w}}_{t - 1,t \mid  n}^{\left( i,j\right) }}{\frac{1}{n - 1}  \mathop{\sum }\limits_{{t = 1}}^{n-1}{\widehat{w}}_{t \mid  n}^{\left( i\right) }}.
\end{align}
For the denominator, by Theorem \ref{Th2}, we have
\begin{align}
 \frac{1}{n-1}\mathop{\sum }\limits_{{t = 1}}^{{n - 1}}{\widehat{w}}_{t \mid n}^{\left( i\right) } = \frac{1}{n - 1}\mathop{\sum }\limits_{{t = 1}}^{{n - 1}}{I}_{\left( {s}_{t} = i\right) } + {o}_{p}\left( \frac{1}{{N}^{\eta}}\right) \xrightarrow[]{p}{\pi }_{i}^{o}.
\end{align}
For the numerator, we have
\begin{align}\label{Th5_3}
    \frac{1}{n - 1}\mathop{\sum }\limits_{{t = 2}}^{n}{\widehat{w}}_{t - 1,t \mid n}^{\left( i,j\right) }& = \frac{1}{n - 1}\mathop{\sum }\limits_{{t = 2}}^{n}{P}_{r}\left( {{s}_{t - 1} = i,{s}_{t} = j \mid  \mathcal{Y}_{n};\widehat{\theta}}\right) .\notag\\
   & = \frac{1}{n - 1}\mathop{\sum }\limits_{{t = 2}}^{n}{P}_{r}\left( {{s}_{t - 1} = i,\mid {s}_{t} = j,\mathcal{Y}_{n};\widehat{\theta}}\right) {\widehat{w}}_{t \mid n}^{\left( j\right) }
    \notag\\
   & = \frac{1}{n - 1}\mathop{\sum }\limits_{{t = 2}}^{n}{P}_{r}\left( {{s}_{t - 1} = i \mid  {s}_{t} = j,\mathcal{Y}_{n};\widehat{\theta}}\right) \left\lbrack  {I_{\left( {{s}_{t} = j}\right)}  + {o}_{p}\left( \frac{1}{{N}^{\eta}}\right) }\right\rbrack .
\end{align}
The last equality of (\ref{Th5_3})  follows from Theorem \ref{Th2}. Further,
\begin{align}\label{Th5_4}
{P}_{r}\left( {{s}_{t - 1} = i \mid  {s}_{t} = j,\mathcal{Y}_{n};\widehat{\theta}}\right) &
= {P}_{r}\left( {{s}_{t - 1} = i \mid  {s}_{t} = j,\mathcal{Y}_{t - 1};\widehat{\theta}}\right)
\notag\\
&= \frac{{P}_{r}\left( {{s}_{t - 1} = i,{s}_{t} = j \mid  \mathcal{Y}_{t - 1};\widehat{\theta}}\right) }{{P}_{r}\left( {{s}_{t} = j \mid  \mathcal{Y}_{t - 1};\widehat{\theta}}\right) } = \frac{{\widehat{p}}_{ij} \cdot  {\widehat{w}}_{{t - 1} \mid {t-1}}^{\left( i\right) }}{\mathop{\sum }\limits_{{h = 1}}^{M}{\widehat{p}}_{hj} \cdot  {\widehat{w  }}_{{t - 1} \mid {t-1}}^{\left( h\right) }}\notag\\
&= I_{\left( {{s}_{t - 1} = i}\right)}  + {o}_{p}\left( \frac{1}{{N}^{\eta}}\right) ,
\end{align}
where the last equality follows from Theorem \ref{Th2}. Since $s_t$ follows a Markov process,
\( \frac{1}{n - 1}\mathop{\sum }\limits_{{t = 1}}^{n}I_{\left( {{s}_{t} = j}\right)} I_{\left( {{s}_{t - 1} = i}\right)} \overset{p}{ \rightarrow  }E\left\lbrack  {I_{\left( {{s}_{t} = j}\right)} I_{\left( {{s}_{t - 1} = i}\right)} }\right\rbrack   = E\left\lbrack  {E\left\lbrack  {I_{\left( {{s}_{t} = j}\right)} I_{\left( {{s}_{t - 1} = i}\right)} |I_{\left( {{s}_{t - 1} = i}\right)} }\right\rbrack  }\right\rbrack   = {p}_{ij}^{0}{\pi}_{i}^{0}. \)

Taking equations (\ref{Th5_1})-(\ref{Th5_4}) together, we have shown \( \widehat{p}_{ij}\xrightarrow[]{{p}}{p}_{ij}^0 \) .
\end{proof}

\setcounter{equation}{0}
 \setcounter{subsection}{0}
 \renewcommand{\theequation}{G.\arabic{equation}}
 \renewcommand{\thesubsection}{G.\arabic{subsection}}
\section*{G ~~ Details for Theorem 6}
\begin{proof}

 Firstly, consider the asymptotic representation of \(\widehat{w}_{t\mid n}^{(k)}\widehat{F}_{t\mid n}^{(k)}\) . From Lemma \ref{lemaC1}, we have
\begin{align}\label{G1}
\widehat{F}_{t\mid n}^{(k)}=&\frac{1}{pq}(\widehat{D}_k^{(1)})^{-1}\widehat{R}_k^\prime Y_t \widehat{C}_k (\widehat{D}_k^{(2)})^{-1} + O_p(p^{-1}q^{-1}) \notag \\
=&\sum_{l=1}^M I_{(s_t=l)} \frac{1}{pq}(\widehat{D}_k^{(1)})^{-1}\widehat{R}_k^\prime R_l^0 F_t {C_l^0}^\prime \widehat{C}_k(\widehat{D}_k^{(2)})^{-1}  \notag \\
&+ \frac{1}{pq}(\widehat{D}_k^{(1)})^{-1}\widehat{R}_k^\prime E_t \widehat{C}_k(\widehat{D}_k^{(2)})^{-\frac{1}{2}} + O_p(p^{-1}q^{-1}) \\
=&\sum_{l=1}^M I_{(s_t=l)} \frac{1}{pq}(\widehat{D}_k^{(1)})^{-1}\widehat{R}_k^\prime R_l^0 F_t {C_l^0}^\prime \widehat{C}_k(\widehat{D}_k^{(2)})^{-1}+ O_p(p^{-\frac{1}{2}}q^{-\frac{1}{2}}) + o_p(n^{-\frac{1}{2}}),\notag
\end{align}
where the last equation holds due to the order results given in Theorem \ref{Th3}. Then, combining the result of \(\widehat{w}_{t\mid n}^{(k)}\) in Theorem \ref{Th2}, we have
\begin{align}\label{G2}
    \widehat{w}_{t\mid n}^{(k)}\widehat{F}_{t\mid n}^{(k)}=&I_{(s_t=k)}\widehat{F}_{t\mid n}^{(k)} + o_p(\frac{1}{N^{\eta}}) \\
    =&I_{(s_t=k)}(\widehat{D}_k^{(1)})^{-1}\frac{\widehat{R}_k^\prime R_k^0}{p} F_t \frac{{C_k^{0\prime}} \widehat{C}_k}{q} (\widehat{D}_k^{(2)})^{-1} + O_p(p^{-\frac{1}{2}}q^{-\frac{1}{2}}) + o_p(n^{-\frac{1}{2}}). \notag
\end{align}
From equation (\ref{Th3_5}), that is, \(\frac{\widetilde{R}_k^{\prime} \widetilde{R}_k^0}{p} \frac{\widetilde{R}_k^{0\prime} \widetilde{R}_k}{p} = I_{k_1} + O_p(w_1^{\frac{1}{2}})\), we can get \(\widehat{D}_k^{(1)} = D_k^{(1)} + o_p(n^{-\frac{1}{2}})\) under the assumption \(\frac{\widehat{R}_k^{\prime} R_k^0}{p} = D_k^{(1)} + o_p(n^{-\frac{1}{2}})\). It further implies that \((\widehat{D}_k^{(1)})^{-1} \frac{\widehat{R}_k^\prime R_k^0}{p} = I_{k_1} + o_p(n^{-\frac{1}{2}})\). Similarly, \(\frac{C_k^{0\prime} \widehat{C}_k}{q} (\widehat{D}_k^{(2)})^{-1} = I_{k_2} + o_p(n^{-\frac{1}{2}})\), which means equation (\ref{G2}) can be reduced to
\begin{align}\label{G3}
     \widehat{w}_{t\mid n}^{(k)}\widehat{F}_{t\mid n}^{(k)}=I_{(s_t=k)}F_t + O_p(p^{-\frac{1}{2}}q^{-\frac{1}{2}}) + o_p(n^{-\frac{1}{2}}).
\end{align}

Next, consider the asymptotic representation of \(\widehat{F}_{t-1\mid n}^{*(k)}\). By equation (\ref{Th5_4}) in the proof of Theorem \ref{Th5}, we have
\begin{align}\label{G4}
    \widehat{F}_{t-1\mid n}^{*(k)}=& E[F_{t-1}\mid \mathcal{Y}_n,s_t=k;\widehat{\theta}] \notag \\
    =&\sum_{i=1}^M E[F_{t-1}\mid \mathcal{Y}_n,s_{t-1}=i,s_t=k;\widehat{\theta}]P_r[s_{t-1}=i\mid \mathcal{Y}_n,s_t=k;\widehat{\theta}]\notag\\
    =&\sum_{i=1}^M\widehat{F}_{t-1\mid n}^{(i)}[I_{(s_{t-1}=i)}+o_p(\frac{1}{N^{\eta}})]\notag\\
    =&\sum_{i=1}^M I_{(s_{t-1}=i)}F_{t-1} + O_p(p^{-\frac{1}{2}}q^{-\frac{1}{2}})+ o_p(n^{-\frac{1}{2}})\notag\\
    =&F_{t-1} + O_p(p^{-\frac{1}{2}}q^{-\frac{1}{2}}) + o_p(n^{-\frac{1}{2}}).
\end{align}
Finally, from the estimation equation of \(\widehat{B}_k\) given in the main text, we have
\begin{align*}
    \widehat{B}_k=&\frac{1}{\frac{1}{n}\sum_{t=1}^n \widehat{w}_{t\mid n}^{(k)}} \frac{1}{n} \sum_{t=1}^n I_{(s_t=k)} [F_t-\widehat{\Phi}_k F_{t-1} \widehat{\Gamma}_k^\prime] + O_p(p^{-\frac{1}{2}}q^{-\frac{1}{2}}) + o_p(n^{-\frac{1}{2}})\\
    =&\frac{1}{\pi_k^0}\left[ \frac{1}{n} \sum_{t=1}^n I_{(s_t=k)} F_t - \widehat{\Phi}_k(\frac{1}{n} \sum_{t=1}^n I_{(s_t=k)} F_{t-1})\widehat{\Gamma}_k^\prime \right] +  O_p(p^{-\frac{1}{2}}q^{-\frac{1}{2}}) + o_p(n^{-\frac{1}{2}}).
\end{align*}
Since \(\frac{1}{n} \sum_{t=1}^n I_{(s_t=k)} F_t=\frac{1}{n} \sum_{t=1}^n I_{(s_t=k)}(\Phi_k^0F_{t-1}\Gamma_k^{0\prime} + \epsilon_t)\), which means
\begin{align*}
     \frac{1}{n} \sum_{t=1}^n I_{(s_t=k)} F_{t-1} = \Phi_k^{0-1}[ \frac{1}{n} \sum_{t=1}^n I_{(s_t=k)} F_t -  \frac{1}{n} \sum_{t=1}^n I_{(s_t=k)} \epsilon_t](\Gamma_k^{0\prime})^{-1} = O_p(n^{-\frac{1}{2}}).
\end{align*}
Then we finish the proof of $\widehat{B}_k$.

 Consider the estimation equation of \(\widehat{\Phi}_k\) given in the main text, that is
\begin{align}\label{G5}
    \widehat{\Phi}_k =& \left[\sum_{t=1}^n \widehat{w}_{t\mid n}^{(k)}(P_{t,t-1\mid n}^{(2k)} - \widehat{B}_k \widehat{\Gamma}_k F_{t-1\mid n}^{*(k)\prime})\right]\left(\sum_{t=1}^n \widehat{w}_{t\mid n}^{(k)}P_{t-1\mid n}^{*(2k)}\right)^{-1}.
\end{align}
Firstly, consider \(\sum_t \widehat{w}_{t\mid n}^{(k)} P_{t-1\mid n}^{*(2k)}\), where
$$P_{t-1\mid n}^{*(2k)} = \mathop{\sum }\limits_{d=1}^{k_2}(e_{k_2}^{(d)\prime} \widehat{\Gamma}_k \otimes I_{k_1}) \widehat{P}_{t-1\mid n}^{*(k)} (\Gamma_k^\prime e_{k_2}^{(d)}\otimes I_{k_1}) $$ with
\begin{align*}
    \widehat{P}_{t-1\mid n}^{*(k)} =& \sum_{i=1}^M (\widehat{V}_{t-1\mid n}^{(i)} + \widehat{f}_{t-1\mid n}^{(i)}\widehat{f}_{t-1\mid n}^{(i)\prime})\widehat{w}_{t-1,t\mid n}^{(i,k)}/\widehat{w}_{t\mid n}^{(k)}.\tag{By Section A.5}
\end{align*}
By equations (\ref{lemac1_2}) - (\ref{lemac1_21}) and the definition of \(\widehat{V}_{t\mid n}^{(i)}\) given in the main text, we can show that \(\widehat{V}_{t-1\mid n}^{(i)} = O_p(p^{-1}q^{-1})\). Following the proof process of equation (\ref{Th5_4}), we can further have \(w_{t-1,t\mid n}^{(i,k)} = I_{(s_{t-1}=i,s_t=k)} + o_p(\frac{1}{N^{\eta}})\). Then we have \(\widehat{P}_{t-1\mid n}^{*(k)} = \mathop{\sum}\limits_{i=1}^M I_{(s_{t-1}=i)}\widehat{f}_{t-1\mid n}^{(i)}\widehat{f}_{t-1\mid n}^{(i)\prime} + o_p(\frac{1}{N^\eta})\), which further implies that
\begin{align}\label{G6}
    P_{t-1\mid n}^{*(2k)} =& \sum _{d=1}^{k_2}\sum_{i=1}^M I_{(s_{t-1}=i)}(e_{k_2}^{(d)\prime} \widehat{\Gamma}_k \otimes I_{k_1}) \widehat{f}_{t-1\mid n}^{(i)}\widehat{f}_{t-1\mid n}^{(i)\prime}(\widehat{\Gamma}_k^\prime e_{k_2}^{(d)} \otimes I_{k_1}) + o_p(\frac{1}{N^\eta}) \notag \\
    =&\sum _{d=1}^{k_2}\sum_{i=1}^M I_{(s_{t-1}=i)}\text{Vec}(\widehat{F}_{t-1\mid n}^{(i)} \widehat{\Gamma}_k^\prime e_{k_2}^{(d)})\text{Vec}^\prime (\widehat{F}_{t-1\mid n}^{(i)} \widehat{\Gamma}_k^\prime e_{k_2}^{(d)})+o_p(\frac{1}{N^\eta})\notag\\
    =&\sum _{d=1}^{k_2}\sum_{i=1}^M I_{(s_{t-1}=i)}\widehat{F}_{t-1\mid n}^{(i)} \widehat{\Gamma}_k^\prime e_{k_2}^{(d)}e_{k_2}^{(d)\prime} \widehat{\Gamma}_k \widehat{F}_{t-1\mid n}^{(i)\prime} + o_p(\frac{1}{N^\eta})\notag\\
    =&\sum_{i=1}^M I_{(s_{t-1}=i)}\widehat{F}_{t-1\mid n}^{(i)} \widehat{\Gamma}_k^\prime\widehat{\Gamma}_k \widehat{F}_{t-1\mid n}^{(i)\prime} + o_p(\frac{1}{N^\eta})\notag\\
    =& \sum_{i=1}^M I_{(s_{t-1}=i)} F_{t-1}\widehat{\Gamma}_k^\prime\widehat{\Gamma}_k F_{t-1}^\prime + O_p(p^{-\frac{1}{2}}q^{-\frac{1}{2}}) + o_p(n^{-\frac{1}{2}}) \tag{By equation (\ref{G3})}\notag \\
    =& F_{t-1}\widehat{\Gamma}_k^\prime\widehat{\Gamma}_k F_{t-1}^\prime + O_p(p^{-\frac{1}{2}}q^{-\frac{1}{2}}) + o_p(n^{-\frac{1}{2}}).
\end{align}
Then we have
\begin{align}\label{G7}
    \frac{1}{n}\sum_{t=1}^n \widehat{w}_{t\mid n}^{(k)} P_{t-1\mid n}^{*(2k)} = \frac{1}{n}\sum_{t=1}^n I_{(s_t=k)} F_{t-1}\widehat{\Gamma}_k^\prime\widehat{\Gamma}_k F_{t-1}^\prime +  O_p(p^{-\frac{1}{2}}q^{-\frac{1}{2}}) + o_p(n^{-\frac{1}{2}}).
\end{align}
Next, consider \(\mathop{\sum}\limits_{t=1}^n \widehat{w}_{t\mid n}^{(k)} P_{t,t-1\mid n}^{(2k)}\), where \(P_{t,t-1\mid n}^{(2k)} = \mathop{\sum}\limits_{d=1}^{k_2}(e_{k_2}^{(d)\prime}\otimes I_{k_1}) \widehat{P}_{t,t-1\mid n}^{(k)} (\widehat{\Gamma}_k^\prime e_{k_2}^{(d)} \otimes I_{k_1})\). Since
\begin{align*}
    \widehat{P}_{t,t-1\mid n}^{(k)}\triangleq &E[f_t f_{t-1} \mid \mathcal{Y}_n,s_t=k;\widehat{\theta}]\\
    =&\text{Cov}[f_t ,f_{t-1} \mid \mathcal{Y}_n,s_t=k;\widehat{\theta}]+\widehat{f}_{t\mid n}^{(k)} \widehat{f}_{t-1\mid n}^{*(k)\prime}\\
    =&\text{Cov}[\beta_k^0 + (\Gamma_k^0 \otimes \Phi_k^0)f_{t-1} + \epsilon_t,f_{t-1}\mid \mathcal{Y}_n, s_t=k;\widehat{\theta}]+\widehat{f}_{t\mid n}^{(k)} \widehat{f}_{t-1\mid n}^{*(k)\prime}\\
    =& (\Gamma_k^0 \otimes \Phi_k^0)[\widehat{P}_{t-1\mid n}^{*(k)} - \widehat{f}_{t-1\mid n}^{*(k)} \widehat{f}_{t-1\mid n}^{*(k)\prime}] + \widehat{f}_{t\mid n}^{(k)} \widehat{f}_{t-1\mid n}^{*(k)\prime},
\end{align*}
and \(\widehat{f}_{t-1\mid n}^{*(k)} = \mathop{\sum}\limits_{i=1}^M I_{(s_{t-1}=i}) \widehat{f}_{t-1\mid n}^{(i)} + o_p(\frac{1}{N^\eta})\)(See equation (\ref{G4})), and the asymptotic representation of \(\widehat{P}_{t-1}^{*(k)}\)given above , we have
\begin{align*}
    \widehat{P}_{t,t-1\mid n}^{(k)} = \widehat{f}_{t\mid n}^{(k)} \widehat{f}_{t-1\mid n}^{*(k)\prime} + o_p(\frac{1}{N^\eta}),
\end{align*}
which further implies that
\begin{align*}
     \widehat{P}_{t,t-1\mid n}^{(2k)} =& \sum_{d=1}^{k_2}(e_{k_2}^{(d)\prime} \otimes I_{k_1}) \widehat{P}_{t,t-1\mid n}^{(k)} (\widehat{\Gamma}_k^\prime e_{k_2}^{(d)} \otimes I_{k_1})\\
     =&\sum_{d=1}^{k_2}(e_{k_2}^{(d)\prime} \otimes I_{k_1}) \widehat{f}_{t\mid n}^{(k)} \widehat{f}_{t-1\mid n}^{*(k)\prime}(\widehat{\Gamma}_k^\prime e_{k_2}^{(d)} \otimes I_{k_1}) + o_p(\frac{1}{N^\eta})\\
     =& \sum_{d=1}^{k_2} \text{Vec}(\widehat{F}_{t\mid n}^{(k)}e_{k_2}^{(d)}) \text{Vec}^\prime (\widehat{F}_{t-1\mid n}^{*(k)}\widehat{\Gamma}_k^\prime e_{k_2}^{(d)}) + o_p(\frac{1}{N^\eta})\\
     =&\sum_{d=1}^{k_2} \widehat{F}_{t\mid n}^{(k)}e_{k_2}^{(d)} e_{k_2}^{(d)\prime} \widehat{\Gamma}_k \widehat{F}_{t-1\mid n}^{*(k)\prime} + o_p(\frac{1}{N^\eta})\\
     =& \widehat{F}_{t\mid n}^{(k)} \widehat{\Gamma}_k \widehat{F}_{t-1\mid n}^{*(k)\prime}+ o_p(\frac{1}{N^\eta}).
\end{align*}
 By equations (\ref{G3}) and (\ref{G4}),
\begin{align}\label{G8}
     \frac{1}{n}\sum_{t=1}^n \widehat{w}_{t\mid n}^{(k)} P_{t-1\mid n}^{(2k)} =& \frac{1}{n} \sum_t\widehat{w}_{t\mid n}^{(k)} \widehat{F}_{t\mid n}^{(k)} \widehat{\Gamma}_k \widehat{F}_{t-1\mid n}^{*(k)\prime} + o_p(\frac{1}{N^\eta}) \notag \\
     =&\frac{1}{n} \sum_t I_{(s_t=k)}F_t \widehat{\Gamma}_k F_{t-1}^\prime + O_p(p^{-\frac{1}{2}}q^{-\frac{1}{2}}) + o_p(n^{-\frac{1}{2}}) \notag\\
     =&\frac{1}{n} \sum_t I_{(s_t=k)} \Phi_k^0 F_{t-1} \Gamma_k^{0\prime} \widehat{\Gamma}_k F_{t-1}^\prime + \frac{1}{n} \sum_t I_{(s_t=k)} \epsilon_t \widehat{\Gamma}_k F_{t-1}^\prime \notag\\
     &+ O_p(p^{-\frac{1}{2}}q^{-\frac{1}{2}}) + o_p(n^{-\frac{1}{2}}).
\end{align}
Combined equations (\ref{G5}),  (\ref{G7}), (\ref{G8}) and results of \(\widehat{B}_k\), we finally have
\begin{align*}
    \|\widehat{\Phi}_k - \Phi_k^0 \widehat{P}_{k_2} \|_F^2 = O_p(p^{-1}q^{-1}) + O_p(\frac{1}{n}).
\end{align*}
Similarly, we can get the result of \(\widehat{\Gamma}_k\).

\end{proof}

\end{document}